%% file: ms.tex
\documentclass[11pt,a4paper]{article}

\usepackage[margin=1in]{geometry}
\usepackage[utf8]{inputenc}
\usepackage[T1]{fontenc}
\usepackage{lmodern}
\usepackage[tbtags]{mathtools}
\allowdisplaybreaks
\usepackage{amssymb,amsthm}
\usepackage{cases}
\usepackage{hyperref}
\usepackage{url}
\usepackage[svgnames]{xcolor}
\usepackage[capitalise,nameinlink]{cleveref}
\hypersetup{colorlinks={true},linkcolor={DarkBlue},citecolor=[named]{DarkGreen}}
\usepackage[font=footnotesize]{caption}
\usepackage[numbers,sort&compress]{natbib}
\usepackage{doi}
\usepackage{subcaption}
\usepackage{tikz}
\usetikzlibrary{math,patterns}
\usepackage{pgfplots}
\usepgfplotslibrary{fillbetween}
\pgfplotsset{compat=1.14}

\newtheorem{theorem}{Theorem}
\newtheorem{corollary}{Corollary}
\newtheorem{lemma}{Lemma}
\newtheorem{proposition}{Proposition}
\newtheorem*{lemmanonum}{Lemma}
\theoremstyle{definition}
\newtheorem{definition}{Definition}
\newtheorem{example}{Example}
\newtheorem*{remarknonum}{Remark}

\def \R{\mathbb{R}}
\newcommand{\abb}{\mathbb{A}}
\newcommand{\ebb}{\mathbb{E}}
\newcommand{\fbb}{\mathbb{F}}
\newcommand{\fcal}{\mathcal{F}}
\newcommand{\gcal}{\mathcal{G}}
\newcommand{\lcal}{\mathcal{L}}

\newcommand{\ssets}[1]{\{ #1\}}
\newcommand{\union}{\cup}
\newcommand{\map}{\longrightarrow}
\newcommand{\ifif}{\Longleftrightarrow} 
\newcommand{\then}{\Longrightarrow} 
\newcommand{\bigland}{\bigwedge}

\DeclareMathOperator*{\expectation}{\mathbb E}
\newcommand{\expect}[2][]{\expectation_{#1}\nolimits\left[#2\right]}
\newcommand{\expectsmall}[2][]{\expectation_{#1}\nolimits[#2]}

\DeclareMathOperator{\probability}{\mathrm{Pr}}
\newcommand{\prob}[1]{\probability\left[#1\right]}

\DeclareMathOperator*{\variance}{\mathrm{Var}}
\newcommand{\var}[2][]{\variance_{#1}\nolimits\left[#2\right]}

\newcommand{\opt}{\ensuremath{\mathrm{OPT}}}
\newcommand{\val}{\ensuremath{\mathrm{VAL}}}
\newcommand{\rev}{\ensuremath{\mathrm{REV}}}
\newcommand{\wel}{\ensuremath{\mathrm{WEL}}}
\newcommand{\apx}{\ensuremath{\mathrm{APX}}}
\newcommand{\dapx}{\ensuremath{\mathrm{DAPX}}}

\newcommand{\plog}{P^{\mathrm{log}}}

\title{Robust Revenue Maximization\\ Under Minimal Statistical
		Information\thanks{Supported by the Alexander von Humboldt Foundation with
		funds from the German Federal Ministry of Education and Research (BMBF). The
		second author further acknowledges the support of FCT via LASIGE Research Unit, ref. UIDB/00408/2020 and ref. UIDP/00408/2020. The
		last author further acknowledges the support of the German Research
		Foundation (DFG) within the Research Training Group AdONE (GRK 2201).
		Part of this work was done while the first two authors were members of the Operations Research group at the Technical University of Munich.
		\newline 
		\indent An extended abstract of this paper appeared in WINE'20~\cite{Giannakopoulos2019}.
		}}

\author{
		Yiannis Giannakopoulos\thanks{Friedrich-Alexander-Universität Erlangen-Nürnberg, Department of Data Science. 
		Email:
		{\tt
		\href{mailto:yiannis.giannakopoulos@fau.de}{\nolinkurl{yiannis.giannakopoulos@fau.de}}
		}}
	\and
		Diogo Poças\thanks{LASIGE, Faculdade de Ciências, Universidade de Lisboa.
		Email: {\tt \href{mailto:dmpocas@fc.ul.pt}{\nolinkurl{dmpocas@fc.ul.pt}}}}
	\and 
		Alexandros Tsigonias-Dimitriadis\thanks{Technical University of Munich. Email: {\tt
		\href{mailto:alexandros.tsigonias@tum.de}{\nolinkurl{alexandros.tsigonias@tum.de}}}. Associated researcher with the Research Training Group ``Advanced
		Optimization in a Networked Economy'' (GRK 2201), funded by the German
		Research Foundation (DFG).} }

\date{April 28, 2022}

\begin{document}

\maketitle
\setcounter{page}{0}
\thispagestyle{empty}

\begin{abstract} 
We study the problem of multi-dimensional revenue maximization when selling $m$
items to a buyer that has additive valuations for them, drawn from a (possibly
correlated) prior distribution.
Unlike traditional Bayesian auction design, we assume that the seller has a very
restricted knowledge of this prior: they only know the mean $\mu_j$ and an upper
bound $\sigma_j$ on the standard deviation of each item's marginal distribution.
Our goal is to design mechanisms that achieve good revenue against an ideal optimal
auction that has \emph{full} knowledge of the distribution in advance.
Informally, our main contribution is a tight quantification of the interplay between
the dispersity of the priors and the aforementioned robust approximation ratio.
Furthermore, this can be achieved by very simple selling mechanisms.

More precisely, we show that selling the items via separate price lotteries achieves
an $O(\log r)$ approximation ratio where $r=\max_j(\sigma_j/\mu_j)$ is the maximum
coefficient of variation across the items. To prove the result, we leverage a price lottery
for the single-item case. 
If forced to restrict ourselves to deterministic mechanisms, this guarantee degrades
to $O(r^2)$.
Assuming independence of the item valuations, these ratios can be further improved
by pricing the full bundle.
For the case of identical means and variances, in particular, we get a guarantee of
$O(\log(r/m))$ which converges to optimality as the number of items grows
large.
We demonstrate the optimality of the above mechanisms by providing
matching lower bounds. Our tight analysis for the single-item deterministic case resolves an open gap from the work of \citeauthor{Azar:2013aa} [ITCS'13].

As a by-product, we also show how one can directly use our upper
bounds to improve and extend previous results related to the \emph{parametric auctions}
of~\citeauthor{Azar2013} [SODA'13]. 
\end{abstract}
\clearpage

\section{Introduction}\label{sec:intro}

Optimal auction design is one of the most well-studied and fundamental problems in
(algorithmic) mechanism design. In the traditional Myersonian~\citep{Myerson1981}
setting, an auctioneer has a single item for sale and there are $n$ interested
bidders. Each bidder has a (private) valuation for the item which, intuitively,
represents the amount of money they are willing to spend to buy it. The standard
Bayesian approach is to assume that the seller has only an incomplete knowledge of
these valuations, in the form of a prior joint distribution $F$. A selling mechanism
receives bids from the buyers and then decides to whom the item should be allocated
(which, in general, can be a randomized rule) and for what price.
The goal is to design a truthful\footnote{In such mechanisms it is to the buyers'
best interest to honestly report their actual private valuation when bidding. For
formal definitions, see~\cref{sec:model}. We notice here that this is essentially
without loss to the revenue maximization objective, due to the Revelation Principle
(see, e.g.,~\citep{Nisan07}).} selling mechanism that maximizes the auctioneer's
revenue, in expectation over $F$.

\Citet{Myerson1981} provided a complete and very elegant solution for this problem
when bidder valuations are independent, that is, $F$ is a product distribution. In
particular, when the distributions are identical and further satisfy a regularity
assumption, the optimal mechanism takes the very satisfying form of a second-price
(Vickrey) auction with a reserve price. Unfortunately, in general these
characterizations collapse when we move to multi-dimensional environments where
there are $m>1$ items for sale. Multi-item optimal auction design is one of the most
challenging and currently active research areas of mechanism design. Given that the
exact description of the revenue maximizing auctions in such settings is a
notoriously hard task, there is an impressive stream of recent papers, predominantly
from the algorithmic game theory community, that try to provide good approximation
guarantees to the optimal revenue.

The critical common underlying assumption throughout the aforementioned optimal
auction design settings is that the seller has \emph{full knowledge} of the prior
joint distribution $F$ of the bidders' valuations. In many applications though, this
might arguably be an unrealistic assumption to make: usually an auctioneer can
derive some distributional properties about the bidder population, but to completely
determine the actual distribution would require enormous resources. Thus, inspired
by the parametric auctions of~\citet{Azar:2012aa} for the single-dimensional case,
we would like to be able to design robust auctions that (1) make only use of
\emph{minimal statistical information} about the valuation distribution, namely its
mean and variance; and (2) still provide good revenue guarantees even in the worst
case against an adversarial selection of the actual distribution $F$; in particular,
\emph{no further assumptions} (e.g., independence of item valuations or regularity)
should in general be made about $F$. This is our main goal in this paper.

\subsection{Related Work}\label{sec:related}

As mentioned in the introduction, there has been an impressive stream of recent work
on
optimal~\citep{Manelli2007a,Daskalakis:2017aa,gk2014_sicomp,Cai2016,Haghpanah:2015aa}
and
approximately-optimal~\citep{Hart:2017aa,Li2013a,Chawla2010a,Babaioff2014a,Yao2015,Cai:2017aa,Rubinstein:2018aa}
multi-dimensional auction design, which tries to extend the traditional,
single-dimensional auction setting studied in the seminal paper
of~\citet{Myerson1981}. A prominent characteristic that can often be seen in these
papers is the ``simplicity vs optimality'' approach: knowing the computational
hardness~\citep{Daskalakis2013,Chen:2015aa,Chen:2018aa} and structural
complexity~\citep{Hart2013a,Daskalakis:2017aa} of describing exact optimality,
emphasis is placed on designing both simple and practical mechanisms that can still
provide good  revenue guarantees. Of course, this idea can be traced back to the
work of~\citet{Hartline2009a} and~\citet{Bulow1996} for the single-dimensional
setting. For a more thorough overview we refer to the recent review article
of~\citet{Roughgarden:2019aa} and the textbook of~\citet{Hartlinea}.

Related to this, and placed under the general theme of what has come to be known as
``Wilson's doctrine''~\citep{Wilson:1987aa} (see
also~\citep[Section~5.2]{Milgrom:2004aa}), there has also been significant effort
towards the direction of \emph{robust} revenue maximization: designing auctions that
make as few assumptions as possible on the seller's prior knowledge about the
bidders' valuations for the items. Examples include models where the auctioneer can
perform quantile queries~\citep{Chen:2018ab} or knows some estimate of the actual
prior~\citep{Bergemann:2011aa,Li:2019aa,Cai:2017ab}.
Another line of work studies robustness with respect to the correlation of
valuations across bidders or items~\citep{Carroll:2017aa,Gravin:2018aa,Bei:2019aa}.
Other approaches regarding the parameterization of partial distributional knowledge
were considered by~\citet{Dutting2019} and~\citet{BandiBertsimas2014}. See also the
recent survey by~\citet{Carroll2019}.

Most relevant to our work in the present paper is the model of \emph{parametric
auctions}, introduced by~\citet{Azar:2012aa}. More specifically, they study
single-dimensional (digital goods and single-item) auction settings with independent
item valuations, under the assumption that the seller has only access to the mean
$\mu_i$ and the variance $\sigma^2_i$ of each buyer's $i$ prior distribution. Using
Chebyshev-like tail bounds, they show that for the special single-bidder,
single-item case, deterministically pricing at a multiple of the standard deviation
below the mean, i.e.\ offering a take-it-or-leave-it price of $\mu-k\cdot \sigma$,
guarantees an approximation ratio of $\tilde\rho(r)$, where $\tilde\rho$ is an
increasing function taking values in $[1,\infty)$ and $r=\sigma/\mu$. In
\cref{sec:det_azar_micali}, we actually quantify this bound and show that it grows
quadratically. Under an extra assumption of Monotone Hazard Rate (MHR), they show
how the even simpler selling mechanism that just prices at $\mu$ achieves an
approximation ratio of $e$.

It is interesting to notice here that~\citet{Azar:2012aa} provide an \emph{exact}
solution, for deterministic mechanisms, to the robust optimization problem of
maximizing the expected revenue. Then, they use this maximin revenue-optimal
mechanism and compare it to the optimal social welfare (which is trivially also an
upper bound on the optimal revenue), to finally derive their upper bound guarantee
on the approximation ratio of revenue. As such, their results are not tailored to be
tight for the \emph{ratio} benchmark.
As a matter of fact, in~\citep{Azar:2013aa} the authors also provide an explicit
lower bound that can be written as $1+r^2$. This is an important motivating factor
for our work, since one of our main goals in this paper is to close these gaps and
provide tight approximation ratio bounds.

\Citet{Azar2013} use a clever reduction (see also the work
of~\citet{Chawla2014a}) to show how these results can be paired with the work
of~\citet{Dhangwatnotai2014a} regarding the VCG mechanism with reserves, in order to
design parametric auctions for very general single-dimensional settings. In
particular, they show how in matroid-constrained environments with the extra
assumption of regularity on the prior distributions (or MHR for more general
downward-closed environments), using the aforementioned parametric prices as lazy
reserves guarantees a $2\tilde\rho(r)$-approximation to the optimal (Myersonian)
revenue and a $\tilde\rho(r)$-approximation to the optimal social welfare. Here
$r=\max_i\sigma_i/\mu_i$.

Another work which is close to ours is that of~\citet{Carrasco2018}. The authors
essentially extend the model of~\citet{Azar:2012aa} to randomized mechanisms,
solving the maximin robust optimization problem with respect to revenue. Again, in
principle their results cannot be immediately translated to tight bounds for the
approximation ratio; however, unlike the deterministic case for which in the present
paper we have to design a new mechanism in order to achieve ratio optimality, we
will show that the maximin optimal lottery of~\citet{Carrasco2018} is actually also
optimal for the ratio benchmark.

\paragraph{Sample access vs knowledge of moments}  
Another stream of research studies models where the auctioneer has sample access to
the
distribution~\citep{Huang:2018aa,Cole2014a,Fu2015,Morgenstern:2016aa,Gonczarowski:2018aa,Dhangwatnotai2014a,Syrgkanis:2017aa,Allouah2022,Hu2021,Zampetakis2020}.
It is not hard to imagine scenarios where such access to individual past data might
be infeasible or impractical, e.g.\ due to data protections and privacy
restrictions. Furthermore, there might also exist computational limitations in
representing a distribution, or storing and reasoning with a large number of
samples.
In such settings, it is more natural to assume access to only some statistical
aggregates of the underlying data, such as the mean and the standard deviation.

From a theoretical perspective, the sample access model is incomparable with the
moment-based model of the present paper, as they rely on different distributional
assumptions. In particular, independence, regularity and/or upper bounds on the
support are standard assumptions in the aforementioned sample complexity papers. As
a matter of fact, these are necessary to derive non-trivial results (see e.g. the
counterexample of~\citet[Footnote~3]{Cole2014a}). Furthermore, if independence is
dropped, \citet{Dughmi2014} demonstrate that an exponential number of samples is
required in order to achieve a constant-factor approximation to the optimal revenue.
In our setting, on the other hand, we require none of the above. However, we do
assume (as a design principle) exact knowledge of the mean and an upper bound on the standard deviation.
This information cannot be retrieved exactly via any finite amount of samples, although intervals of confidence can be used to estimate it; we leave as future work the study of the revenue maximization problem when having only approximate knowledge of the distribution moments.

\paragraph{Maximin robustness for approximation ratio vs revenue}

The literature so far has focused on solving the maximin robust optimization problem with respect to revenue, but we chose to formulate and use the robust approximation ratio instead. Apart from being a natural choice for a computer scientist, since ratios of this sort are often used in algorithm design, it can also complement the existing objective and offer further insights. 

Both quantities have strengths from a theoretical and practical perspective, and their comparison can be subject to a broader discussion (see also \cref{sec:discussion}). In any case, we believe that the robust approximation ratio will come in handy in some scenarios. For instance, in large markets, where the seller's task of approximately quantifying the revenue that they expect to obtain might become daunting, the ``scale-free'' approximation ratio can be helpful. It is a more interpretable objective because the seller can observe their loss due to the limited statistical information as just a percentage of the full knowledge benchmark. Moreover, they can easily follow how changes in the mean or the standard deviation drive the optimal they can achieve with a robust mechanism away from the revenue of an ideal optimal auction.

Similarly, the robust approximation ratio is probably the suitable benchmark for understanding the effect of the knowledge of higher moments. An ambitious, meaningful open question is to show how the seller should design robustly optimal mechanisms when they know up to $N$ moments of the distribution. By assuming in our current model knowledge of also, e.g., the third moment, the seller's revenue will increase since she learns more about the underlying distribution. The ratio benchmark can show us \textit{at which rate} the maximin revenue is improving every time we add the knowledge of a higher moment and when it becomes near-optimal.

\subsection{Results and Techniques}\label{sec:results} The main focus of our paper
is a multi-dimensional auction setting where a single bidder has additive valuations
for $m$ items, drawn from a joint probability distribution $F$. We make no further
assumptions on $F$; in particular, we do not require $F$ to be a product
distribution nor do we enforce any kind of regularity.
The seller knows only the mean $\mu_j$ and (an upper bound on) the standard
deviation $\sigma_j$ of each item's $j$ marginal distribution. Based on this limited
statistical information, they are asked to fix a truthful (possibly randomized)
mechanism to sell the items. Then, an adversary chooses the actual distribution $F$
(respecting, of course, the statistical $(\mu_j,\sigma_j)$-information) and the
seller realizes the expected revenue of the auction, in the standard Bayesian way,
in expectation with respect to $F$. The main quantity of interest, which we call the
\emph{robust approximation ratio} is the ratio of the optimal revenue (which has
full knowledge of $F$ in advance) to this revenue.

Our worst-case, min-max approach is similar in spirit to the previous work of~\citet{Azar:2013aa,Azar2013} and~\citet{Carrasco2018}. However, the critical
difference in the present paper is that our main goal is to optimize the
\emph{ratio} against the optimal revenue and not just the expected revenue of the
selling mechanism on its own. It turns out that, similarly to the aforementioned
previous work, our bounds can be stated with respect to the ratio
$r_j=\sigma_j/\mu_j$ of each item's marginal distribution. This is an important statistical quantity called the
\emph{coefficient of variation (CV)};
it is essentially a ``unit-independent''
measure of the dispersion of the distribution (see, e.g.,~\citep{Moriguti1951a}
or~\citep[Sec.~2.21]{Kendall:1948aa}).

In \cref{sec:prelims} we formally introduce our model and necessary notation.
In the following two sections we focus on the single-item case, since this will be
the building block for all our results. In particular,
in~\cref{sec:deterministic-upper} we show that the robust approximation ratio of
deterministic mechanisms is \emph{exactly} $\rho_D(r)\approx 1+4r^2$ (see
\cref{def:function-upper-deterministic}), closing a gap open from the work
of~\citet{Azar:2013aa}. Similarly to previous work, in order to achieve this we solve
exactly the corresponding min-max problem (see~\cref{lem:deterministic_inner_opt});
however, the method and the solution itself have to be different, since we are
dealing with the ratio, which is a more ``sensitive'' quantity than the revenue on
its own. By ``sensitive'' we mean that its value changes in a less smooth and more unpredictable way for small perturbations of the distribution and the mechanism.

Next, in \cref{sec:randomized-upper-single} we deal with general randomized auctions
and we show that a lottery proposed
by~\citet{Carrasco2018}, which we term \emph{log-lottery}, although designed for a
different objective achieves an approximation ratio of
$\rho(r)\approx 1+\ln(1+r^2)$ (see~\cref{def:function-upper}) in our setting, which is
asymptotically optimal. We start with a quantitative analysis of the log-lottery mechanism (\cref{th:randomized-upper-single}). In particular, we show an upper bound to the robust approximation ratio that grows logarithmically in $r$. This bound already establishes a strong separation between the power of deterministic and randomized mechanisms. The question then becomes if a different randomized selling mechanism can achieve a sublogarithmic or even constant upper bound. We answer this in the negative by showing that the logarithmic upper bound is asymptotically tight. The construction of the lower bound instance
(\cref{thm:lower-bound-randomized}) is arguably the most technically challenging
part of our paper, and is based on a novel utilization of Yao's minimax principle
(see also~\cref{sec:yao-appendix}) that might be of independent interest for
deriving robust approximation lower bounds in other Bayesian mechanism design
settings as well. Informally, the adversary offers a distribution over two-point
mass distributions, finely-tuned such that the resulting mixture becomes a truncated
``equal-revenue style'' distribution (see~\cref{fig:randomlowbound-mixture}). The main difference to other settings in the literature where Yao's principle is applied is that the adversary has to randomize over probability distributions, which form an infinite-dimensional space.
We can imagine this as a space of ``distributions over distributions''. This introduces new technical challenges since the adversary's model of randomization needs to be properly defined, and more importantly, Yao's principle does not hold anymore. Thus, our goal is twofold: we need to carefully describe how the adversary constructs a space of distributions over distributions and then show that we can extend Yao's principle to such spaces.

It is important to restate that we work under the assumption that we know an \emph{upper
bound} on the standard-deviation $\sigma$ and not its exact value. Although this
makes our upper bounds more powerful, it is not a
source of ``artificial'' additional power for the adversary when designing our lower
bounds. We formalize this in~\cref{lem:exactsigmawlog}. Furthermore, this helps us
to formally demonstrate (see~\cref{prop:lower-randomized-simplify}) that our
aforementioned, Yao-based, lower bound construction lies at the ``border of
simplicity'' of any non-trivial lower bound.

In \cref{sec:many-items} we demonstrate how the $O(\log r)$-approximate mechanism of
the single-item case can be utilized to provide optimal approximation ratios for the
multi-dimensional case of $m$ items as well. More specifically, we show that selling
each item $j$ separately using the log-lottery guarantees an approximation ratio of
$\rho(r_{\max})$ where $r_{\max}=\max_j r_j$ is the maximum CV across the items. If
the seller has extra information that item valuations are independent (that is, $F$
is a product distribution), then switching to a lottery that offers all items in a
single full bundle can give an improved approximation ratio of $\rho(\bar r)$, where
$\bar r=\sqrt{\sum_j\sigma^2_j}/\sum_j \mu_j$ is the CV of the average valuation. We
complement these upper bounds by tight lower bounds in
\cref{th:lower-randomized-no-iid}; these constructions have at their core the
single-item lower bound, but they take care of delicately assigning valuations to
the remaining items so that they respect independence and the common prior
statistical information. We want to highlight that the lower bound of \cref{th:lower-randomized-no-iid} is strong enough to hold for \emph{any} number of items and \emph{any} choice of coefficients of variation $r_1,r_2, \dots, r_m$.
An interesting corollary of our upper bounds
(\cref{th:upper-bound-many-iid}) is that for the special case of independent
valuations with the same mean and variance, the approximation ratio is at most
$\rho\left(\frac{\sigma}{\mu\sqrt{m}}\right)$, converging to optimality as the
number of items grows large.

In \cref{sec:parametric} we diverge from our main model to discuss some additional
``peripheral'' results that can be deduced as direct corollaries of previous work
combined with our upper bounds, in a ``black-box'' way. First, we study the
single-dimensional, multi-bidder setting of \emph{parametric auctions} introduced
by~\citet{Azar:2013aa}. More specifically, we show how the positive results derived
in~\citet[Theorem~4.3]{Azar2013} can be further improved: running VCG with lazy
reserve prices drawn from the log-lottery guarantees a $2\rho(r)$
approximation to the optimal Myersonian revenue (\cref{th:lazy_many_items}).

Secondly, in~\cref{sec:lambda-regularity} we discuss how a relaxation of our model
that only assumes knowledge of the mean (that is, without any information
about the variance $\sigma^2$) can still produce good robust approximation ratios
under an extra regularity assumption. More precisely, in~\cref{prop:lambdaregaux}
we give an upper bound on the approximation ratio of the mechanism that just offers
the mean $\mu$ as a take-it-or-leave-it price, under the extra assumption that the
item's valuation distribution is $\lambda$-regular (see~\cref{fig:plot_regular});
this is a general notion of regularity that interpolates smoothly between regularity
à la Myerson ($\lambda=1$) and the Monotone Hazard Rate (MHR) condition ($\lambda=0$);
see, e.g.,~\citep{SchweizerSzech2019,gpz19_arxiv}. This result extends the
$e$-approximation for MHR distributions of~\citet[Theorem~3]{Azar:2012aa}.
Finally, we provide a more detailed characterization of the
relationship between the knowledge of $\lambda$-regularity and knowledge of $\sigma$,
with respect to the resulting robust approximation ratio upper bound
(see~\cref{fig:regvsdisp}).

\paragraph{Size of the coefficient of variation}
It is worth discussing briefly the implications of the size of the CV, our main
quantity of interest, for our results. We can observe that our upper bounds do
\emph{not} increase with the number of items $m$; as a matter of fact, for the case
of independently distributed items with the same mean and variance, the upper bound
even decreases with respect to $m$.
Although the CV of a distribution could be arbitrarily large in general, one could
argue that, for many practical scenarios, it is unlikely to encounter data with very
large dispersion. From a theoretical perspective, note that the CV is actually
bounded for important special classes of distributions, like MHR (which include,
e.g., the truncated normal, uniform, exponential and gamma~\citep{Barlow1996}) and,
more generally, $\lambda$-regular for a fixed $\lambda<1/2$
(see~\eqref{eq:bound_CV_regular}). Furthermore, for general distributions, if one
assumes that the CV of the item marginals are bounded by a universal constant, then
our bounds yield a constant robust approximation ratio to the optimal pricing, even
for correlated distributions (and regardless of the number of items).

\section{Preliminaries} \label{sec:prelims}

\subsection{Model and Notation}\label{sec:model}

A real nonnegative random variable will be called \emph{$(\mu,\sigma)$-distributed}
if its expectation is $\mu$ and its standard deviation is \emph{at most} $\sigma$.
We let $\fbb_{\mu,\sigma}$ denote the class of $(\mu,\sigma)$ distributions. We
shall also briefly (see~\cref{lem:exactsigmawlog}) discuss the restriction to
distributions with standard deviation of \emph{exactly} $\sigma$; this subclass will
be denoted by $\fbb^=_{\mu,\sigma}$.

For the most part of this paper we study auctions with $m$ items and a single
additive bidder, whose valuations $(v_1,\dots,v_m)$ for the items are drawn from a
joint distribution $F$ over $\R_{\geq 0}^m$. We denote the marginal distribution of
$v_j$ by $F_j$, and assume that it has finite mean and variance. In general, we make
no further assumptions for $F$; in particular, we do not assume independence of the
random variables $v_1,\dots,v_m$ nor do we enforce any regularity or continuity
assumption. For vectors $\vec{\mu}=(\mu_1,\dots,\mu_m)\in\R_{>0}^m$,
$\vec{\sigma}=(\sigma_1,\dots,\sigma_m)\in\R_{\geq 0}^m$ we denote by
$\fbb_{\vec{\mu},\vec{\sigma}}$ the class of all $m$-dimensional
distributions whose $j$-th marginal is $(\mu_j,\sigma_j)$-distributed, for all
$j=1,\dots,m$.

A (direct revelation, possibly randomized) selling \emph{mechanism} for a single
bidder and $m$ items is defined by a pair $(x,\pi)$ where $x:\R_{\geq
0}^m\rightarrow[0,1]^m$ is the \emph{allocation rule} and $\pi:\R_{\geq
0}^m\rightarrow\R_{\geq 0}$ is the \emph{payment rule}. If the buyer submits as bid
a valuation vector of $\vec{v}$, then they receive each item $i$ with probability
$x_i(\vec{v})$, and are charged (a total of) $\pi(\vec{v})$. We restrict our study
to \emph{truthful} mechanisms, which are characterized by the conditions
\begin{align}x(\vec{v})\cdot \vec{v}-\pi(\vec{v})&\geq
x(\vec{w})\cdot\vec{v}-\pi(\vec{w}) & \text{for all }\vec{v},\vec{w};\\
x(\vec{v})\cdot \vec{v}-\pi(\vec{v})&\geq 0 & \text{for all
}\vec{v}.\label{eq:ind_rat}\end{align}
Informally, the first condition states that the bidder can not be ``better off'' by
misreporting their true valuation; the second condition, known as \emph{individual
rationality}, ensures that the bidder cannot harm themselves by truthfully
participating in the mechanism.

Let $\abb_m$ denote the space of all truthful selling mechanisms.  Then, given an
$m$-dimensional distribution $F$, we denote by
\begin{itemize}
    \item $\rev(A;F)=\expectation_{\vec{v}\sim F}[\pi(\vec{v})]$, the expected
    \emph{revenue} of $A$ (the expectation is taken w.r.t.\ $F$);
    \item $\wel(A;F)=\expectation_{\vec{v}\sim F}[x(\vec{v})\cdot \vec{v}]$, the
    expected \emph{welfare of} $A$;
    \item $\opt(F)=\sup_{A\in\abb_m} \rev(A;F)$, the optimum revenue;
    \item $\val(F)=\sup_{A\in\abb_m} \wel(A;F)$, the optimum welfare. By definition,
    this is also the welfare of a VCG auction; moreover, for a single additive
    bidder with a joint distribution in $\fbb_{\vec{\mu},\vec{\sigma}}$, this is
    just the sum of the marginal expectations, $\val(F)=\sum_{j=1}^m\mu_j$.
\end{itemize}
Note that, due to \eqref{eq:ind_rat}, we immediately have the so-called
\emph{welfare bounds} for the above quantities: for any mechanism and distribution,
\[\rev(A;F)\leq\wel(A;F)\qquad\text{and}\qquad\opt(F)\leq\val(F).\]

Our goal in this paper is to quantify the following benchmark
\begin{equation}\label{eq:apxopt}
\apx(\vec{\mu},\vec{\sigma})=\adjustlimits\inf_{A\in\abb_m}\sup_{F\in\fbb_{\vec{\mu},\vec{\sigma}}}\frac{\opt(F)}{\rev(A;F)},
\end{equation} which we call the \emph{robust approximation ratio}. The semantics
are the following: a seller chooses the best (revenue-maximizing) selling mechanism
$A$, given only knowledge of the means $\vec\mu$ and standard deviations
$\vec\sigma$ and then an adversary (``nature'') responds by choosing a worst-case
``valid'' distribution that respects the statistical information $\vec\mu$ and
$\vec\sigma$. At some parts of our paper, we restrict our attention to
\emph{deterministic} mechanisms $A$; that is, mechanisms whose allocation rule
satisfies $x(\vec v)\in\ssets{0,1}^m$, for all $\vec v$. Under this additional
constraint, the quantity in~\eqref{eq:apxopt} will be denoted by
$\dapx(\vec{\mu},\vec{\sigma})$.

For the special case of a single item ($m=1$), we know from the seminal work
of~\citet{Myerson1981} that an auction $A\in\abb_1$ is truthful if and only if its
allocation rule is monotone nondecreasing and the payment rule is given by
$\pi(v)=v\cdot x(v)-\int_0^vx(z)\,dz$. In particular, this implies that every
deterministic mechanism $A\in\abb_1$ is completely determined by a single
take-it-or-leave-it price $p\geq 0$; thus, we will feel free to sometimes abuse
notation and write $\rev(p;F)$ instead of $\rev(A;F)$ if $A$ is the deterministic
auction that sells at price $p$.

Most importantly for our work, every randomized auction for a single item can be
seen as a nonnegative random variable over prices
(see~\citet[Footnote~10]{Carrasco2018}). In particular, since the allocation rule is
monotone and takes values in $[0,1]$, it can be interpreted as the cumulative
distribution of a certain randomization over prices, which assigns the item with the
same probability as the original mechanism.\footnote{There are only two subtle
technical issues that need to be taken into account; $x$ need not be
right-continuous, and $\lim_{v\rightarrow\infty}x(v)$ need not equal 1; we can
assume these without loss of generality. Otherwise, one could take the
right-continuous closure of $x$, and either assign the remainder probability to high
prices, or apply a suitable scaling, which would only increase expected revenue.} In
this way, for a randomized single-item auction we can abuse notation and write
$p\sim A$ to denote that a price $p$ is sampled according to $A$. In this way,
$\rev(A;F)=\expectation_{p\sim A}[\rev(p;F)]$.

Finally, from~\citet{Myerson1981} we also know that for single-item settings the
optimum revenue can always be achieved by a deterministic mechanism, that is,
\begin{equation}
\label{eq:myerson_single}
\opt(F)=\sup_{p\geq0}\rev(p;F)=\sup_{p\geq 0}p\cdot (1-F(p-))
\end{equation} where we use $F(\cdot)$ for the cumulative function (cdf) of
distribution $F$ and $F(p-)=\prob{X<p}=\lim_{x\to p^-} F(x)$, where $X\sim F$. We
shall call $\opt(\cdot)$ the \emph{Myerson operator} and for now we simply observe
that this is a functional mapping distributions to real nonnegative numbers.

\subsection{Determinism vs Randomization}\label{sec:detvsrand}

We would like to give some basic intuition on how randomization helps to hedge
uncertainty. To this end, we present a simple example where a randomized strategy
beats every price.
\begin{example} Assume that we are facing a very restricted adversary who can choose
between two distributions. Distribution A has just a point mass at $1$. Distribution
B is a two-point mass distribution, which returns either $0$ or $2$ with probability
$1/2$ each.

If the seller is restricted to deterministic pricing rules, it is not hard to see
that their best strategy is to post a price equal to $1$ (and for the adversary to
choose distribution B), for a worst-case expected revenue of $\frac{1}{2}$. If the
seller posts anything above $1$, then the adversary will always respond with
distribution A, resulting in zero revenue. Consider now the following randomization
over prices: The seller posts a price of $1$ with probability $2/3$, and a price of
$2$ with probability $1/3$. If the adversary chooses Distribution A, then the
expected revenue will be $1 \cdot \frac{2}{3} = \frac{2}{3}$. Similarly if
Distribution B is chosen, then the expected revenue becomes $1 \cdot \frac{2}{3}
\cdot \frac{1}{2} + 2 \cdot \frac{1}{3} \cdot \frac{1}{2} = \frac{2}{3}$.
\end{example} 
Regardless of the adversarial response, a randomization over two prices strictly
outperforms the best deterministic pricing. In subsequent sections we formalize this
intuition, by showing a significant separation between the power of deterministic
and randomized mechanisms. A separation between determinism and randomization in
single-dimensional settings has been demonstrated by~\citet{Fu2015} under a sample access model, and by~\citet{Bergemann2008} when the seller only knows the support of the buyer's valuation distribution. Beyond the context of auction design, the observation than ambiguity averse individuals can randomize over their choices to hedge uncertainty can be traced back decades ago in the work of~\citet{Raiffa1961}.

\subsection{Auxiliary Functions and Distributions}\label{sec:auxiliary-functions}
\begin{figure}
\centering
\scalebox{0.95}{\footnotesize \input{ratio_deterministic_small}}
~
\scalebox{0.95}{\footnotesize \input{ratio_randomized_small}}
\caption{The robust approximation ratio for deterministic (left) and randomized
(right, blue) selling mechanisms for a single $(\mu,\sigma)$-distributed item, for
small values of the coefficient of variation $r=\sigma/\mu$. The former is tight and
given in~\cref{th:deterministic-upper_and_lower}. The latter is the upper bound
given by~\cref{th:randomized-upper-single}; it is asymptotically matching the lower
bound (red) of~\cref{thm:lower-bound-randomized}.}
\label{fig:plot_small}
\end{figure}
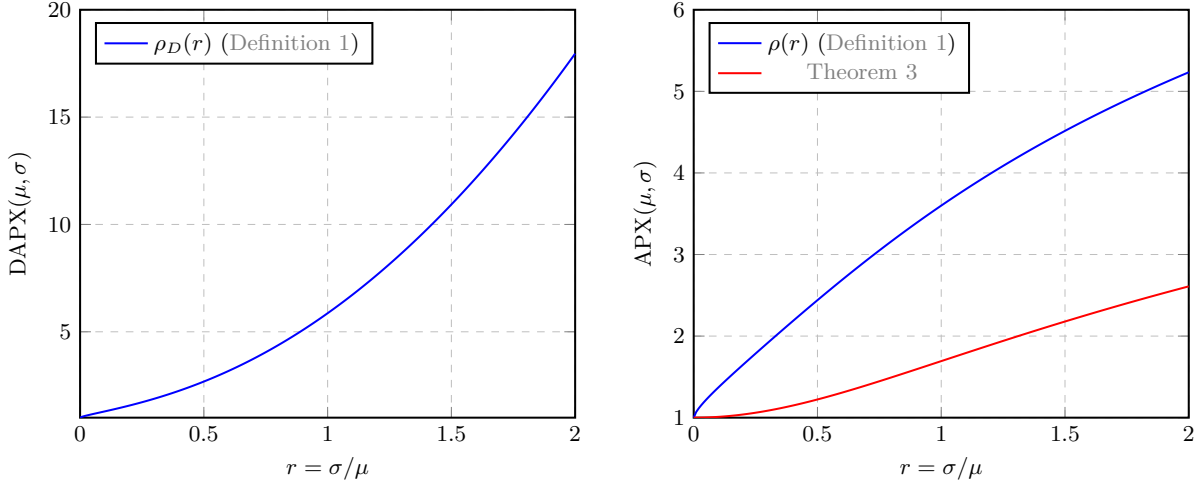

To state our bounds, it will be convenient to define the following auxiliary
functions. We will use function $\rho_D$ in \cref{sec:deterministic-upper} and $\rho$ in \cref{sec:randomized-upper-single,sec:many-items,sec:further_results}.
\begin{definition}[Functions $\rho_D$, $\rho$]
\label{def:function-upper-deterministic}
\label{def:function-upper} For any $r\geq0$, let $\rho_D(r)=\rho$, resp.
$\rho(r)=\rho$, be the unique positive solution of equation
\[\frac{(\rho-1)^3}{(2\rho-1)^2}=r^2,\qquad\text{resp.}\qquad\frac{1}{\rho^2}\left(2e^{\rho
-1}-1\right)=r^2+1.\]
\end{definition}
Plots of these functions, for small values of $r$, can be seen in
\cref{fig:plot_small}. Their asymptotic behaviour is given in the following lemma,
whose proof is deferred to \cref{sec:technical_lemmas}
(\cref{lem:exact_deterministic,lem:random_approxsol}).
\begin{lemma}
\label{lemma:upper_deterministic_bounds}
\label{lemma:upper_randomized_asymptotic} For the functions $\rho_D$, $\rho$ defined
in~\cref{def:function-upper-deterministic}, we have the bounds and asymptotics,
\[1+4r^2\leq\rho_D(r)\leq2+4r^2\qquad\text{for all }r\geq0;\qquad\qquad\rho(r) =
1+\left(1+o(1)\right)\ln(1+r^2).\]
\end{lemma}

The asymptotics for $\rho$ hold for large enough values of $r$ and that is why there is a gap in \cref{fig:plot_small} between the upper bound of $\rho$ and the lower bound of $1+\ln(1+r^2)$. The bounds asymptotically match; to be precise, they are within an $(1+o(1))$-constant factor.

\section{Single Item: Deterministic Pricing}\label{sec:deterministic-upper}

In this section we begin our study of robust revenue maximization by looking at the
simplest case: one item and deterministic pricing rules. Note that
\citet{Azar:2013aa} already established a lower bound of $1+r^2$ for this setting,
together with an upper bound which can be shown to be
$1+\left(\frac{27}{4}+o(1)\right)r^2$ (they actually characterized the upper bound
via the solution of a cubic equation; we provide the exact asymptotics of that
solution in \cref{sec:det_azar_micali}). Our result
(\cref{th:deterministic-upper_and_lower}) is a refined analysis that captures the
exact robustness ratio (and in particular the ``correct'' constant in the quadratic
term).

Our first observation (\cref{lem:deterministic_inner_opt}) will be that the
worst-case adversarial response (for a specific selling price) can be characterized
in terms of a two-point mass distribution, which allows the problem to be solved
exactly. These types of distributions have appeared already in the results of
\citet{Azar:2012aa} and \citet{Carrasco2018}, and we will start by introducing some
notation to reason about them.

A two-point mass distribution $F$ takes some value $x$ with probability $\alpha$ and
some value $y$ with probability $1-\alpha$, where without loss $x<y$. When the
distribution is constrained to have mean $\mu$ and variance exactly equal to
$\sigma^2$, only one free parameter remains, i.e.\ $F$ can be characterized by the
position $x$ of its first point mass. The other two parameters can be obtained as
\[y(x) = \mu + \cfrac{\sigma^2}{\mu -x}\quad\text{and}\quad\alpha(x)
=\cfrac{\sigma^2}{\sigma^2 + \left( \mu - x \right)^2},\]
by solving the first and second moment conditions $\mu=\alpha x +(1-\alpha) y$ and
$\mu^2+\sigma^2=\alpha x^2 +(1-\alpha) y^2$.
For the remainder, we let $F_x$, $x\in[0,\mu)$, denote this distribution. Note that
the limiting case $x\rightarrow\mu$ corresponds to $\alpha(x)\rightarrow 1$ and
$y(x)\rightarrow\infty$, meaning that $F_x$ weakly converges to $\mu$.

By first solving the innermost optimization problem in \eqref{eq:apxopt},
i.e.\ by characterizing the worst-case adversarial response against a specific
\emph{deterministic} pricing, we can derive the robustness ratio for deterministic mechanisms:

\begin{lemma}\label{lem:deterministic_inner_opt} For any choice of mean $\mu$ and
variance $\sigma^2$, and any deterministic pricing scheme, the worst-case robust
approximation ratio is achieved over a limiting two-point mass distribution.
Formally, for any $\mu,\sigma$, and any price $p$,
\begin{enumerate}
    \item if $p\geq\mu$, then the worst-case response corresponds to playing $F_x$
    with $x\rightarrow\mu^-$, and
    \[\sup_{F\in\fbb_{\mu,\sigma}}\frac{\opt(F)}{\rev(p;F)}=\infty;\]
    \item if $0<p<\mu$, then the worst-case response corresponds to playing $F_x$
    with $x\rightarrow p^-$, and
    \[\sup_{F\in\fbb_{\mu,\sigma}}\frac{\opt(F)}{\rev(p;F)}  
    =\max\left\{1+\frac{\sigma^2}{(\mu-p)^2},\frac{\mu}{p}+\frac{\sigma^2}{p(\mu-p)}\right\}.\]
\end{enumerate}
\end{lemma}

\begin{proof}
If $p\geq\mu$, then the worst-case robust approximation ratio can become arbitrarily
large by taking $x\rightarrow\mu^-$, that is, $x$ arbitrarily close to $\mu$, so
that $\alpha(x)\rightarrow 1$. Indeed, we have that $\rev(p;F_x) \leq p(1-\alpha(x))
\rightarrow 0$, whereas $\opt(F_x) \geq x \rightarrow \mu$, so that the supremum of
the ratio is unbounded.

Next, let us suppose that $0<p<\mu$. First, we compute the limit of the
approximation ratio for distribution $F_x$, as $x\rightarrow p^-$. Observe that
$\opt(F_x)=\max\{x,(1-\alpha(x))y(x)\}$; and since $x<p$, we sell the item with
probability $1-\alpha(x)$, to obtain $\rev(p;F_x)=p(1-\alpha(x))$. Therefore,
\begin{align*}
\lim_{x\rightarrow p^-}\frac{\opt(F_x)}{\rev(p,F_x)} &=\lim_{x\rightarrow
p^-}\frac{\max\{x,(1-\alpha(x))y(x)\}}{p(1-\alpha(x))}\\
&=\max\left\{\frac{1}{1-\alpha(p)},\frac{y(p)}{p}\right\}\\ &=
\max\left\{1+\frac{\sigma^2}{(\mu-p)^2},\frac{\mu}{p}+\frac{\sigma^2}{p(\mu-p)}\right\}.
\end{align*}
Thus, it only remains to show that for any random variable $X$ drawn from a
$(\mu,\sigma)$ distribution $F$, we have that
\[\frac{\opt(F)}{\rev(p;F)} \leq
\max\left\{\frac{1}{1-\alpha(p)},\frac{y(p)}{p}\right\}.\]

We first derive a lower bound on the probability of selling the item at price $p$
via a one-sided version of Chebyshev's inequality, also called Cantelli's
inequality\footnote{Although the original statement of Cantelli's inequality is for
a random variable with variance \emph{equal to} $\sigma^2$, by monotonicity the same
holds if $\sigma^2$ is instead an upper bound on the variance.} (see, e.g.,
\citep[p.~46]{BLM:2013}),
\begin{equation}\label{eq:Cantelli}
	\probability[X \geq p] = \probability[X -\mu \geq - (\mu - p)] \geq 1 -
	\frac{\sigma^2}{\sigma^2+(\mu - p)^2} = 1 - \alpha(p).
\end{equation} Let $p^\ast$ denote the optimal take-it-or-leave-it price for
distribution $F$, so that $\opt(F)=p^\ast \probability[x\geq p^\ast]$. Again, we
consider two cases: if $p^\ast \leq p$, then we have
  \[ \frac{\opt(F)}{\rev(p,F)} = \cfrac{p^\ast \probability[X \geq p^\ast]}{p
\probability[X \geq p]} \leq  \frac{1}{1-\alpha(p)} \leq
\max\left\{\frac{1}{1-\alpha(p)},\frac{y(p)}{p}\right\} \] where in the first
inequality we used \eqref{eq:Cantelli} and the bounds $p^\ast \leq p$,
$\probability[X \geq p^\ast] \leq 1$.

Next, consider the case $p^\ast > p$. By looking at the conditional random variable
$(X|X\geq p)$, we observe that
\begin{equation}
\label{eq:welfare_bound}
\frac{p^\ast \probability[X \geq p^\ast]}{\probability[X \geq p]} = p^\ast \prob{X
\geq p^\ast |  X \geq p} = \rev( p^\ast;F | X \geq p)
\leq \expectation  \left[X | X \geq p \right];
\end{equation} the inequality holds because the social welfare is always an upper
bound to the revenue.

In order to bound the conditional expectation, we use a result in
\citet[Eq.~(1.2)]{Mallows:1969aa}. It states that if $X$ is a real-valued random
variable with mean $\mu$ and variance $\sigma^2$ and $E$ is a non-zero probability
event, then
\[\expectation[ X \mid E ] - \mu \leq \sigma
\sqrt{\cfrac{1-\probability[E]}{\probability[E]}.}\] In our case, we use $E = (X
\geq p)$, together with the lower bound in \eqref{eq:Cantelli}, to get
\begin{align*}
 \expectation \left[X | X \geq p \right]
 \leq \mu + \sigma \sqrt{\frac{1}{\probability[X \geq p]} -1}
 \leq \mu + \sigma \sqrt{\cfrac{1}{1 - \alpha(p)} -1} =\mu+ \frac{\sigma^2}{\mu - p
 } = y(p).
\end{align*}
Finally, combining the above with \eqref{eq:welfare_bound} yields
\[\frac{\opt(F)}{\rev(p,F)} = \frac{p^\ast \probability[X \geq p^\ast]}{p
 \probability[X \geq p]} \leq \frac{\expectation  \left[X | X \geq p \right]}{p}
 \leq \frac{y(p)}{p} \leq \max\left\{\frac{1}{1-\alpha(p)},\frac{y(p)}{p}\right\},\]
 which concludes the proof.
\end{proof}

\begin{theorem} \label{th:deterministic-upper_and_lower} The \emph{deterministic}
robust approximation ratio of selling a single $(\mu,\sigma)$-distributed item is
\emph{exactly} equal to
\[\dapx(\mu,\sigma)=\rho_D(r) \approx 1+4\cdot r^2,\] where $r=\sigma/\mu$ and
function $\rho_D(\cdot)$ is given in~\cref{def:function-upper-deterministic}. In
particular, this is achieved by offering a take-it-or-leave-it price of
$p=\frac{\rho_D(r)}{2\rho_D(r)-1}\cdot\mu$.
\end{theorem}

\begin{proof} For fixed $\mu$ and $\sigma$, \cref{lem:deterministic_inner_opt} gives
the worst-case approximation ratio for any choice of $p$. Thus, from the seller's
perspective, it is clear that one should offer a price below the mean, and
furthermore the outermost optimization problem reduces to finding
\[\rho=\inf_{0<
p<\mu}\max\left\{1+\frac{\sigma^2}{(\mu-p)^2},\frac{\mu}{p}+\frac{\sigma^2}{p(\mu-p)}\right\}.\]
In \cref{lem:exact_deterministic} (see~\cref{sec:technical_lemmas}), we prove that
this quantity is minimized when both branches coincide; that it corresponds to the
unique positive solution of the equation
\[\frac{(\rho-1)^3}{(2\rho-1)^2}=\left(\frac{\sigma}{\mu}\right)^2;\] the desired
asymptotics; and also the characterization of the selling price $p$ in terms of
$\rho$.
\end{proof}

\section{Single Item: Lotteries}\label{sec:randomized-upper-single}

In this section, we continue to focus on a single-item setting, but now we study the
robust approximation ratio that can be achieved by a randomized mechanism, i.e.\ by
randomizing over posted prices. We first define a specific randomized selling mechanism, which essentially
corresponds to the lottery proposed by~\citet[Prop.~4]{Carrasco2018}:
\begin{definition}[Log-Lottery]
\label{def:log-pricing} Fix any $\mu>0$ and $\sigma\geq 0$. A \emph{log-lottery} is
a randomized mechanism that sells at a price $\plog_{\mu,\sigma}$, which is
distributed over the nonnegative interval support $[\pi_1,\pi_2]$ according to the
cdf
$$
F_{\mu,\sigma}^{\mathrm{log}}(x)=\frac{\pi_2\ln\frac{x}{\pi_1}-(x-\pi_1)}{\pi_2\ln\frac{\pi_2}{\pi_1}-(\pi_2-\pi_1)},
$$ 
where parameters $\pi_1,\pi_2$ are the (unique) solutions of the system
\begin{subnumcases}{}
\pi_1 \left(1+\ln\frac{ \pi_2 }{ \pi_1 }\right)=\mu
\label{eq:randomized_upper_parameter_1}\\
\quad \pi_1 (2 \pi_2 - \pi_1 )=\mu^2+\sigma^2.
\label{eq:randomized_upper_parameter_2}
\end{subnumcases}
\end{definition}
We will sometimes slightly abuse notation and use $\plog_{\mu,\sigma}$ to refer both
to the log-lottery mechanism and the corresponding random variable of the prices. 

\Citet{Carrasco2018} have given the explicit
solution to the robust \emph{absolute} revenue problem,
\begin{equation}\label{eq:vanilla_rev}
    \adjustlimits\sup_{A\in\abb_1}\inf_{F\in\fbb_{\mu,\sigma}}\rev(A;F).
\end{equation} We state below a proposition that can be directly derived from their
work and which would be very useful for our setting (the detailed derivation can be
found in~\cref{sec:vanilla_carrasco_appendix}).
\begin{proposition}\label{prop:vanilla_carrasco} For $\mu>0$, $\sigma\geq 0$, the
value of the maximin problem \eqref{eq:vanilla_rev} is given by
\[\adjustlimits\sup_{A\in\abb_1}\inf_{F\in\fbb_{\mu,\sigma}}\rev(A;F)=\pi_1,\] where
$\pi_1$ is derived by the unique solution of the system
\eqref{eq:randomized_upper_parameter_1}-\eqref{eq:randomized_upper_parameter_2}.
Moreover, this value is achieved by the log-lottery $\plog_{\mu,\sigma}$ described
in~\cref{def:log-pricing}.
\end{proposition}

 An intuitive interpretation of the result is the following: Since the seller is playing a game against an adversary, and since the seller is randomizing over prices in the equilibrium, they should expect the same revenue regardless of which value of the randomizing interval is sampled. In particular, we can evaluate the maximin expected revenue taking the lowest possible price $\pi_1$, in which case the item is always sold, and, thus, the resulting revenue is $\pi_1$. The above characterization can be directly used to derive a logarithmic upper bound on the robust approximation ratio:
\begin{theorem} \label{th:randomized-upper-single} The robust approximation ratio of
selling a single $(\mu,\sigma)$-distributed item is at most
$$
\apx(\mu,\sigma) \leq
\rho(r)\approx 1+\ln(1+r^2),
$$
where $r=\sigma/\mu$ and function $\rho$ is given in~\cref{def:function-upper}. In
particular, this is achieved by the log-lottery  described
in~\cref{def:log-pricing}.
\end{theorem}

\begin{proof} 
By \cref{prop:vanilla_carrasco}, if $A$ is the log-lottery from
\cref{def:log-pricing}, then for any $(\mu,\sigma)$ distribution $F$ we have that
$\rev(A;F)\geq\pi_1$. Thus, using the trivial upper bound of $\opt(F)\leq
\mu$ for the optimal revenue, we can derive an upper bound of $\frac{\mu}{\pi_1}$ on
the approximation ratio. For convenience, let us denote this by
$\rho\equiv\mu/\pi_1$.

Manipulating~\eqref{eq:randomized_upper_parameter_1} we get
$$
\pi_1 \left(1+\ln\frac{ \pi_2 }{ \pi_1 }\right)=\mu
\quad\ifif\quad
\ln\frac{\pi_2}{\pi_1}=\frac{\mu}{\pi_1}-1
\quad\ifif\quad
\frac{\pi_2}{\pi_1} = e^{\rho-1}
$$
and so from~\eqref{eq:randomized_upper_parameter_2} we can derive
\begin{equation*}
\quad \pi_1 (2 \pi_2 - \pi_1 )=\mu^2+\sigma^2
\quad\ifif\quad
\frac{\pi_1^2}{\mu^2}\left(2\frac{\pi_2}{\pi_1}-1\right)=\frac{\sigma^2}{\mu^2}+1
\quad\ifif\quad
\frac{1}{\rho^2}\left(2e^{\rho -1}-1\right)=r^2+1,
\label{eq:randomized_upper_implicit}
\end{equation*} 
which is exactly the equation in~\cref{def:function-upper}. The asymptotic behaviour
follows from~\cref{lemma:upper_randomized_asymptotic}.
\end{proof}
By looking at the proof of the previous theorem, it is not difficult to see that our
upper bound is also an upper bound with respect to welfare (which for a single
$(\mu,\sigma)$ distribution is simply given by $\mu$). If we were interested in
comparing the revenue of our auction to the maximum welfare, then it immediately
follows from \cref{prop:vanilla_carrasco} that the bound is exact and tight.
However, our main goal in the present paper is to provide tight bounds with respect
to the optimal \emph{revenue}, and achieving this requires some extra
work. The rest of our section is devoted to proving and discussing the following
lower bound, which asymptotically matches that of \cref{th:randomized-upper-single}. Note that even though there is a gap between the lower and upper bound for small values of $r$ (see also \cref{fig:plot_small}), the two bounds are asymptotically within a $(1+o(1))$-factor of each other.

\begin{theorem}\label{thm:lower-bound-randomized}
For a single $(\mu,\sigma)$-distributed item, the robust approximation ratio is at least
\[\apx(\mu,\sigma)\geq 1+\ln(1+r^2),\]
where $r=\sigma/\mu$.
\end{theorem}

Before we go into the actual construction of our lower bound instances, we need some
technical preliminaries and to recall Yao's principle (see, e.g.,
\citep[Sec.~8.3]{Borodin1998a} or~\citep[Sec.~2.2.2]{Motwani1995a}). As we already
mentioned (see~\cref{sec:model}), a randomized mechanism $A\in\abb_1$ can be
interpreted as a randomization over prices $p\sim A$.
From \eqref{eq:apxopt}, we are interested in the value of a game in which the
mechanism designer plays first, randomizing over posted prices, and the adversary
plays second, choosing a worst-case distribution. Intuitively, Yao's principle states
that this is at least the value of another game in which the adversary plays first,
randomizing over their choices, and the mechanism designer plays second, choosing a
deterministic response, i.e.\ a single posted price. 

However, to define this second
game formally, we would have to first explain what it means for the adversary to
randomize over probability distributions, which form an infinite-dimensional space.
In order to avoid technical or measure-theoretical issues, we focus on a
specific model of randomization, which in the literature gives rise to the concept
of \emph{mixture} or \emph{contagious} distribution (see, e.g.,~\citet[Ch.~III.4]{MGB:1974}).

\begin{definition} Let $\mathfrak{F}$ be a class of cumulative distribution
functions over the nonnegative reals, and consider any measure space over a ground
set $T$. By an \emph{$\mathfrak{F}$-mixture with parameter space $T$}, we mean a
pair $(\varTheta,F)$, where $\varTheta$ is a probability measure in $T$, and $F$ is
a measurable function of type $F:\R_{\geq0}\times T\rightarrow \R$, whose
sections are in $\mathfrak{F}$; i.e.\ for any parameter $\theta\in T$, the function
\[F_\theta:\R_{\geq0}\rightarrow \R,\quad F_\theta(x)=F(x;\theta),\] is a
cumulative distribution in $\mathfrak{F}$.

Given an $\mathfrak{F}$-mixture $(\varTheta,F)$, we denote its \emph{posterior}
distribution by $\ebb_{\theta\sim\varTheta}[F_\theta]$; this is specified by the cdf
\[\expectation_{\theta\sim\varTheta}[F_\theta](z)=\int
F(z;\theta)d\varTheta(\theta)=\expectation_{\theta\sim\varTheta}[F_\theta(z)].\]
\end{definition}
When $\mathfrak{F}=\fbb_{\mu,\sigma}$ is the class of $(\mu,\sigma)$ distributions,
we shall let $\Delta_{\mu,\sigma}$ denote the class of $(\mu,\sigma)$ mixtures, that
is, the class of mixtures over $\fbb_{\mu,\sigma}$ (with arbitrary, unspecified
parameter space).
We can interpret $(\varTheta,F)$ as a convex combination of distributions, so that
the cdf of $\expectation_{\theta\sim\varTheta}[F_\theta]$ is the convex combination
of the corresponding cdfs; alternatively, $\expectation_{\theta\sim\varTheta}[F]$
can be seen as the cdf of a random variable that first samples a distribution
$F_\theta$ according to $\theta\sim\varTheta$, and then samples a value $z$
according to $F_\theta$.

Now that we have carefully described the adversarial model, we can formally state a
version of Yao's principle (\cref{lem:yao_lem_mixtures} below) that will help us
prove lower bounds. Since this applies on ``non-standard'' continuous spaces, for
completeness we need to formally derive it ``from scratch''; we give such a
self-contained proof in~\cref{sec:yao-appendix}.

\begin{lemma}\label{lem:yao_lem_mixtures} For any $\mu,\sigma$, we have the
following lower bound on the robust approximation ratio,

\[\adjustlimits\inf_{A\in\abb_1}\sup_{F\in\fbb_{\mu,\sigma}}\frac{\opt(F)}{\rev(A;F)}\geq\adjustlimits\sup_{(\varTheta,F)\in\Delta_{\mu,\sigma}}\inf_{p\geq
0}\frac{\expectation_{\theta\sim\varTheta}[\opt(F_\theta)]}{\expectation_{\theta\sim\varTheta}[\rev(p;F_\theta)]}.\]
\end{lemma}

Note that, by using~\eqref{eq:myerson_single}, we can rewrite the denominator of the
previous quantity as follows:
\begin{align*}
\adjustlimits\sup_{p\geq 0}\expectation_{\theta\sim\varTheta}\left[\rev(p;F_\theta)\right]
&=\adjustlimits\sup_{p\geq 0}\expectation_{\theta\sim\varTheta}\left[p(1-F_\theta(p-))\right]\\
&=\sup_{p\geq 0}p\left(1-\expectation_{\theta\sim\varTheta}[F_\theta(p-)]\right)\\
&=\sup_{p\geq 0}p\left(1-\expectation_{\theta\sim\varTheta}[F_\theta](p-)\right)\\
&=\sup_{p\geq 0}\rev\left(p;\expectation_{\theta\sim\varTheta}[F_\theta]\right)\\
&=\opt\left(\expectation_{\theta\sim\varTheta}[F_\theta]\right).
\end{align*}
The second equality comes from linearity of expectation and the third one follows
from the definition of a mixture distribution. Putting all these together, we arrive at the following key technical result:
\begin{lemma} 
For any $\mu,\sigma$, the robust approximation ratio is lower bounded by
\begin{equation}\label{eq:apxlowerbound}\apx(\mu,\sigma)\geq\sup_{(\varTheta,F)\in\Delta_{\mu,\sigma}}\frac{\expectation_{\theta\sim\varTheta}[\opt(F_\theta)]}{\opt(\expectation_{\theta\sim\varTheta}[F_\theta])}.\end{equation}
\end{lemma}

From a practical perspective, the above result has a positive consequence. It allows
us to obtain lower bounds by constructing a single $(\mu,\sigma)$ mixture,
$(\varTheta,F)$, and calculating the expected optimal revenue before and after the
realization of $\theta\sim \varTheta$. Our goal is to make this ratio as
high as possible and, ideally, match the competitive ratio of the log-lottery
pricing. From this, we can gain some insight into how to construct a ``good'' 
mixture. By looking at the right-hand side of the inequality in
\cref{lem:yao_lem_mixtures}, we would intuitively expect that different posted
prices $p$ yield similar revenues of
$\expectation_{\theta\sim\varTheta}\left[\rev(p;F_\theta)\right]=\rev\left(p;\expectation_{\theta\sim\varTheta}[F_\theta]\right)$.
Thus, we would aim for a mixture $(\varTheta,F)$ for which the posterior
distribution has this property for at least some subset of its support.

From a theoretical perspective, the quantity
in~\eqref{eq:apxlowerbound} is interesting by itself. One can check that the Myerson
operator is convex, that is, the revenue achieved by a convex combination of
distributions can only be smaller than the convex combinations of the corresponding
revenues. Thus, by Jensen's inequality, the ratio in \eqref{eq:apxlowerbound} is
always at least 1. On the other hand, for a linear functional $\lcal$, we have that
$\expectation_{\theta\sim \varTheta}[\lcal(F_\theta)]=\lcal(\expectation_{\theta\sim
\varTheta}[F_\theta])$. Thus,
\eqref{eq:apxlowerbound} somehow attempts to quantify the extent to which $\opt$ is
nonlinear, or in other words, it can be understood as a \emph{measure of convexity
of the Myerson operator}. In any case, we can use this result to construct lower
bound instances and prove the main result of this section:

\begin{proof}[Proof of \cref{thm:lower-bound-randomized}]
We shall construct a $(\mu,\sigma)$ mixture over two-point mass distributions. Each
two-point mass distribution $F_\varepsilon$ is given by a unique choice of parameter
$\varepsilon\in(0,1]$; $F_\varepsilon$ returns $0$ with probability $1-\varepsilon$
and $\mu/\varepsilon$ with probability $\varepsilon$. Note that $F_\varepsilon$ has
mean $\mu$ and variance $\mu^2(1/\varepsilon-1)$. The upper bound of $\sigma^2$ on
the variance implies that we can only take values of
$\varepsilon\geq\varepsilon_0\equiv\frac{1}{1+r^2}$, where $r$ is the coefficient of
variation (our quantity of interest).

Our next step is to describe the convex mixture of these distributions. Define a
random variable with support $[\varepsilon_0,1]$ and distributed according to $B$ as
follows:
\begin{itemize}
    \item $B$ has a point mass at $\varepsilon_0$ of size $c$;
    \item $B$ is continuous over $(\varepsilon_0,1]$, with density
    $\beta(\varepsilon)=c/\varepsilon$.
\end{itemize}
The value of $c$ is given by $c=\frac{1}{1+\ln\left(1+r^2\right)}$ and is chosen as
a normalizing constant; indeed,
\[1=\expectation_{\varepsilon\sim B}[1]=c+c\ln\frac{1}{\varepsilon_0}=c\left(1+\ln\left(1+r^2\right)\right).\]

Our $(\mu,\sigma)$ mixture distribution thus corresponds to sampling $F_\varepsilon$
where $\varepsilon\sim B$. Next, we describe the posterior distribution
$G=\expectation_{\varepsilon\sim B}[F_\varepsilon]$. Its cumulative function can be
seen in~\cref{fig:randomlowbound-mixture}.
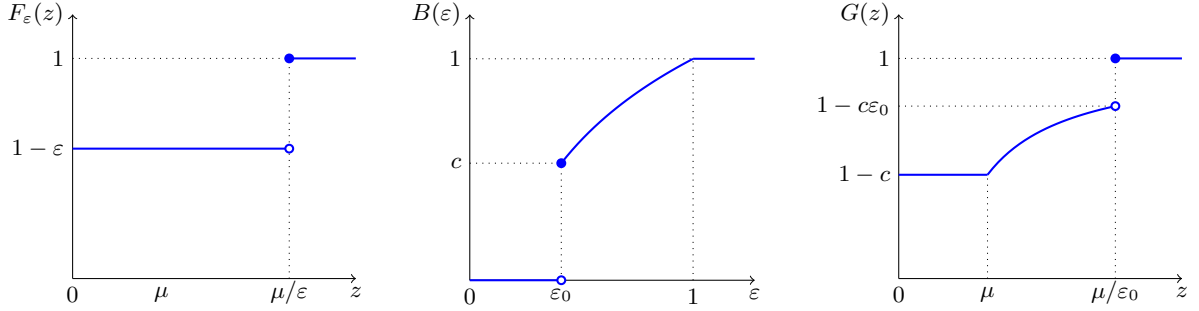
\begin{figure}
    \centering
    \begin{subfigure}[t]{0.31\textwidth}
        \resizebox{!}{4.2cm}{\footnotesize
			\input{lower_bound_two_point}
			}
			\caption{\footnotesize Two-point mass distribution with one mass at 0
			and another at $\mu/\varepsilon$.}
		\label{fig:randomlowbound-two-point}
    \end{subfigure}
    ~
        \begin{subfigure}[t]{0.31\textwidth}
  		    \resizebox{!}{4.2cm}{\footnotesize
			\input{lower_bound_mixing}
			}
			\caption{\footnotesize Mixing in parameter space; each $\varepsilon$
			corresponds to a two-point mass distribution as in
			\cref{fig:randomlowbound-two-point}. Note that this distribution has a
			mass at $\varepsilon_0$.}
		\label{fig:randomlowbound-mixing}
    \end{subfigure}
	~
    \begin{subfigure}[t]{0.31\textwidth}
        \resizebox{!}{4.2cm}{\footnotesize
			\input{equal_revenue_mixture}
			}
			\caption{\footnotesize Posterior distribution from the mixture obtained
			via \cref{fig:randomlowbound-two-point,fig:randomlowbound-mixing}.}
		\label{fig:randomlowbound-mixture}
    \end{subfigure}
    \caption{The cdfs of the various distributions used in the lower bound
    construction of \cref{thm:lower-bound-randomized}.} \label{fig:randomlowbound}
\end{figure}
\begin{itemize}
    \item Mass at 0: as each $F_\varepsilon$ has a point mass at $0$, so does $G$. The
    value of this mass is given by
    \[\expectation_{\varepsilon\sim B}[\text{mass of }F_\varepsilon\text{ at
    }0]=\int_{\varepsilon_0}^1(1-\varepsilon)\beta(\varepsilon)d\varepsilon+(1-\varepsilon_0)c=c\ln\frac{1}{\varepsilon_0}=1-c;\]
    \item Mass at $\mu/\varepsilon_0$: as $B$ has a point mass at $\varepsilon_0$ and
    $F_{\varepsilon_0}$ has a point mass at $\mu/\varepsilon_0$, this implies that $G$ has
    a point mass at $\mu/\varepsilon_0$ of size $c\varepsilon_0$;
    \item cdf in $[\mu,\mu/\varepsilon_0)$: for each $z\in[\mu,\mu/\varepsilon_0)$,
    $F_\varepsilon(z)$ is $(1-\varepsilon)$ for $\varepsilon<\mu/z$ and $1$ for $\varepsilon\geq
    \mu/z$; thus the cdf of $G$ can be computed as
    \[G(z)=\int_{\varepsilon_0}^{\mu/z}(1-\varepsilon)\beta(\varepsilon)d\varepsilon+\int_{\mu/z}^1\beta(\varepsilon)d\varepsilon+(1-\varepsilon_0)c=1-\frac{c\mu}{z}.\]
\end{itemize}
We can interpret $G(z)$ as a \emph{truncated equal-revenue} distribution over the
interval $[\mu,\mu/\varepsilon_0)$, \emph{with additional point masses} at $0$ and
$\mu/\varepsilon_0$. In particular, every posted price in $[\mu,\mu/\varepsilon]$
yields the same (optimal) revenue, and $\opt(G)=c\mu=\frac{\mu}{1+\ln(1+r^2)}$. On
the other hand, note that for every $\varepsilon>0$ we have
$\opt(F_\varepsilon)=\mu$, so $\expectation_{\varepsilon\sim
B}[\opt(F_\varepsilon)]=\mu$. Plugging these into
\eqref{eq:apxlowerbound} yields a lower bound of $1/c=1+\ln(1+r^2)$ as desired.
\end{proof}

From the previous proof, some further discussion and remarks are in order. Note that
our mixture uses distributions $F_\varepsilon$, which for
$\varepsilon>\varepsilon_0$ have a variance \emph{strictly smaller} than $\sigma^2$.
Since we have defined our adversarial model to play $(\mu,\sigma)$ distributions,
such instances are allowed. However, one may wish to ensure that the adversary only
picks distributions in $\fbb^=_{\mu,\sigma}$ (i.e.\ with \emph{exact} equality on
the variance); this might be relevant, for example, if the seller had extra
information about the exact value of $\sigma$; or, from a theoretical perspective,
such a restriction of the adversary would only make our lower bound more ``clear''
and powerful. We shall now argue that indeed our assumption on having just a bound
on the standard deviation, is not only a technical convenience (and, arguably, more
realistic), but also is without loss of generality for our bounds. Intuitively, for
any mechanism $A$ and any $(\mu,\sigma)$ distribution $F$, one can ``perturb'' $F$
into a distribution in $\fbb^=_{\mu,\sigma}$ having nearly the same approximation
ratio. Below we formalize this intuition for single-item settings, although it is
not hard to see how to generalize it to higher dimensions.

\begin{lemma}\label{lem:exactsigmawlog}
For single-item settings, the restriction of the robust approximation problem from
$(\mu,\sigma)$ distributions to distributions in $\fbb^=_{\mu,\sigma}$ does not
change its value. Formally, for any $\mu>0$, $\sigma\geq 0$, and any mechanism $A$,
we have
\[\sup_{F\in\fbb_{\mu,\sigma}}\frac{\opt(F)}{\rev(A;F)}=\sup_{F\in\fbb^=_{\mu,\sigma}}\frac{\opt(F)}{\rev(A;F)};\]
and hence
\[\adjustlimits\inf_{A\in\abb_1}\sup_{F\in\fbb_{\mu,\sigma}}\frac{\opt(F)}{\rev(A;F)}=\adjustlimits\inf_{A\in\abb_1}\sup_{F\in\fbb^=_{\mu,\sigma}}\frac{\opt(F)}{\rev(A;F)}.\]
\end{lemma}

\begin{proof}
Let $\mu$ and $\sigma$ be given, and let $A$ be any mechanism and $F_0$ any
$(\mu,\sigma)$ distribution. Suppose that the variance of $F_0$ is
$\tilde{\sigma}^2<\sigma^2$. For each $\delta\in(0,1]$, let us define the perturbed
distribution $F_\delta$ as the following convex combination of distributions:
\begin{itemize}
    \item with probability $1-\delta$, sample a value according to $F_0$;
    \item with probability $\delta$, sample a value according to the rare event
    distribution that is $0$ with probability $1-\varepsilon$ and $\mu/\varepsilon$
    with probability $\varepsilon$;
    \item the value of $\varepsilon$ is chosen so that $F_\delta$ has variance
    \emph{exactly equal to} $\sigma^2$; in other words, it is obtained by solving
    the system
    \[(1-\delta)(\mu^2+\tilde{\sigma}^2)+\delta\mu^2/\varepsilon=\mu^2+\sigma^2\quad\then\quad\varepsilon=\frac{\delta\mu^2}{\delta\mu^2+\sigma^2-(1-\delta)\tilde{\sigma}^2}.\]
\end{itemize}

Note that, for each $\delta$, $F_\delta$ has the desired mean of $\mu$ as it is the
convex combination of two distributions of mean $\mu$. Moreover, as
$\delta\rightarrow 0$, also $\varepsilon\rightarrow 0$, so that $F_\delta$ weakly
converges to $F_0$. Finally, we have the trivial bounds
\[\rev(A;F_\delta)\leq (1-\delta)\rev(A;F_0)+\delta\mu;\qquad \opt(F_\delta)\geq(1-\delta)\opt(F_0),\]
which can be combined to yield
\[\frac{\opt(F_\delta)}{\rev(A;F_\delta)}\geq\frac{(1-\delta)\opt(F_0)}{(1-\delta)\rev(A;F_0)+\delta\mu}.\]
By letting $\delta$ go to $0$, we have
\[\sup_{F\in\fbb^=_{\mu,\sigma}}\frac{\opt(F)}{\rev(A;F)}\geq\lim_{\delta\rightarrow0}\frac{(1-\delta)\opt(F_0)}{(1-\delta)\rev(A;F_0)+\delta\mu}=\frac{\opt(F_0)}{\rev(A;F_0)}.\]
Taking suprema over $F_0$ on the right-hand side yields the first statement of our
lemma; and taking infima over $A$ on both sides yields the last statement.
\end{proof}

It should also be mentioned that, in principle, we could accommodate the proof of
\cref{thm:lower-bound-randomized} to handle distributions with exact equality with
respect to $\sigma$, with minor technical modifications. More precisely, one would
define $F_{\varepsilon,\delta}$ as a perturbation of $F_\varepsilon$ as in the proof
of \cref{lem:exactsigmawlog}. This would yield an approximation ratio that depends
on $\delta$, which would then be taken in the limit $\delta\rightarrow 0$.

Another observation is that the ``bad instances'' that we used for
\cref{thm:lower-bound-randomized} were two-point mass distributions, with one of the
points being 0. Note that these differ from the instances we used in the
deterministic lower bounds (\cref{lem:deterministic_inner_opt},
\cref{th:deterministic-upper_and_lower}), which were two-point mass distributions
with exact variance of $\sigma^2$. These latter instances were actually shown in
\citep{Carrasco2018} to be worst-case distributions for their objective function,
and they were also used in \citet{Azar:2013aa} to prove maximin optimality in their
model.
Thus, it would be natural to wonder whether such instances could have been actually
enough to prove a matching lower bound in the randomized setting. Below we answer
this question in the negative; in other words, we prove a \emph{constant} upper
bound when the adversary is forced to pick one of these distributions.

\begin{proposition}\label{prop:lower-randomized-simplify}
For every choice of $\mu,\sigma$, there is a randomized mechanism $A$ that achieves
(at least) a $\frac{1}{4}$-fraction of the optimal revenue on \emph{any}
distribution $F$ that is a two-point mass with mean $\mu$ and variance $\sigma^2$.
In particular, $A$ is the mechanism that offers price $\frac{1}{2}\mu$ with
probability $\frac{1}{2}$ and $\mu+\frac{\sigma^2}{\mu}$ with probability
$\frac{1}{2}$.
\end{proposition}

\begin{proof}
Let us analyse the performance of $A$ on a two-point mass distribution $F_x$, say
with a point mass at $x$ of size $\alpha(x)$ and another at $y(x)$ of size
$1-\alpha(x)$, with $x<\mu<y(x)$. If $\frac{1}{2}\mu\leq x$ then the mechanism
chooses with probability $1/2$ a price that always sells, guaranteeing revenue of
$\frac{\mu}{4}$, which is also a $1/4$-fraction of $\opt(F)$. Next, suppose that
$x\leq\frac{1}{2}\mu$. This implies
\[1-\alpha(x)=\frac{(\mu-x)^2}{\sigma^2+(\mu-x)^2},\quad
y(x)=\mu+\frac{\sigma^2}{\mu-x}\leq\mu+2\frac{\sigma^2}{\mu},\] since $y(x)$ is a
nondecreasing function. Moreover, we have that
\[(1-\alpha(x))y(x)=\frac{\sigma^2(\mu-x)+(\mu-x)^2\mu}{\sigma^2+(\mu-x)^2}\geq\frac{\mu}{2}\frac{\sigma^2+2(\mu-x)^2}{\sigma^2+(\mu-x)^2}\geq\frac{\mu}{2}\geq
x,\] so that $\opt(F)$ is achieved by pricing at $y(x)$. Our mechanism $A$ chooses
with probability $1/2$ a price of $\mu+\frac{\sigma^2}{\mu}$, which sells with
probability $1-\alpha(x)$. Thus the approximation ratio is at least
\[\frac{\frac{1}{2}\left(1-\alpha(x)\right)\left(\mu+\frac{\sigma^2}{\mu}\right)}{(1-\alpha(x))y(x)}\geq\frac{1}{2}\frac{\mu+\frac{\sigma^2}{\mu}}{\mu+2\frac{\sigma^2}{\mu}}=\frac{1}{4}\frac{\sigma^2+\mu^2}{\sigma^2+\frac{1}{2}\mu^2}>\frac{1}{4};\]
so that the mechanism achieves a $1/4$-fraction of $\opt(F)$ in this case as well.
\end{proof}

The proposition above implies that the lower bound from
\cref{thm:lower-bound-randomized} would break down, if the adversary is restricted to the family of two-point mass distributions with exact variance of
$\sigma ^2$.

\section{Multiple Items}\label{sec:many-items}

In this section we finally consider the more general setting of a single additive
buyer with valuations for $m$ items. As it turns out, the main tools developed in
\cref{sec:randomized-upper-single} can be leveraged very naturally to produce
similar upper and lower bounds. We begin by proving upper bounds for both correlated
and independent item valuations.

\begin{theorem} \label{th:randomized-upper-many}
The robust approximation ratio of selling $m$ (possibly correlated)
$(\vec\mu,\vec\sigma)$-distributed items is at most
$$
\apx(\vec\mu,\vec\sigma) 
\leq  \rho(r_{\max}),
\qquad\qquad\text{where}\;\;\; r_{\max} =
\max_{j=1,\dots,m}r_j,\;\; r_j = \frac{\sigma_j}{\mu_j}
$$
and function $\rho$ is given in~\cref{def:function-upper}. This is achieved by
selling each item $j$ \emph{separately} using the log-lottery
$\plog_{\mu_j,\sigma_j}$ from~\cref{def:log-pricing}.

Furthermore, if the items are \emph{independently} distributed, the above bound
improves to
	$$
	\apx(\vec\mu,\vec\sigma) \leq \rho (\bar{r}), \qquad\qquad\text{where}\;\;\; \bar{r}=
	\frac{\bar{\sigma}}{\bar{\mu}},\;\; \bar{\mu}=\sum_{j=1}^m \mu_j,\;\; \bar{\sigma}=\sqrt{\sum_{j=1}^m \sigma_j^2},
	$$
	achieved by selling the items in a single \emph{full-bundle} using the
	log-lottery $\plog_{\bar{\mu},\bar{\sigma}}$ from~\cref{def:log-pricing}.
\end{theorem}

\begin{proof}
Let $X_j$, $j=1,\dots,m$, be $(\mu_j,\sigma_j)$-distributed random variables
corresponding to the marginals of the joint $m$-dimensional valuation distribution
$F$. Their sum $Y=\sum_{i=1}^m X_i$ has an expected value of
$\expect{Y}=\sum_{j=1}^m\mu_j=\bar{\mu}=\val(F)$. Furthermore, if $X_1,\dots,X_j$
are independent, its variance is
$\var{Y}=\sum_{j=1}^m\var{X_j}\leq\sum_{j=1}^m\sigma_j^2=\bar{\sigma}^2$. Denote the
distribution of $Y$ by $F_Y$.
Also, recall that the optimal revenue of $F$ cannot exceed the expected welfare,
thus we have the trivial upper bound of
\begin{equation*}
\label{eq:upper_bound_opt_many_trivial}
\opt(F) \leq \val(F) = \sum_{j=1}^m\mu_j,
\end{equation*} 
no matter if the distributions are independent or not.
	
For our general upper bound first, observe that selling item $j$ using a lottery
$A_j$, where $A_j=\plog_{\mu_j,\sigma_j}$ is the log-lottery
of~\cref{def:log-pricing}, guarantees (\cref{th:randomized-upper-single}) a revenue
of at least
\begin{equation}
    \label{eq:helper7}
\rev(A_j;F_j) \geq \frac{\mu_j}{\rho(r_j)}.
\end{equation}
Thus, if $A$ is the mechanism that sells independently each item $j$ using $A_j$, we
can get the following approximation ratio upper bound for our total revenue 
$$
\frac{OPT(F)}{\rev(A;F)} =
\frac{OPT(F)}{\sum_{j=1}^m \rev(A_j;F_j)}
\leq \frac{\sum_{j=1}^m \mu_j}{\sum_{j=1}^m \frac{\mu_j}{\rho(r_j)}}
\leq \rho(r_{\max}),
$$
where the last inequality holds due to the monotonicity of $\rho(\cdot)$: $\rho(r_j)\leq \rho (r_{\max}) $ for all $j$.

For the case of independent valuations, observe that a feasible selling mechanism
for our items is to bundle them all together and treat them as a single item, i.e.\
price their sum of valuations $Y$. Since $Y$ is
$(\bar{\mu},\bar{\sigma})$-distributed, offering a log-lottery
$A=\plog_{\bar{\mu},\bar{\sigma}}$ for $Y$ results in an approximation ratio
guarantee of
$$
\apx(\vec\mu,\vec\sigma) \leq
\frac{OPT(F)}{\rev(A;F_Y)} \leq
\frac{\expect{Y}}{\frac{1}{\rho(\bar{r})}\expectsmall{Y}}
=\rho(\bar{r}),
$$
for $\bar{r}=\bar{\sigma}/\bar{\mu}$.

Finally, to verify that $\rho(\bar{r})\leq\rho(r_{\max})$, due to the monotonicity
of $\rho(\cdot)$ it is enough to see that
$$
\bar{r}
=\frac{\bar{\sigma}}{\bar{\mu}}
=\frac{\left(\sum_{j=1}^m\sigma_j^2\right)^{1/2}}{\bar{\mu}}
\leq \frac{\sum_{j=1}^m\sigma_j}{\bar{\mu}}
= \frac{\sum_{j=1}^m\mu_jr_j}{\sum_{j=1}^m\mu_j}
$$
is a weighted average of $r_1,r_2,\dots,r_m$, and thus at most $r_{\max}$.
\end{proof}

\begin{corollary}\label{th:upper-bound-many-iid}
The robust approximation ratio of selling $m$ independently $(\mu,\sigma)$-distributed items is at most
	$$
	\apx(\vec\mu,\vec\sigma) \leq \rho\left(\frac{r}{\sqrt{m}}
	\right),
	$$
	where $r=\sigma/\mu$,
	achieved by selling the items in a single \emph{full-bundle} using the mechanism
	given in~\cref{th:randomized-upper-single}.
\end{corollary}
\begin{proof}
In the proof of~\cref{th:randomized-upper-many}, if $X_1,\dots,X_m$ are independent random
variables with mean $\mu$ and standard deviation at most $\sigma$, then for their sum $Y$ we
have $\bar{\mu}=m\cdot \mu$ and $\bar{\sigma}\leq\sqrt{m\sigma^2}=\sqrt{m}\sigma$.
\end{proof}

\begin{remarknonum}
For deterministic mechanisms, it is not difficult to see that the robust
approximation ratio of selling $m$ (possibly correlated)
$(\vec\mu,\vec\sigma)$-distributed items is at most $\dapx(\vec{\mu},\vec{\sigma})
\leq \tilde{\rho}(r_{\max})$ (where $\tilde\rho$ is given
in~\cref{sec:det_azar_micali}); just replace $\rho$ by $\tilde\rho$ in the proof
of~\cref{th:randomized-upper-many}. In particular, the validity
of~\eqref{eq:helper7} is implied by~\eqref{eq:azar-micali-approx}.
\end{remarknonum}

We make a few observations at this point. Notice that when moving from a single item
to many items, our approximation guarantees do not degrade; in particular, the
robust approximation ratio is at most that of the ``worst'' item (i.e.\ the item
with the highest coefficient of variation). In fact, for $m$ independently
$(\mu,\sigma)$-distributed items the approximation ratio even converges to
optimality (\cref{th:upper-bound-many-iid}); this can be seen as a reinterpretation
of the known result that full-bundling is asymptotically  optimal for an additive
bidder and many i.i.d.\ items (see \citet[A.5.]{Hart:2017aa}), but in our framework
of minimal statistical information.

Although the mechanisms presented in \cref{th:randomized-upper-many} are extremely
simple (lotteries over separate pricing or bundle pricing), we can actually show
asymptotically matching lower bounds for \emph{any} choice of the coefficients of
variation:
\begin{theorem}\label{th:lower-randomized-no-iid}
Fix any positive integer $m$ and positive real numbers $r_1,\ldots,r_m$, and let
$r=\max_j r_j$. Then, for any $\varepsilon>0$, there exist
$\vec\mu=(\mu_1,\dots,\mu_m)\in\R_{>0}^m$,
$\vec\sigma=(\sigma_1,\dots,\sigma_m)\in\R_{\geq 0}^m$ with $r_j=\sigma_j/\mu_j$,
such that
$$
\apx(\vec\mu,\vec\sigma)\geq 1 -\varepsilon +\ln(1+r^2).
$$
Furthermore, this lower bound is achieved by \emph{independent} $(\mu_j,\sigma_j)$-distributions.
\end{theorem}
\begin{proof}
Let $m,r_1,\ldots,r_m,\varepsilon$ be as in the statement of the theorem, and
without loss assume $\max_j r_j=r_1$. Let $\delta>0$ be chosen such that
$\delta\ln(1+r^2)(1+\ln(1+r^2))^2<\varepsilon$. We shall choose the values for the
mean and variance as
\begin{align*}\mu_1=1,&\quad \sigma_1=r_1,\\
\mu_j=\frac{\delta}{m-1},&\quad\sigma_j=r_j\frac{\delta}{m-1}\quad\text{for }j\geq 2.\end{align*}

The idea is that we create a ``bad'' instance in which items $2,\ldots,m$ are rare
event distributions with very little welfare and so their contribution to the
revenue will be negligible. To that end, we must first introduce some notation. For
every item $j\geq 2$, denote
\[p_j=\frac{1}{1+r_j^2},\quad\alpha_j=(1+r^2_j)\frac{\delta}{m-1},\]
and for every $S\subseteq\{2,\ldots,n\}$, i.e.\ for every subset of the ``low'' items,
\[p_S=\prod_{j\in S} p_j \cdot \prod_{j\not\in S\union\ssets{1}}(1-p_j).\]
Also, define the event
$$
E_S = \left[\bigland_{j\in S}{(v_j=\alpha_j)}\right]\wedge\left[\bigland_{j\notin S\union\ssets{1}}{(v_j=0)}\right].
$$

Next, let $A$ be any $m$-dimensional truthful mechanism, i.e.\ a mechanism for
selling $m$ items to a single bidder. For each $S\subseteq\{2,\ldots,n\}$, let $A_S$
be the 1-dimensional mechanism induced by event $E_S$; intuitively, this mechanism
allocates according to $A$ with the values $v_j$ set as in $E_S$, but discounting
the payment by the welfare from items in $S$. Formally, if $A$ is defined by
allocation and payment rules, $A=(\vec{x},\pi)$, then $A_S=(x_S,\pi_S)$ can be
defined as
\[x_S(v_1)=x_1(v_1,\vec{v}_{-1}),\qquad\pi_S(v_1)=\pi(v_1,\vec{v}_{-1})-\vec{v}_{-1}\cdot
\vec{x}_{-1}(0,\vec{v}_{-1}),\] where $\vec{v}_{-1}=(v_2,\ldots,v_m)$ and, for
$j\geq 2$, we have $v_j=\alpha_j$ if $j\in S$; and $v_j=0$ if $j\not\in S$. One can
directly check that $A_S$ defines a truthful mechanism.

Now define $\bar{A}=\sum_S p_SA_S$ to be the convex combination of mechanisms $A_S$.
This can be interpreted as the one-dimensional mechanism that samples a subset
$S\subseteq\{2,\ldots,n\}$ with probability $p_S$ and then runs mechanism $A_S$.
Finally, we apply \cref{thm:lower-bound-randomized} that ensures the existence of a
``bad'' single-item distribution for mechanism $A_S$, i.e.\ a distribution $F_1$
with mean $\mu_1$ and standard deviation $\sigma_1$ such that

\begin{equation}
\label{eq:distro_1_helper_lower_multi}
\rev(\bar{A};F_1) \leq \frac{\opt(F_1)}{1+\ln(1+r^2)} .
\end{equation}

Each of the remaining distributions, $F_j$ for $j=2,\dots,m$, is a rare event
distribution that assigns a mass of $p_j$ on value $\alpha_j$, and a mass of $1-p_j$
on value $0$. It is not hard to see that $F_j$ has the desired mean of $\mu_j$ and
variance of $\sigma_j^2$. To conclude the proof, let $F=F_1\times\cdots\times F_m$
be the product distribution corresponding to item-independent valuations; it only
remains to show that
\[\frac{\opt(F)}{\rev(A;F)}\geq 1-\varepsilon+\ln(1+r_1^2).\]

We first recall a standard revenue-decomposition inequality (see the proof of \
\citet[Lemma~8]{Hart:2017aa}). For any $S\subseteq\{2,\ldots,n\}$, we know that
\[\rev(A;F_1\times \cdots\times F_m|E_S)\leq\rev(A_S;F_1)+\val(F_2\times\cdots\times F_m|E_S).\]
By the construction of our two-point mass distributions $F_j$, $j\geq 2$, we know
that $E_S$ form a partition of all possible valuation profiles, each event occurring
with probability $p_S$; in this way, we can sum over the conditional expected
revenues,
\begin{align}
    \rev(A;F)&=\sum_{S}p_S\rev(A;F_1\times \cdots\times F_m|E_S)\nonumber \\
    &\leq\sum_Sp_S\left(\rev(A_S;F_1)+\val(F_2\times\cdots\times F_m|E_S)\right)\nonumber \\
    &=\rev\left(\sum_Sp_SA_S;F_1\right)+\sum_Sp_S\val(F_2\times\cdots\times F_m|E_S)\nonumber \\
    &\leq \frac{\opt(F_1)}{1+\ln(1+r^2)}+\val(F_2\times\cdots\times F_m).\label{eq:lower_rand_mult_1}
\end{align}

Next, we consider two cases. If $\rev(A;F)\leq\frac{1}{(1+\ln(1+r^2))^2}$, then
recall that by the mechanism presented in \cref{th:randomized-upper-single} one can
extract revenue of at least $\frac{1}{1+\ln(1+r^2)}$ from $F_1$, hence
\[\frac{\opt(F)}{\rev(A;F)}\geq\frac{1/(1+\ln(1+r^2))}{1/(1+\ln(1+r^2))^2}=1+\ln(1+r^2).\]

Hence we can assume that $\rev(A;F)\geq\frac{1}{(1+\ln(1+r^2))^2}$. Note that by
selling the items separately, and in particular using a price of $\alpha_j$ for
items $j=2,\dots,m$ we can lower bound the optimal revenue by
\begin{equation}
\label{eq:opt_bound_lower_multi_helper}
\opt(F_1,F_2,\dots,F_m) \geq \opt(F_1) + \sum_{j=2}^m \opt(F_j)=\opt(F_1) + \val(F_2\times\cdots\times F_m).
\end{equation}
Using this bound, together with the derivation in \eqref{eq:lower_rand_mult_1} and
the fact that $\val(F_2\times\cdots\times F_m)=\delta$, yields
\begin{align*}
    \frac{\opt(F)}{\rev(A;F)}&\geq \frac{\opt(F_1)+\delta}{\rev(A;F)}\\
    &\geq\frac{(1+\ln(1+r^2))(\rev(A;F)-\delta)+\delta}{\rev(A;F)}\\
    &=1+\ln(1+r^2)-\delta\frac{\ln(1+r^2)}{\rev(A;F)}\\
    &\geq 1+\ln(1+r^2)-\delta\ln(1+r^2)(1+\ln(1+r^2))^2\\
    &\geq 1+\ln(1+r^2)-\varepsilon,
\end{align*}
as we wanted to prove.
\end{proof}

One observation at this point is that our result for multiple items is in line with the main result of \citet{Carroll:2017aa}, but for the robust approximation ratio objective and in our framework of minimal statistical information. \citeauthor{Carroll:2017aa} also considers a multi-dimensional setting with $m$ items and a single additive buyer. In contrast to ours, the seller has full knowledge of the marginal distributions (but again does not know the joint distribution) and wants to optimize the maximin expected revenue. A crucial common point with our model is that the seller knows nothing about the correlation between the items. Similar to our main result, he proves that selling the items separately is maximin optimal. In other words, with no information regarding correlations, the seller chooses to \emph{never} bundle items. A possible interpretation of this result is the following: We know that for some correlation structures, bundling works fine, while for others, it can be very bad. Thus, the seller, who wants to be robust against an unknown, possibly correlated joint distribution, might hesitate to sell as a single unit items with no information about their correlation. At the same time, the seller can calculate the expected revenue from selling each item separately in Carroll's model. Combining these two facts intuitively makes selling separately a natural candidate for maximin optimality of the expected revenue. Our result supports this interpretation for the ratio objective and partial distributional knowledge of the marginals. Even when facing uncertainty for the revenue from a single item, the seller still chooses not to bundle items when the correlation structure is entirely unknown.  

\section{Further Results}\label{sec:further_results}

\subsection{Parametric Auctions with Lazy Reserves}\label{sec:parametric}

In this section, we present (\cref{th:lazy_many_items}) an additional immediate
consequence of our results to the setting of \citet{Azar2013}. Since this is not the
main focus of our work, we refer to the above papers, as well as
\citet[Ch.~4]{Hartlinea} for formal definitions. The key components are that we
consider a single-dimensional, matroid-constrained environment with $n$ bidders,
meaning that the set of feasible allocations forms a matroid over $\{1,\ldots,n\}$.
A class of mechanisms of particular interest are called Lazy-VCG with reserve prices
$(P_1,\ldots,P_n)$, where $P_1,\ldots,P_n$ are nonnegative random variables. This
auction works by first selecting a welfare-maximizing set $W$ of candidate winners
(i.e.\ running a VCG auction) and then offering to an agent $i\in W$ a
take-it-or-leave-it price sampled according to $P_i$. An important result in this
setting is the following black-box reduction from many bidders to one bidder with
good performance guarantees (see also \citet[Thm.~A.3]{Chawla2014a}):

\begin{theorem}[\citet{Azar2013}]
\label{th:lazy_many_daskalakis}
Assume a single-dimensional, matroid-constrained environment with $n$ bidders having
valuations drawn independently from regular distributions $F_1,F_2,\dots,F_n$. If
$P_1,\dots,P_n$ are nonnegative random variables such that for all
players $i$
\begin{equation*}
\expect[p\sim P_i]{\rev(p;F_i)} \geq c_1\cdot \opt(F_i)\qquad\text{and}\qquad
\expect[p\sim P_i]{\wel(p;F_i)} \geq c_2\cdot \val(F_i)
\end{equation*}
for constants $c_1,c_2\in[0,1]$, then Lazy-VCG with random reserves
$(P_1,\dots,P_n)$ guarantees (in expectation) a $\frac{1}{2}c_1$-fraction of the
optimal revenue and a $c_2$-fraction of the optimal welfare.
\end{theorem}

As an immediate consequence, since our log-lotteries
from~\cref{sec:randomized-upper-single} satisfy the conditions of
\cref{th:lazy_many_daskalakis} with a suitable choice of $c_1,c_2$, we get the
following:
\begin{corollary}\label{th:lazy_many_items}
Assume a single-dimensional, matroid-constrained environment with $n$ bidders having
independent regular valuations with mean $\mu_{i}$ and standard deviation
$\sigma_{i}$. Then Lazy-VCG with a reserve for player $i$ drawn from the log-lottery
$\plog_{\mu_i,\sigma_i}$ (see~\cref{def:log-pricing}) guarantees a
$2\rho(r)$-approximation to the optimal revenue and a $\rho(r)$-approximation to the
optimal welfare, where $r=\max_{i}\frac{\sigma_{i}}{\mu_{i}}$ and function
$\rho(\cdot)$ is defined in~\cref{def:function-upper}.
\end{corollary}

\begin{proof}
Take $c_1=c_2=\frac{1}{\rho(r)}\leq \frac{1}{\rho(\sigma_i/\mu_i)}$ for all $i$.
Note that the welfare bounds come ``for free'' since for any mechanism $A\in\abb_1$
we have $\wel(A;F_i)\geq\rev(A;F_i)$ and the upper bound in
\cref{th:randomized-upper-single} was derived with respect to $\val(F_i)=\mu_i$.
\end{proof}

\subsection{Regularity vs Dispersion}\label{sec:lambda-regularity}

Note that regularity plays an important role in the previous
\cref{th:lazy_many_items}, as it enables the black-box reduction of \citet{Azar2013}
to achieve meaningful upper bounds on the robust approximation ratio for a class of
multi-bidder auctions. Given this observation, an obvious question would be whether
additional knowledge of regularity can help us design better auctions, even for the
single-item, single-bidder setting of
\cref{sec:deterministic-upper,sec:randomized-upper-single}. In this section, we
consider the notion of $\lambda$-regularity, which has already been studied in the
context of revenue maximization, e.g.\ by \citet{SchweizerSzech2019} and
\citet{Cole2017}\footnote{To be precise, \citet{Cole2017} use the notion of
$\alpha$-strong regularity, originally introduced by~\citet{Cole2014a}; this
corresponds exactly to the notion of $\lambda$-regularity used
in~\citep{SchweizerSzech2019} and this paper, for $\alpha=1-\lambda$.}, prove some
basic results (\cref{cor:regular-upper}) and discuss some interesting implications.

Consider a continuous distribution $F$ supported over an interval $D_F$ of
nonnegative reals, and a real parameter $\lambda\in[0,1]$. Let $f$ denote the
density function of $F$. We will say that $F$ is \emph{$\lambda$-regular} if its
\emph{$\lambda$-generalized virtual valuation} function $$\phi_\lambda(x) \equiv
\lambda \cdot x-\frac{1-F(x)}{f(x)}$$ is monotonically nondecreasing in $D_F$.

It is not difficult to see that, for any $0\leq\lambda \leq
\lambda'\leq 1$, any $\lambda$-regular distribution is also $\lambda'$-regular.
For the special case of $\lambda=1$, the above definition recovers exactly the
standard notion of regularity à la~\citet{Myerson1981}. On the other extreme of the
range, for $\lambda=0$ we get the definition of \emph{Monotone Hazard Rate (MHR)}
distributions. Intuitively, MHR distributions have exponentially decreasing tails.
Although they represent the strictest class within the $\lambda$-regularity
hierarchy, they are still general enough to give rise to a wide family of natural
distributions, such as the uniform, exponential, (truncated) normal and gamma.

We will also need the following auxiliary results for $\lambda$-regular
distributions, which follow from Propositions~2 and~4, and their corresponding
proofs, of \cite{SchweizerSzech2019}.

\begin{proposition}[\citet{SchweizerSzech2019}] \ 
\label{prop:lambdaregaux}
\begin{enumerate}
    \item Let $F$ be $\lambda$-regular for some $\lambda\in[0,1)$. Then $F$ has a
    finite mean, say $\mu$, and we have the inequality
    \[P(X>\mu)\geq (1-\lambda)^{\frac{1}{\lambda}}\quad\text{for }\lambda\neq
    0,\quad P(X>\mu)\geq\frac{1}{e}\quad\text{for }\lambda=0.\]
    \item Let $F$ be $\lambda$-regular for some $\lambda\in[0,1/2)$. Then $F$ has a
    finite variance, say $\sigma^2$, and we have the inequality
    \[\sigma^2\leq\frac{\mu^2}{1-2\lambda}.\]
\end{enumerate}
\end{proposition}

Now we can state our main result in this section:
\begin{corollary}\label{cor:regular-upper} 
Consider a single-item, single-bidder setting in which the seller has knowledge of
the mean $\mu$ and an upper bound on the regularity $\lambda\in(0,1]$ of
distribution $F$. Then we can achieve a robust approximation ratio of
$(1-\lambda)^{-1/\lambda}$ by offering the mean as a selling price.
\end{corollary}

\begin{proof}
Using an upper bound of $\mu$ on the revenue of an optimal auction, and the lower
bound on the selling probability given by \cref{prop:lambdaregaux}, the result
immediately follows as
\[\frac{\opt(F)}{\rev(\mu;F)}\leq\frac{\mu}{\mu(1-\lambda)^{1/\lambda}}=(1-\lambda)^{-1/\lambda}.\]
\end{proof}

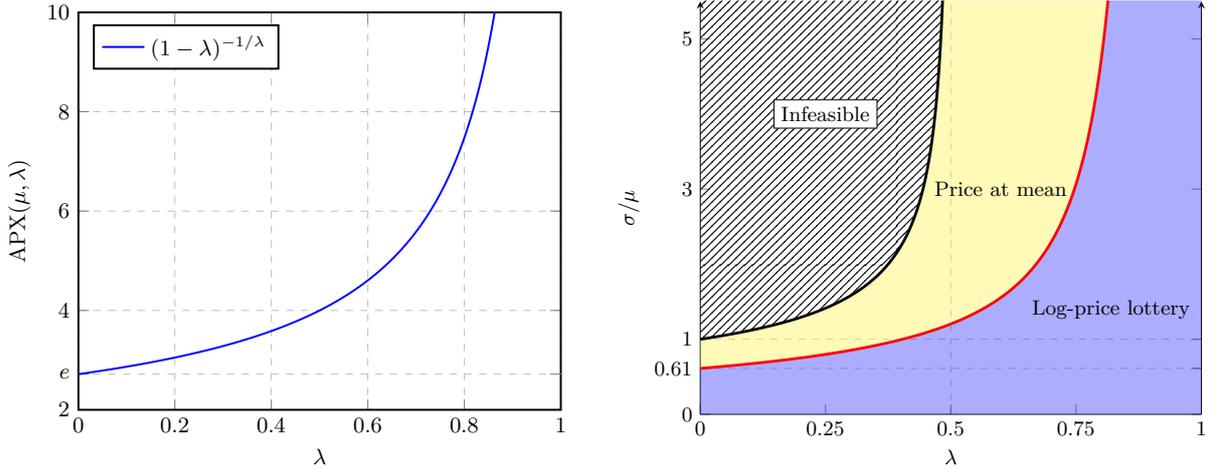
\begin{figure}
    \centering
    \begin{subfigure}[t]{0.48\textwidth}
        \resizebox{\textwidth}{!}{\footnotesize \input{ratio_regular}}
			\caption{\footnotesize Approximation ratio for $\lambda$-regular distributions.}
		\label{fig:plot_regular}
    \end{subfigure}
    ~
    \begin{subfigure}[t]{0.48\textwidth}
        \resizebox{1.05\textwidth}{!}{\footnotesize \input{regularity_dispersion}}
			\caption{\footnotesize Knowledge of $\lambda$-regularity vs variance $\sigma^2$.}
		\label{fig:regvsdisp}
    \end{subfigure}
    \caption{(a) The robust approximation ratio upper bound when pricing at
    the mean $\mu$ of a $\lambda$-regular distribution. 
    (b) Description of our proposed single-item, single-bidder mechanism under
    knowledge of $(\mu,\sigma,\lambda)$. Note that $\lambda<1/2$ already implies an
    upper bound on the coefficient of variation $\sigma/\mu$ (black curve).
    Moreover, if $\sigma/\mu$ is    sufficiently small (in particular, smaller than
    the function of $\lambda$ given by the red curve), then offering a lottery over
    prices (blue area) guarantees a better approximation ratio than simply  pricing
    at the mean (yellow area).} \label{fig:regularity-both}
\end{figure}

Note that \cref{cor:regular-upper} gives an upper bound that degrades from $e$ at
$\lambda=0$ (MHR), to $\infty$ at $\lambda=1$ (regular); see \cref{fig:plot_regular}
for a plot of this quantity. Next, we compare this ratio against the logarithmic
ratio from \cref{th:randomized-upper-single}. In other words, consider a model in
which the bidder has information about three quantities of the distribution $F$: its
mean $\mu$, an upper bound of $\sigma^2$ on its variance, and an upper bound of
$\lambda$ on its ``regularity''. Combining our results so far, we can postulate a
selling strategy, summarized in \cref{fig:regvsdisp}.
The first observation is that some triples $(\mu,\sigma,\lambda)$ are infeasible in
the following sense: if the seller knows an upper bound on $\lambda$, and
furthermore $\lambda<1/2$, then this immediately implies an upper bound on the
coefficient of variation by \cref{prop:lambdaregaux}; in particular, the seller
would know that 
\begin{equation}
    \label{eq:bound_CV_regular}
    \sigma/\mu\leq\sqrt{1/(1-2\lambda)}
\end{equation}
Thus, we can assume without
loss that triple $(\mu,\sigma,\lambda)$ obeys this inequality.

Next, we compare the robust approximation ratios of our two candidate strategies, to
determine when the log-lottery of \cref{def:log-pricing} outperforms the
pricing-at-the-mean from \cref{cor:regular-upper}. This amounts to solving the
inequality
$$
\rho\left(\frac{\sigma}{\mu}\right)\leq\frac{1}{(1-\lambda)^{1/\lambda}}.
$$
Since $\rho$ is strictly increasing, this is equivalent to
$$
\frac{\sigma}{\mu}\leq\rho^{-1}\left(\frac{1}{(1-\lambda)^{1/\lambda}}\right) = 
\sqrt{\frac{1}{(1-\lambda)^{2/\lambda}}\left(2e^{(1-\lambda^{1/\lambda}-1}-1\right)-1},
$$
where for the last equality we simply rewrote the equation in
\cref{def:function-upper} in terms of $r$. The conclusion is that the upper bound
for the log-lottery is better than the upper bound for pricing-at-the-mean iff
$\sigma/\mu$ is below a certain cutoff point (which depends on $\lambda$). Note that
\cref{fig:regvsdisp} does not show the actual approximation ratio, but rather it
partitions the $(\lambda,\sigma/\mu)$-space into regions where (the approximation
guarantee of) each mechanism is better.

Some additional observations about~\cref{fig:regvsdisp} are in order. First, in the
limit $\lambda\rightarrow 1$, our best guarantee comes from offering the log-lottery
mechanism (i.e. knowledge of $1$-regularity does not improve the currently best
known approximation guarantees for a single-item and a single-bidder); secondly,
there is a value of $\sigma/\mu$, approximately equal to $0.61$, below which
offering the log-lottery mechanism achieves a better guarantee than that provided by
pricing-at-the-mean, \emph{regardless of} the regularity parameter
$\lambda\in[0,1]$. Intuitively, one could say that \emph{knowing that the standard
deviation of $F$ is at most $61\%$ its mean gives better revenue guarantees than
knowing that $F$ is MHR}, at least for single-item, single-bidder settings.

\section{Discussion and Future Directions}\label{sec:discussion}

In this paper, we studied the robust approximation ratio of revenue maximization
under minimal statistical information of the bidders' prior distribution on the item
valuations. The fundamental quantities of interest turn out to be the coefficients
of variation (CV), $r_j=\sigma_j/\mu_j$, of the marginal distributions. For the
single-item, single-bidder case, we completely characterized this ratio for
deterministic mechanisms (quadratic in $r$) and gave asymptotically tight bounds for
randomized mechanisms (logarithmic in $r$). This yields natural upper bounds for the
multi-item, single-additive-bidder setting. The tight lower bound is particularly powerful as it works for any choice of the $r_j$. Moreover, the results hold for a possibly correlated prior distribution $F$ over the items, with only knowledge of the mean and an upper bound on the standard deviation of each marginal. The optimal mechanism turns out to be very simple: sell the items separately using the optimal randomized mechanism for the single-item case. It is also worth mentioning that although the upper bounds for the single item generalize straightforwardly to multiple items via the welfare bounds (which are trivial upper bounds to the optimal revenue), proving that these are the ``correct'' bounds requires careful technical work. At the heart of our analysis lies a new version of Yao's principle, which applies to the ``non-standard'' continuous spaces that arise in the single-item setting and might be of independent interest. 
As an interesting consequence, we have observed how our results
can be immediately applied to the single-dimensional, multi-bidder setting proposed
by \citet{Azar2013}, and also made a short digression into a setting in which
additional information on the regularity is assumed.

We believe that the general topic of ``robust revenue with minimal statistical
information'' gives rise to many interesting questions and variants; below we
propose directions for possible future work.

\paragraph{Approximation ratio vs absolute revenue} As we already mentioned in the
paper, besides the robust approximation \emph{ratio} in~\eqref{eq:apxopt}, another
quantity of independent interest is that given in~\eqref{eq:vanilla_rev}:
\begin{equation*}\label{eq:vanilla_rev_2}
\adjustlimits\sup_{A\in\abb_1}\inf_{F\in\fbb}\rev(A;F).
\end{equation*} This can be seen as a ``vanilla'' notion of robust revenue
maximization, and it was considered in \citet{Azar:2013aa} (where they proved
maximin optimality for deterministic mechanisms); it was also of central interest in
the work by \citet{Carrasco2018} and in other works in the economics and management science literature (e.g., \citep{Ko_yi_it_2020,Suzdaltsev2020,Carroll:2017aa}).

It is perhaps subjective to ask which, if any, of the two quantities is ``better'',
as both have their merits. From a theoretical perspective, the \emph{absolute
revenue} in \eqref{eq:vanilla_rev} is a simpler quantity (e.g.\ it behaves linearly
with respect to convex combinations of mechanisms and distributions) and thus
probably easier to extend to other settings; furthermore, it might be more appealing
to an economist. On the other hand, the \emph{approximation ratio} in
\eqref{eq:apxopt} is ``dimensionless'' or ``scale-free'', and arguably rather
natural for a computer scientist.

Consider the following thought experiment, that highlights this comparison from a
more practical perspective.
You are the head of a selling platform and your marketing team offers you two
possible selling mechanisms:
\begin{itemize}
    \item Mechanism A (in expectation) guarantees 10\$ on each item for sale, but
    only 25\% of the optimal revenue.
    \item Mechanism B (in expectation) guarantees 50\% of the optimal revenue, but
    only 5\$ on each item for sale.
\end{itemize}

One possible answer could be that ``it (almost) doesn't matter'': for single-item
randomized mechanisms, we proved that the maximin optimal lottery of
\citet{Carrasco2018} yields asymptotically the best possible guarantee for the
robust approximation ratio. However, it is not at all clear if, in general settings,
the maximin optimal auction always achieves a guarantee ``similar to'' (say, a
constant away of) the robust ratio-optimal auction.
Providing an answer to the debate above is of course beyond the scope of the
present paper. Nevertheless, we briefly presented it here as a potentially
stimulating topic for future work and discussion, both from a theoretical and an
empirical/behavioural point of view.

\paragraph{Tighter bounds} Although the single-item, single-bidder case is almost
completely solved, our tight analysis still has an ``$\varepsilon$-gap'' between the
upper (\cref{th:randomized-upper-single}) and lower
(\cref{thm:lower-bound-randomized}) bounds; it would be worth trying to close this
gap, either by providing stronger lower bound instances, or by improving the upper
bound analysis; similar comments apply to the multi-item, single-bidder setting,
where it would be interesting to quantify the finer dependencies of the
approximation ratio on various choices of $(\vec{\mu},\vec{\sigma})$.

\paragraph{Multiple bidders} We would also like to point out a
qualitative change between the many-items and many-bidders settings, when moving to
them from the basic single-bidder, single-item scenario: for a single bidder and
many items, the approximation guarantee does not degrade; it is essentially bounded
by the approximation guarantee of the ``worst'' item
(see~\cref{th:randomized-upper-many}). For a single item and many bidders, however,
even with the assumption of independent, regular distributions, we gain an extra
factor of $2$ (see~\cref{th:lazy_many_items}), coming from the general black box
reduction in \citet{Azar2013}. It would be interesting to see if this factor can be
dropped (or alternatively, provide stronger lower bounds). We believe that a
promising way to attack this question would be to study existing or novel bounds on
the coefficient of variation of the \emph{maximum order statistic} of random
variables, which may be of independent interest to statisticians.
Of course, the most ambitious extension would be to consider multi-dimensional,
multi-bidder settings (a generalization of both our work and that of
\citet{Azar2013,Azar:2013aa}).

Finally, below we propose alternative, or more general, models of limited
statistical information that might be interesting for future work:
\paragraph{Intervals of confidence} In some situations, it may be impractical to
assume \emph{exact} knowledge of the mean $\mu$ of a distribution; instead, one
could assume an estimate on the mean (via historical data), say
$\mu\in[\underline{\mu},\overline{\mu}]$ or $\mu\in(1\pm\varepsilon)\mu_0$. It would
be interesting to incorporate this type of information to the results presented in
the present paper, and study how they affect the approximation guarantees.
\paragraph{Broader classes of value functions} An interesting next case would be to
study the setting of, say, a single unit-demand bidder and many items, or perhaps
more generally, other valuation models such as constrained additivity or
submodularity.
\paragraph{Regularity vs dispersion} In \cref{sec:lambda-regularity}, we
considered a very simple mechanism (price at the mean) which achieves good
performance guarantees for $\lambda$-regular distributions. Note that $\lambda$ was
not used for designing the mechanism, but only for its analysis. Could we perhaps use
the knowledge of $\lambda$ to design (possibly randomized) mechanisms with better
performance? And could we combine knowledge of $\lambda$ and $\sigma$ to design
better mechanisms in a non-trivial way?
    
\paragraph{Higher-order moments} \Citet{Carrasco2018} already looked at a
single-item, single-bidder case for the ``vanilla'' revenue maximization
problem~\eqref{eq:vanilla_rev} under knowledge of the first $N$ moments of the
valuation distribution; they characterized the solution in terms of an
$N$-dimensional optimization problem, and briefly described it for the case of
$N=3$. The most intriguing question in this line of work would be to understand the
dependence of the approximation guarantee on the number of moments $N$ and,
specifically, whether it converges to optimality and at what rate. In other words,
what would be \emph{the ``moment complexity'' of robust revenue maximization}?

\appendix

\section{Technical Lemmas}\label{sec:technical_lemmas}

\begin{lemma}
\label{lem:exact_deterministic}
For any $\mu>0$ and $\sigma\geq0$, let $r=\sigma/\mu$ and
\begin{equation}
\label{eq:technical_branches}
\rho=\inf_{0< p<\mu}\max\left\{1+\frac{\sigma^2}{(\mu-p)^2},\frac{\mu}{p}+\frac{\sigma^2}{p(\mu-p)}\right\}.
\end{equation}
Then
\begin{itemize}
    \item the infimum in~\eqref{eq:technical_branches} occurs at a point $p^\ast$
    that is the unique positive solution of equation
    $$\frac{(\mu-p)^3}{2p-\mu}=\sigma^2.$$
    \item the value of the infimum, $\rho$, is the unique solution over
    $[1,\infty)$ of the equation 
    $$\frac{(\rho-1)^3}{(2\rho-1)^2}=r^2,$$
    and further satisfies
    $$
    1+4r^2\leq\rho\leq2+4r^2
    \qquad\text{and}\qquad
    p^\ast=\frac{\rho}{2\rho-1}\cdot \mu.
    $$
\end{itemize}
\end{lemma}

\begin{proof}
We begin by analysing when the first branch of the maximum
in~\eqref{eq:technical_branches} is higher than the second branch. Some algebraic
manipulation yields
\begin{equation}\label{eq:ex_det_1}1+\frac{\sigma^2}{(\mu-p)^2}\geq\frac{\mu}{p}+\frac{\sigma^2}{p(\mu-p)}\quad\ifif\quad(\mu-p)^3\leq\sigma^2(2p-\mu).\end{equation}
When $p\leq\mu/2$, the right expression is nonpositive and hence the second branch
of the maximum is highest. Next, observe that $\frac{(\mu-p)^3}{2p-\mu}$ is
decreasing over $p\in(\mu/2,\mu)$, with a positive pole at $p=\mu/2$, and vanishing
at $p=\mu$. Hence, for any choice of $\mu,\sigma$, there is a unique point $p^\ast$
at which \eqref{eq:ex_det_1} holds with equality. It follows that for $p\geq p^\ast$
the maximum is achieved on the first branch and for $p\leq p^\ast$ the maximum is
achieved on the second branch.

Next, observe that $1+\frac{\sigma^2}{(\mu-p)^2}$ is increasing on
$p\in(p^\ast,\mu)$. To see that the second branch of the maximum
in~\eqref{eq:technical_branches} is decreasing on $p\in(0,p^\ast)$, we take its
derivative
\[\frac{d}{dp}\left(\frac{\mu}{p}+\frac{\sigma^2}{p(\mu-p)}\right)=\frac{-\mu(\mu-p)^2+\sigma^2(2p-\mu)}{p^2(\mu-p)^2}.\]
When $p\leq p^\ast$ we have (by definition of $p^\ast$) that
$\sigma^2(2p-\mu)\leq(\mu-p)^3$ and hence the above quantity is at most $-1/p$,
which is negative. We conclude that the minimum occurs when both branches intersect,
i.e.\ at $p=p^\ast$; using the fact that $(\mu-p^\ast)^3=\sigma^2(2p^\ast-\mu)$, we
can further express the value of the minimum as
\[\rho=1+\frac{\sigma^2}{(\mu-p^\ast)^2}=1+\frac{\mu-p^\ast}{2p^\ast-\mu}=\frac{p^\ast}{2p^\ast-\mu}.\]
We can now use this to express $p^\ast$ in terms of $\rho$,
\[p^\ast=\frac{\rho}{2\rho-1}\cdot\mu;\qquad\mu-p^\ast=\frac{\mu(\rho-1)}{2\rho-1};\qquad
2p^\ast-\mu=\frac{\mu}{2\rho-1}.\]
Putting these together, we get
\[\sigma^2=\frac{(\mu-p^\ast)^3}{2p^\ast-\mu}=\mu^2\frac{(\rho-1)^3}{(2\rho-1)^2}
\quad\ifif\quad 
\frac{(\rho-1)^3}{(2\rho-1)^2}=\left(\frac{\sigma}{\mu}\right)^2\equiv
r^2.\]

One can directly check that the expression $\frac{(\rho-1)^3}{(2\rho-1)^2}$ is
increasing and goes from $0$ at $\rho=1$ to $\infty$ at $\rho\rightarrow\infty$, so
that for any nonnegative $r$ there is a unique solution $\rho\in[1,\infty)$ to the
above equation. Moreover, we can write
\[r^2=\frac{(\rho-1)^3}{(2\rho-1)^2}=\frac{1}{4}\rho-\frac{1}{4}-\frac{(\rho-\frac{3}{4})(\rho-1)}{4(\rho-\frac{1}{2})^2}
\quad\ifif\quad
\rho=1+4r^2+\frac{(\rho-\frac{3}{4})(\rho-1)}{(\rho-\frac{1}{2})^2};\] since the
fraction appearing on the right-hand side takes values between 0 and 1 (for
$\rho\in[1,\infty)$), this gives us the desired global bounds.
\end{proof}

\begin{lemma}
\label{lem:random_approxsol}
For $r>0$, let $\rho(r)$ denote the (unique) positive solution of the equation
$$\frac{1}{\rho^2}\left(2e^{\rho -1}-1\right)=r^2+1.$$ Then, for any
$\varepsilon>0$, it holds that
$$
\rho(r) \leq 1+(1+\varepsilon)\ln(1+r^2)
$$
for large enough values of $r$.
\end{lemma}

\begin{proof}
Fix an $\varepsilon >0$. For convenience, define the functions $f,g:(0,\infty)\map\R$ with
$$f(x)=\frac{1}{x^2}\left(2e^{x -1}-1\right)
\qquad\text{and}\qquad
g(x)= 1+(1+\varepsilon)\ln(1+x^2).
$$
By considering their derivatives, it is straightforward to see that both $f$ and $g$
are increasing functions. So, to prove our lemma, it is enough to show that
$$
f(g(r)) \geq r^2+1
$$
for large enough values of $r$.

Indeed, taking $r$ large enough we can guarantee that 
$$
g(r)= 1+(1+\varepsilon)\ln(1+r^2) \leq (1+r^2)^{\varepsilon/2},
$$
since $\ln(1+x)=o(x^{\varepsilon/2})$. Thus we have
$$
f(g(r))=\frac{2e^{g(r)-1}-1}{\left[g(r)\right]^2} 
\geq \frac{2e^{(1+\varepsilon)\ln(1+r^2)}-1}{\left[(1+r^2)^{\varepsilon/2}\right]^2}
=\frac{2(1+r^2)^{1+\varepsilon}-1}{(1+r^2)^{\varepsilon}} 
= 2(1+r^2)-\frac{1}{(1+r^2)^{\varepsilon}}
$$
which is greater than $1+r^2$ for large enough $r$, since $\frac{1}{x^\varepsilon}=o(x)$.
\end{proof}

\section{Proof of \texorpdfstring{\cref{prop:vanilla_carrasco}}{Proposition~1}}
\label{sec:vanilla_carrasco_appendix}
In this proof we refer to multiple points in the paper from \citet{Carrasco2018}.
The optimal mechanism for \eqref{eq:vanilla_rev} is given by the allocation rule
(see their Proposition 4)
\begin{equation}\label{eq:carrasco_mechanism}
x(v)=
\begin{cases} 0, & \text{for}\qquad\;\;\; v\leq\pi_1,\\ \lambda_1\ln\frac{v}{\pi_1}+2\lambda_2(v-\pi_1), & \text{for}\;\;
\pi_1 \leq v \leq \pi_2,\\ 1, & \text{for}\;\; \pi_2 \leq v,
\end{cases}\end{equation}
and the value of the maximin problem \eqref{eq:vanilla_rev} is (see end of page
274\footnote{\citet{Carrasco2018} define their solutions in terms of the moments
$k_1\equiv\mu$ and $k_2\equiv\mu^2+\sigma^2$.})
\begin{equation}\label{eq:carrasco_value_lam}
\adjustlimits\sup_{A\in\abb_1}\inf_{F\in\fbb_{\mu,\sigma}}\rev(A;F)=\lambda_0+\lambda_1\mu+\lambda_2(\mu^2+\sigma^2),
\end{equation}
where the values of $\lambda_0,\lambda_1,\lambda_2$ are given by (see (B.4-B.6))
\begin{equation}\label{eq:lambda_carrasco}
\lambda_0=-\frac{\pi_1(2\pi_2-\pi_1)}{2\left(\pi_2\ln\frac{\pi_2}{\pi_1}-(\pi_2-\pi_1)\right)};\lambda_1=\frac{\pi_2}{\pi_2\ln\frac{\pi_2}{\pi_1}-(\pi_2-\pi_1)};\lambda_2=-\frac{1}{2\left(\pi_2\ln\frac{\pi_2}{\pi_1}-(\pi_2-\pi_1)\right)}.
\end{equation}

Note that, as we explained at the end of \cref{sec:model}, the allocation rule
$x(v)$ from \eqref{eq:carrasco_mechanism} can be interpreted as the cdf of a
randomization over prices which forms an equivalent mechanism. Moreover, by
replacing the values of $\lambda_0,\lambda_1,\lambda_2$ as in
\eqref{eq:lambda_carrasco} we get
\[x(v)=\frac{\pi_2\ln\frac{v}{\pi_1}-(v-\pi_1)}{\pi_2\ln\frac{\pi_2}{\pi_1}-(\pi_2-\pi_1)},\]
which is exactly the log-lottery of \cref{def:log-pricing}.

Finally, by replacing the values of $\lambda_0,\lambda_1,\lambda_2$ from
\eqref{eq:lambda_carrasco}, and the values of $\mu$ and $\sigma$ from
\eqref{eq:randomized_upper_parameter_1},\eqref{eq:randomized_upper_parameter_2},
into \eqref{eq:carrasco_value_lam}, the value of the maximin problem can be greatly
simplified to

\begin{align*}
    \lambda_0 & +\lambda_1\mu+\lambda_2(\mu^2+\sigma^2)\\
    &=-\frac{\pi_1(2\pi_2-\pi_1)}{2\left(\pi_2\ln\frac{\pi_2}{\pi_1}-(\pi_2-\pi_1)\right)}
    +\frac{\pi_2\pi_1\left(1+\ln\frac{\pi_2}{\pi_1}\right)}{\pi_2\ln\frac{\pi_2}{\pi_1}-(\pi_2-\pi_1)}
    -\frac{\pi_1(2\pi_2-\pi_1)}{2\left(\pi_2\ln\frac{\pi_2}{\pi_1}-(\pi_2-\pi_1)\right)}\\
    &=\frac{\pi_1\left(\pi_2+\pi_2\ln\frac{\pi_2}{\pi_1}-2\pi_2+\pi_1\right)}{\pi_2\ln\frac{\pi_2}{\pi_1}-(\pi_2-\pi_1)}=\pi_1,
\end{align*}
as we wanted to prove.

\section{Asymptotics of the Mechanism by~\texorpdfstring{\citet{Azar:2012aa}}{Azar et al.}}\label{sec:det_azar_micali}
In this section we look at the upper bound proposed in \citet[Thm.~1]{Azar:2012aa}.
They propose a deterministic mechanism with selling price $p=\mu-k(r)\sigma$, where
$k(r)$ is the unique positive solution of the cubic equation
$\frac{1}{r}=\frac{1}{2}(3k+k^3)$. They derive an approximation guarantee which in
our setting can be expressed as
\begin{equation}
\label{eq:azar-micali-approx}
\apx(\mu,\sigma)
\leq \frac{\mu}{\rev(p;F)}
\leq\frac{1}{1-\frac{3}{2}r k(r)}
\equiv \tilde\rho(r).
\end{equation}

We have the following global bounds and asymptotics:
\begin{lemma}\label{lem:exact_det_azarmicali}
For any $\mu>0$ and $\sigma\geq0$, let $r=\sigma/\mu$ and let $k$ denote the unique
real solution of $\frac{1}{r}=\frac{1}{2}(3k+k^3)$. Furthermore, let
$\tilde{\rho}=\frac{1}{1-\frac{3}{2}r k}$ and $p=\mu-k\sigma$. Then $\tilde{\rho}$
is the unique solution over $[1,\infty)$ of the equation
\[\frac{27}{4}r^2=\frac{(\tilde{\rho}-1)^3}{\tilde{\rho}^2},\]
and further satisfies
$$
1+\frac{27}{4}r^2\leq\tilde{\rho}\leq3+\frac{27}{4}r^2
\qquad\text{and}\qquad
p=\frac{\tilde{\rho}+2}{3\tilde{\rho}}\cdot \mu.
$$
\end{lemma}

\begin{proof}
We begin by rewriting $k$ in terms of $\tilde{\rho}$,
\[\tilde{\rho}=\frac{1}{1-\frac{3}{2}r k}\quad\ifif\quad
k=\frac{2}{3r}\frac{\tilde{\rho}-1}{\tilde{\rho}};\] plugging this in the cubic equation for $k$,
and doing some manipulation, gives
\[\frac{1}{r}=\frac{1}{2}\left(\frac{2}{r}\frac{\tilde{\rho}-1}{\tilde{\rho}}+\frac{8}{27r^3}\frac{(\tilde{\rho}-1)^3}{\tilde{\rho}^3}\right)\quad\ifif\quad\frac{27}{4}r^2=\frac{(\tilde{\rho}-1)^3}{\tilde{\rho}^2}.\]

One can directly check that the expression
$\frac{(\tilde{\rho}-1)^3}{\tilde{\rho}^2}$ is increasing and goes from $0$ at
$\tilde{\rho}=1$ to $\infty$ at $\tilde{\rho}\rightarrow\infty$, so that for any
nonnegative $r$ there is a unique solution $\tilde{\rho}\in[1,\infty)$ to the above
equation. Moreover, we can write
\[p=\mu-k\sigma=\mu-\frac{2}{3}\frac{\sigma}{r}\frac{\tilde{\rho}-1}{\tilde{\rho}}=\frac{\tilde{\rho}+2}{3\tilde{\rho}}\cdot\mu\]
and
\[\frac{27}{4}r^2=\frac{(\tilde{\rho}-1)^3}{\tilde{\rho}^2}=\tilde{\rho}-1-\frac{(2\tilde{\rho}-1)(\tilde{\rho}-1)}{\tilde{\rho}^2}\quad\ifif\quad\tilde{\rho}=1+\frac{27}{4}r^2+\frac{(2\tilde{\rho}-1)(\tilde{\rho}-1)}{\tilde{\rho}^2}.\]
Since the fraction appearing on the right-hand side takes values between 0 and 2
(for $\tilde{\rho}\in[1,\infty)$), this gives us the desired global bounds.
\end{proof}

\section{Yao's Principle for Arbitrary Measures}\label{sec:yao-appendix}

\begin{lemma}\label{lem:yao_lem}
Let $(X,\Sigma_X,\fcal)$ and $(Y,\Sigma_Y,\gcal)$ be arbitrary probability
spaces,\footnote{For formal definitions of the measure-theoretic notions used in
this lemma see, e.g., \citet{tao2011measure}.} i.e.
\begin{itemize}
    \item $\Sigma_X$ and $\Sigma_Y$ are $\sigma$-algebras over $X$ and $Y$ respectively;
    \item $\fcal$ and $\gcal$ are probability measures over $(X,\Sigma_X)$ and $(Y,\Sigma_Y)$ respectively.
\end{itemize}

Let also $h:X\times Y\rightarrow\R_{\geq 0}$, $g:Y\rightarrow\R_{>0}$ be measurable functions.
Then\footnote{Throughout this lemma, we handle ratios of the form $\frac{g}{h}$
where $g>0$ and $h\geq 0$. For convenience, if $h=0$ we interpret the ratio as being
equal to $\infty$. This means that, for any nonnegative real number $\alpha$, we
have the following relation, even when $h=0$:
\[\frac{g}{h}\geq \alpha\quad\ifif\quad g\geq \alpha\cdot h.\]
}
\[\sup_{y\in Y}\frac{g(y)}{\ebb_{x\sim\fcal}[h(x,y)]}\geq
\inf_{x\in X}\frac{\ebb_{y\sim\gcal}[g(y)]}{\ebb_{y\sim\gcal}[h(x,y)]}.\]
\end{lemma}

\begin{proof}
Let $\alpha$ be an arbitrary nonnegative real number, and suppose that
\begin{equation}\label{eq:yao_lem_1}\inf_{x\in X}\frac{\ebb_{y\sim\gcal}[g(y)]}{\ebb_{y\sim\gcal}[h(x,y)]}\geq\alpha,\end{equation}
that is, $\ebb_{y\sim\gcal}[g(y)]\geq \alpha\sup_{x\in X}\ebb_{y\in\gcal}[h(x,y)]$.
This implies that, for every $x\in X$, we have $\ebb_{y\sim\gcal}[g(y)]\geq
\alpha\ebb_{y\in\gcal}[h(x,y)]$. Hence, by sampling $x$ according to
$\fcal$, we also have
\[\ebb_{y\sim\gcal}[g(y)]\geq\alpha\ebb_{x\sim\fcal}[\ebb_{y\sim\gcal}[h(x,y)]]=\alpha\ebb_{y\sim\gcal}[\ebb_{x\sim\fcal}[h(x,y)]];\]
the equality holds due to Tonelli's theorem (see,
e.g.,~\citet[Theorem~1.7.15]{tao2011measure}), since $h$ is measurable and
nonnegative, and $\fcal$, $\gcal$ are finite measures. By the previous inequality
between expectations, we must conclude that it holds for some realization of
$\gcal$, that is, there must exist $y\in\mathrm{supp}(\gcal)$ such that $g(y)\geq
\alpha\ebb_{x\sim\fcal}[h(x,y)]$. This implies that
\[\frac{g(y)}{\ebb_{x\sim\fcal}[h(x,y)]}\geq\alpha,\quad\text{and
hence}\quad\sup_{y\in Y}\frac{g(y)}{\ebb_{x\sim\fcal}[h(x,y)]}\geq\alpha.\]

As $\alpha$ was any real number that satisfies \eqref{eq:yao_lem_1}, the desired inequality follows.
\end{proof}

\begin{lemmanonum}[\cref{lem:yao_lem_mixtures}]
For $\mu>0$, $\sigma\geq 0$, let $\Delta_{\mu,\sigma}$ denote the class of
$(\mu,\sigma)$ mixtures, that is, the class of mixtures over $\fbb_{\mu,\sigma}$.
Then 
\[\adjustlimits\inf_{A\in\abb_1}\sup_{F\in\fbb_{\mu,\sigma}}\frac{\opt(F)}{\expectation_{p\sim
A}[\rev(p;F)]}\geq\adjustlimits\sup_{(\varTheta,F)\in\Delta_{\mu,\sigma}}\inf_{p\geq
0}\frac{\expectation_{\theta\sim\varTheta}[\opt(F_\theta)]}{\expectation_{\theta\sim\varTheta}[\rev(p;F_\theta)]}.\]
\end{lemmanonum}

\begin{proof}
Start by fixing an arbitrary truthful mechanism $A\in\abb_1$ and an arbitrary
$(\mu,\sigma)$ mixture $(\varTheta,F)$ over parameter space $T$. Since $A$ can be
interpreted as a randomization over prices, $(\R_{\geq0},\lcal,A)$ is a well-posed
probability space.

Next, define the functions
\[h:\R_{\geq 0}\times T\rightarrow\R,\quad g:T\rightarrow\R;\]
\[h(p,\theta)=\rev(p;F_\theta);\quad g(\theta)=\opt(F_\theta).\]
Clearly, $h$ is nonnegative and $g$ is positive since $F_\theta$ is
$(\mu,\sigma)$-distributed. We just need to argue that both are measurable. Note
that
\[h(p,\theta)=\rev(p;F_\theta)=p(1-F_\theta(p-))=\inf_{y<p}p(1-F(y;\theta)).\]
As $F$ is measurable and taking extrema preserves measurability, so is $h$. In a similar way, $g$ is measurable as it can be expressed as the supremum
\[g(\theta)=\opt(F_\theta)=\sup_{p\geq 0}\rev(p;F_\theta).\]
Hence, we can directly apply \cref{lem:yao_lem} and conclude that
\[\sup_{F\in\fbb_{\mu,\sigma}}\frac{\opt(F)}{\ebb_{p\sim
A}[\rev(p;F)]}\geq\sup_{\theta\in T}\frac{\opt(F_\theta)}{\ebb_{p\sim
A}[\rev(p;F_\theta)]}\geq
\inf_{p\geq
0}\frac{\ebb_{\theta\sim\varTheta}[\opt(F_\theta)]}{\ebb_{\theta\sim\varTheta}[\rev(p;F_\theta)]}.\]

As $A$ and $(\varTheta,F)$ were arbitrary, we can take the supremum on the
right-hand side over $(\mu,\sigma)$ mixtures, and the infimum on the left-hand side
over truthful mechanisms; the result follows.
\end{proof}

\bibliography{moments}
\end{document}

%% file: ratio_deterministic_small.tex
\begin{tikzpicture}
\hypersetup{colorlinks={true},linkcolor={Grey}}
\begin{axis}[
    xmin=0, xmax=2,
    ymin=1, ymax=20,
    ymajorgrids=true,
    xmajorgrids=true,
    extra y ticks={1},
    extra y tick labels={$1$},
    grid style=dashed,
    xlabel = {$r=\sigma/\mu$},
    ylabel = {$\dapx(\mu,\sigma)$},
    legend pos=north west,
    thick
]
\addplot[
    color=blue, ] coordinates {(0,1) (0.01,1.04943) (0.02,1.08148) (0.03,1.11027)
    (0.04,1.13753) (0.05,1.16397) (0.06,1.18996) (0.07,1.21573) (0.08,1.24143)
    (0.09,1.26716) (0.1,1.293) (0.11,1.31901) (0.12,1.34525) (0.13,1.37174)
    (0.14,1.39853) (0.15,1.42564) (0.16,1.4531) (0.17,1.48092) (0.18,1.50914)
    (0.19,1.53776) (0.2,1.5668) (0.21,1.59627) (0.22,1.6262) (0.23,1.65658)
    (0.24,1.68744) (0.25,1.71878) (0.26,1.75062) (0.27,1.78295) (0.28,1.81581)
    (0.29,1.84918) (0.3,1.88308) (0.31,1.91751) (0.32,1.9525) (0.33,1.98803)
    (0.34,2.02413) (0.35,2.06079) (0.36,2.09802) (0.37,2.13583) (0.38,2.17423)
    (0.39,2.21321) (0.4,2.2528) (0.41,2.29299) (0.42,2.33378) (0.43,2.37519)
    (0.44,2.41721) (0.45,2.45986) (0.46,2.50314) (0.47,2.54705) (0.48,2.59159)
    (0.49,2.63678) (0.5,2.68262) (0.51,2.7291) (0.52,2.77623) (0.53,2.82403)
    (0.54,2.87248) (0.55,2.9216) (0.56,2.97139) (0.57,3.02185) (0.58,3.07299)
    (0.59,3.12481) (0.6,3.17731) (0.61,3.23049) (0.62,3.28436) (0.63,3.33893)
    (0.64,3.39418) (0.65,3.45014) (0.66,3.50679) (0.67,3.56415) (0.68,3.62221)
    (0.69,3.68098) (0.7,3.74046) (0.71,3.80065) (0.72,3.86155) (0.73,3.92317)
    (0.74,3.98551) (0.75,4.04857) (0.76,4.11236) (0.77,4.17687) (0.78,4.2421)
    (0.79,4.30807) (0.8,4.37477) (0.81,4.44219) (0.82,4.51036) (0.83,4.57925)
    (0.84,4.64889) (0.85,4.71927) (0.86,4.79038) (0.87,4.86224) (0.88,4.93484)
    (0.89,5.00819) (0.9,5.08228) (0.91,5.15712) (0.92,5.23271) (0.93,5.30905)
    (0.94,5.38614) (0.95,5.46398) (0.96,5.54258) (0.97,5.62194) (0.98,5.70205)
    (0.99,5.78291) (1.,5.86454) (1.01,5.94692) (1.02,6.03007) (1.03,6.11397)
    (1.04,6.19864) (1.05,6.28407) (1.06,6.37026) (1.07,6.45722) (1.08,6.54495)
    (1.09,6.63344) (1.1,6.7227) (1.11,6.81273) (1.12,6.90353) (1.13,6.99509)
    (1.14,7.08743) (1.15,7.18053) (1.16,7.27441) (1.17,7.36906) (1.18,7.46449)
    (1.19,7.56069) (1.2,7.65766) (1.21,7.7554) (1.22,7.85393) (1.23,7.95322)
    (1.24,8.0533) (1.25,8.15415) (1.26,8.25578) (1.27,8.35818) (1.28,8.46137)
    (1.29,8.56533) (1.3,8.67007) (1.31,8.7756) (1.32,8.8819) (1.33,8.98898)
    (1.34,9.09685) (1.35,9.2055) (1.36,9.31493) (1.37,9.42514) (1.38,9.53613)
    (1.39,9.64791) (1.4,9.76047) (1.41,9.87381) (1.42,9.98794) (1.43,10.1029)
    (1.44,10.2186) (1.45,10.335) (1.46,10.4523) (1.47,10.5704) (1.48,10.6892)
    (1.49,10.8088) (1.5,10.9292) (1.51,11.0504) (1.52,11.1724) (1.53,11.2952)
    (1.54,11.4188) (1.55,11.5431) (1.56,11.6682) (1.57,11.7942) (1.58,11.9209)
    (1.59,12.0484) (1.6,12.1767) (1.61,12.3058) (1.62,12.4356) (1.63,12.5663)
    (1.64,12.6978) (1.65,12.83) (1.66,12.963) (1.67,13.0968) (1.68,13.2315)
    (1.69,13.3669) (1.7,13.5031) (1.71,13.64) (1.72,13.7778) (1.73,13.9164)
    (1.74,14.0558) (1.75,14.1959) (1.76,14.3368) (1.77,14.4786) (1.78,14.6211)
    (1.79,14.7644) (1.8,14.9085) (1.81,15.0535) (1.82,15.1992) (1.83,15.3456)
    (1.84,15.4929) (1.85,15.641) (1.86,15.7899) (1.87,15.9395) (1.88,16.09)
    (1.89,16.2413) (1.9,16.3933) (1.91,16.5461) (1.92,16.6998) (1.93,16.8542)
    (1.94,17.0094) (1.95,17.1654) (1.96,17.3223) (1.97,17.4799) (1.98,17.6383)
    (1.99,17.7975) (2.,17.9574) };
    \addlegendentry{$\rho_D(r)$ (\cref{def:function-upper-deterministic})}
\end{axis}
\end{tikzpicture}

%% file: ratio_randomized_small.tex
\begin{tikzpicture}
\hypersetup{colorlinks={true},linkcolor={Grey}}
\begin{axis}[
    xmin=0, xmax=2,
    ymin=1, ymax=6,
    ymajorgrids=true,
    xmajorgrids=true,
    grid style=dashed,
    xlabel = {$r=\sigma/\mu$},
    ylabel = {$\apx(\mu,\sigma)$},
    legend pos=north west,
    thick
]
\addplot[
    color=blue, ] coordinates {(0,1) (0.01,1.06961) (0.02,1.11305) (0.03,1.151)
    (0.04,1.18606) (0.05,1.21929) (0.06,1.25125) (0.07,1.28226) (0.08,1.31256)
    (0.09,1.34228) (0.1,1.37155) (0.11,1.40044) (0.12,1.42903) (0.13,1.45734)
    (0.14,1.48544) (0.15,1.51335) (0.16,1.54109) (0.17,1.56869) (0.18,1.59617)
    (0.19,1.62353) (0.2,1.6508) (0.21,1.67797) (0.22,1.70507) (0.23,1.73209)
    (0.24,1.75904) (0.25,1.78593) (0.26,1.81276) (0.27,1.83953) (0.28,1.86624)
    (0.29,1.8929) (0.3,1.91951) (0.31,1.94607) (0.32,1.97258) (0.33,1.99904)
    (0.34,2.02544) (0.35,2.0518) (0.36,2.0781) (0.37,2.10435) (0.38,2.13054)
    (0.39,2.15668) (0.4,2.18277) (0.41,2.2088) (0.42,2.23477) (0.43,2.26068)
    (0.44,2.28653) (0.45,2.31232) (0.46,2.33804) (0.47,2.3637) (0.48,2.38929)
    (0.49,2.41481) (0.5,2.44026) (0.51,2.46564) (0.52,2.49095) (0.53,2.51618)
    (0.54,2.54134) (0.55,2.56642) (0.56,2.59142) (0.57,2.61634) (0.58,2.64118)
    (0.59,2.66593) (0.6,2.69061) (0.61,2.71519) (0.62,2.73969) (0.63,2.76411)
    (0.64,2.78843) (0.65,2.81267) (0.66,2.83681) (0.67,2.86087) (0.68,2.88483)
    (0.69,2.90869) (0.7,2.93247) (0.71,2.95614) (0.72,2.97973) (0.73,3.00321)
    (0.74,3.0266) (0.75,3.04989) (0.76,3.07309) (0.77,3.09618) (0.78,3.11917)
    (0.79,3.14207) (0.8,3.16486) (0.81,3.18756) (0.82,3.21015) (0.83,3.23265)
    (0.84,3.25504) (0.85,3.27733) (0.86,3.29951) (0.87,3.3216) (0.88,3.34358)
    (0.89,3.36546) (0.9,3.38724) (0.91,3.40892) (0.92,3.43049) (0.93,3.45196)
    (0.94,3.47333) (0.95,3.4946) (0.96,3.51576) (0.97,3.53683) (0.98,3.55779)
    (0.99,3.57865) (1.,3.5994) (1.01,3.62006) (1.02,3.64061) (1.03,3.66107)
    (1.04,3.68142) (1.05,3.70167) (1.06,3.72182) (1.07,3.74187) (1.08,3.76183)
    (1.09,3.78168) (1.1,3.80143) (1.11,3.82109) (1.12,3.84065) (1.13,3.86011)
    (1.14,3.87947) (1.15,3.89873) (1.16,3.9179) (1.17,3.93697) (1.18,3.95595)
    (1.19,3.97483) (1.2,3.99362) (1.21,4.01231) (1.22,4.03091) (1.23,4.04941)
    (1.24,4.06782) (1.25,4.08614) (1.26,4.10437) (1.27,4.1225) (1.28,4.14055)
    (1.29,4.1585) (1.3,4.17637) (1.31,4.19414) (1.32,4.21183) (1.33,4.22942)
    (1.34,4.24693) (1.35,4.26435) (1.36,4.28169) (1.37,4.29893) (1.38,4.3161)
    (1.39,4.33317) (1.4,4.35016) (1.41,4.36707) (1.42,4.3839) (1.43,4.40064)
    (1.44,4.41729) (1.45,4.43387) (1.46,4.45036) (1.47,4.46677) (1.48,4.48311)
    (1.49,4.49936) (1.5,4.51553) (1.51,4.53162) (1.52,4.54764) (1.53,4.56358)
    (1.54,4.57944) (1.55,4.59522) (1.56,4.61092) (1.57,4.62655) (1.58,4.64211)
    (1.59,4.65759) (1.6,4.673) (1.61,4.68833) (1.62,4.70359) (1.63,4.71877)
    (1.64,4.73389) (1.65,4.74893) (1.66,4.7639) (1.67,4.7788) (1.68,4.79363)
    (1.69,4.80839) (1.7,4.82308) (1.71,4.8377) (1.72,4.85225) (1.73,4.86674)
    (1.74,4.88115) (1.75,4.8955) (1.76,4.90979) (1.77,4.92401) (1.78,4.93816)
    (1.79,4.95225) (1.8,4.96627) (1.81,4.98023) (1.82,4.99413) (1.83,5.00796)
    (1.84,5.02173) (1.85,5.03544) (1.86,5.04908) (1.87,5.06267) (1.88,5.07619)
    (1.89,5.08965) (1.9,5.10305) (1.91,5.1164) (1.92,5.12968) (1.93,5.14291)
    (1.94,5.15607) (1.95,5.16918) (1.96,5.18223) (1.97,5.19522) (1.98,5.20816)
    (1.99,5.22104) (2.,5.23387) };
    \addlegendentry{$\rho(r)$ (\cref{def:function-upper})}
    \addplot[
    color=red, ] coordinates {(0,1) (0.01,1.0001) (0.02,1.0004) (0.03,1.0009)
    (0.04,1.0016) (0.05,1.0025) (0.06,1.00359) (0.07,1.00489) (0.08,1.00638)
    (0.09,1.00807) (0.1,1.00995) (0.11,1.01203) (0.12,1.0143) (0.13,1.01676)
    (0.14,1.01941) (0.15,1.02225) (0.16,1.02528) (0.17,1.02849) (0.18,1.03189)
    (0.19,1.03546) (0.2,1.03922) (0.21,1.04316) (0.22,1.04727) (0.23,1.05155)
    (0.24,1.056) (0.25,1.06062) (0.26,1.06541) (0.27,1.07037) (0.28,1.07548)
    (0.29,1.08075) (0.3,1.08618) (0.31,1.09176) (0.32,1.09749) (0.33,1.10337)
    (0.34,1.10939) (0.35,1.11556) (0.36,1.12186) (0.37,1.12831) (0.38,1.13488)
    (0.39,1.14159) (0.4,1.14842) (0.41,1.15538) (0.42,1.16246) (0.43,1.16966)
    (0.44,1.17697) (0.45,1.1844) (0.46,1.19194) (0.47,1.19959) (0.48,1.20734)
    (0.49,1.21519) (0.5,1.22314) (0.51,1.23119) (0.52,1.23933) (0.53,1.24756)
    (0.54,1.25588) (0.55,1.26429) (0.56,1.27277) (0.57,1.28134) (0.58,1.28998)
    (0.59,1.2987) (0.6,1.30748) (0.61,1.31634) (0.62,1.32527) (0.63,1.33426)
    (0.64,1.34331) (0.65,1.35242) (0.66,1.36158) (0.67,1.3708) (0.68,1.38008)
    (0.69,1.3894) (0.7,1.39878) (0.71,1.40819) (0.72,1.41766) (0.73,1.42716)
    (0.74,1.43671) (0.75,1.44629) (0.76,1.4559) (0.77,1.46556) (0.78,1.47524)
    (0.79,1.48495) (0.8,1.4947) (0.81,1.50447) (0.82,1.51426) (0.83,1.52408)
    (0.84,1.53392) (0.85,1.54378) (0.86,1.55366) (0.87,1.56355) (0.88,1.57346)
    (0.89,1.58339) (0.9,1.59333) (0.91,1.60328) (0.92,1.61324) (0.93,1.62321)
    (0.94,1.63318) (0.95,1.64317) (0.96,1.65316) (0.97,1.66315) (0.98,1.67315)
    (0.99,1.68315) (1.,1.69315) (1.01,1.70315) (1.02,1.71315) (1.03,1.72314)
    (1.04,1.73314) (1.05,1.74313) (1.06,1.75311) (1.07,1.76309) (1.08,1.77307)
    (1.09,1.78303) (1.1,1.79299) (1.11,1.80294) (1.12,1.81288) (1.13,1.82281)
    (1.14,1.83274) (1.15,1.84264) (1.16,1.85254) (1.17,1.86243) (1.18,1.8723)
    (1.19,1.88215) (1.2,1.892) (1.21,1.90183) (1.22,1.91164) (1.23,1.92144)
    (1.24,1.93122) (1.25,1.94098) (1.26,1.95073) (1.27,1.96046) (1.28,1.97017)
    (1.29,1.97987) (1.3,1.98954) (1.31,1.9992) (1.32,2.00883) (1.33,2.01845)
    (1.34,2.02805) (1.35,2.03762) (1.36,2.04718) (1.37,2.05671) (1.38,2.06623)
    (1.39,2.07572) (1.4,2.08519) (1.41,2.09464) (1.42,2.10406) (1.43,2.11347)
    (1.44,2.12285) (1.45,2.13221) (1.46,2.14154) (1.47,2.15086) (1.48,2.16015)
    (1.49,2.16941) (1.5,2.17865) (1.51,2.18787) (1.52,2.19707) (1.53,2.20624)
    (1.54,2.21539) (1.55,2.22451) (1.56,2.23361) (1.57,2.24268) (1.58,2.25173)
    (1.59,2.26076) (1.6,2.26976) (1.61,2.27874) (1.62,2.28769) (1.63,2.29662)
    (1.64,2.30552) (1.65,2.3144) (1.66,2.32325) (1.67,2.33208) (1.68,2.34088)
    (1.69,2.34966) (1.7,2.35841) (1.71,2.36714) (1.72,2.37584) (1.73,2.38452)
    (1.74,2.39317) (1.75,2.4018) (1.76,2.4104) (1.77,2.41898) (1.78,2.42753)
    (1.79,2.43606) (1.8,2.44456) (1.81,2.45304) (1.82,2.46149) (1.83,2.46992)
    (1.84,2.47833) (1.85,2.48671) (1.86,2.49506) (1.87,2.50339) (1.88,2.51169)
    (1.89,2.51997) (1.9,2.52823) (1.91,2.53646) (1.92,2.54466) (1.93,2.55285)
    (1.94,2.561) (1.95,2.56914) (1.96,2.57725) (1.97,2.58533) (1.98,2.59339)
    (1.99,2.60143) (2.,2.60944) };
    \addlegendentry{\cref{thm:lower-bound-randomized}}
\end{axis}
\end{tikzpicture}

%% file: lower_bound_two_point.tex
\begin{tikzpicture}[scale = 3]
\tikzmath{\c=1.2; \m = 0.4; \r =1/(1 + ln(1 + \c^2)); \e0=1/(1+\c^2); \z=\m/\e0 + 0.3;}

\draw[->] (0, 0) -- (\z, 0);
\node[below] at (\z, 0) {$z$};
\draw[->] (0, 0) -- (0, 1.2);
\node[left] at (0, 1.2) {$F_{\varepsilon}(z)$};

\node[below] at (0, 0) {$0$};

\node[below] at (\m, 0) {$\mu$};

\node[below] at (\m/\e0, 0.03) {$\mu/\varepsilon$};
\draw[dotted] (0, 1) -- ({\m/\e0}, 1);
\draw[dotted] (\m/\e0, 0) -- (\m/\e0, 1);

\node[left] at (0, 1) {$1$};

\node[left] at (0, {1- \e0}) {$1 - \varepsilon$};

\draw[blue,thick,domain=0:{\m/\e0}] plot (\x,{1-\e0});
\draw[blue,thick,fill=white] ({\m/\e0}, {1-\e0}) circle (0.5pt);
\draw[blue,thick,domain={\m/\e0}:{\z}] plot (\x, {1});
\draw[blue,thick,fill=blue] ({\m/\e0}, {1}) circle (0.5pt);

\end{tikzpicture}

%% file: lower_bound_mixing.tex
\begin{tikzpicture}[scale = 3]
\tikzmath{\c=1.2; \m = 0.4; \r =1/(1 + ln(1 + \c^2)); \e0=1/(1+\c^2); \z=\m/\e0 + 0.3;}

\draw[->] (0, 0) -- (\z, 0);
\node[below] at (\z, 0) {$\varepsilon$};
\draw[->] (0, 0) -- (0, 1.2);
\node[left] at (0, 1.2) {$B(\varepsilon)$};

\node[below] at (0, 0) {$0$};

\node[below] at (\e0, 0) {$\varepsilon_0$};
\draw[dotted] (\e0, 0) -- (\e0, \r);

\node[below] at (1, 0) {$1$};
\draw[dotted] (0, 1) -- (1, 1);
\draw[dotted] (0, \r) -- (\e0, \r);

\node[left] at (0, 1) {$1$};

\node[left] at (0, {\r}) {$c$};

\draw[dotted] (1, 0) -- (1,1);

\draw[blue,thick,domain=0:\e0] plot (\x,0);
\draw[blue,thick,domain={\e0}:1] plot (\x, {\r*(1+ln(\x/\e0))});
\draw[blue,thick,fill=white] ({\e0}, {0}) circle (0.5pt);
\draw[blue,thick,domain=1:{\z}] plot (\x, {1});
\draw[blue,thick,fill=blue] ({\e0}, {\r}) circle (0.5pt);

\end{tikzpicture}

%% file: equal_revenue_mixture.tex
\begin{tikzpicture}[scale = 3]
\tikzmath{\c=1.2; \m = 0.4; \r =1/(1 + ln(1 + \c^2)); \e0=1/(1+\c^2); \z=\m/\e0 + 0.3;}

\draw[->] (0, 0) -- (\z, 0);
\node[below] at (\z, 0) {$z$};
\draw[->] (0, 0) -- (0, 1.2);
\node[left] at (0, 1.2) {$G(z)$};

\node[below] at (0, 0) {$0$};

\node[below] at (\m, 0) {$\mu$};
\draw[dotted] (\m, 0) -- ({\m}, 1-\r);

\node[below] at (\m/\e0, 0.03) {$\mu/\varepsilon_0$};
\draw[dotted] (0, 1) -- ({\m/\e0}, 1);
\draw[dotted] (\m/\e0, 0) -- (\m/\e0, 1);

\node[left] at (0, 1) {$1$};

\node[left] at (0, {1- \r}) {$1 - c$};

\node[left] at (0, {1- \r * \e0}) {$1 - c\varepsilon_0$};
\draw[dotted] (0, {1- \r * \e0}) -- (\m/\e0, {1- \r * \e0});

\draw[blue,thick,domain=0:\m] plot (\x,{1-\r});
\draw[blue,thick,domain={\m}:{\m/\e0}] plot (\x, {1- \r*\m / \x)});
\draw[blue,thick,fill=white] ({\m/\e0}, {1- \r * \e0}) circle (0.5pt);
\draw[blue,thick,domain={\m/\e0}:{\z}] plot (\x, {1});
\draw[blue,thick,fill=blue] ({\m/\e0}, {1}) circle (0.5pt);

\end{tikzpicture}

%% file: ratio_regular.tex
\begin{tikzpicture}
\begin{axis}[
  xmin=0, xmax=1,
  ymin=2, ymax=10,
	extra y ticks={2.71828},
	extra y tick labels={$e$},
    ymajorgrids=true,
    xmajorgrids=true,
    grid style=dashed,
    xlabel = {$\lambda$},
    ylabel = {$\apx(\mu,\lambda)$},
    legend pos=north west,
    thick
]
    \addlegendentry{$(1-\lambda)^{-1/\lambda}$}
\addplot[
    color=blue, ] coordinates {(0,2.71828) (0.001,2.71964) (0.002,2.72101)
    (0.003,2.72237) (0.004,2.72374) (0.005,2.72511) (0.006,2.72648) (0.007,2.72786)
    (0.008,2.72924) (0.009,2.73062) (0.01,2.732) (0.011,2.73338) (0.012,2.73477)
    (0.013,2.73616) (0.014,2.73756) (0.015,2.73895) (0.016,2.74035) (0.017,2.74175)
    (0.018,2.74316) (0.019,2.74456) (0.02,2.74597) (0.021,2.74738) (0.022,2.7488)
    (0.023,2.75022) (0.024,2.75164) (0.025,2.75306) (0.026,2.75448) (0.027,2.75591)
    (0.028,2.75734) (0.029,2.75877) (0.03,2.76021) (0.031,2.76165) (0.032,2.76309)
    (0.033,2.76453) (0.034,2.76598) (0.035,2.76743) (0.036,2.76888) (0.037,2.77034)
    (0.038,2.7718) (0.039,2.77326) (0.04,2.77472) (0.041,2.77619) (0.042,2.77766)
    (0.043,2.77913) (0.044,2.7806) (0.045,2.78208) (0.046,2.78356) (0.047,2.78504)
    (0.048,2.78653) (0.049,2.78802) (0.05,2.78951) (0.051,2.791) (0.052,2.7925)
    (0.053,2.794) (0.054,2.79551) (0.055,2.79701) (0.056,2.79852) (0.057,2.80003)
    (0.058,2.80155) (0.059,2.80307) (0.06,2.80459) (0.061,2.80611) (0.062,2.80764)
    (0.063,2.80917) (0.064,2.8107) (0.065,2.81224) (0.066,2.81378) (0.067,2.81532)
    (0.068,2.81686) (0.069,2.81841) (0.07,2.81996) (0.071,2.82152) (0.072,2.82308)
    (0.073,2.82464) (0.074,2.8262) (0.075,2.82777) (0.076,2.82934) (0.077,2.83091)
    (0.078,2.83249) (0.079,2.83406) (0.08,2.83565) (0.081,2.83723) (0.082,2.83882)
    (0.083,2.84041) (0.084,2.84201) (0.085,2.84361) (0.086,2.84521) (0.087,2.84681)
    (0.088,2.84842) (0.089,2.85003) (0.09,2.85165) (0.091,2.85326) (0.092,2.85488)
    (0.093,2.85651) (0.094,2.85814) (0.095,2.85977) (0.096,2.8614) (0.097,2.86304)
    (0.098,2.86468) (0.099,2.86632) (0.1,2.86797) (0.101,2.86962) (0.102,2.87128)
    (0.103,2.87293) (0.104,2.8746) (0.105,2.87626) (0.106,2.87793) (0.107,2.8796)
    (0.108,2.88127) (0.109,2.88295) (0.11,2.88463) (0.111,2.88632) (0.112,2.88801)
    (0.113,2.8897) (0.114,2.8914) (0.115,2.8931) (0.116,2.8948) (0.117,2.89651)
    (0.118,2.89822) (0.119,2.89993) (0.12,2.90165) (0.121,2.90337) (0.122,2.90509)
    (0.123,2.90682) (0.124,2.90855) (0.125,2.91029) (0.126,2.91202) (0.127,2.91377)
    (0.128,2.91551) (0.129,2.91726) (0.13,2.91902) (0.131,2.92077) (0.132,2.92253)
    (0.133,2.9243) (0.134,2.92607) (0.135,2.92784) (0.136,2.92962) (0.137,2.9314)
    (0.138,2.93318) (0.139,2.93497) (0.14,2.93676) (0.141,2.93855) (0.142,2.94035)
    (0.143,2.94216) (0.144,2.94396) (0.145,2.94577) (0.146,2.94759) (0.147,2.94941)
    (0.148,2.95123) (0.149,2.95305) (0.15,2.95488) (0.151,2.95672) (0.152,2.95856)
    (0.153,2.9604) (0.154,2.96225) (0.155,2.9641) (0.156,2.96595) (0.157,2.96781)
    (0.158,2.96967) (0.159,2.97154) (0.16,2.97341) (0.161,2.97528) (0.162,2.97716)
    (0.163,2.97904) (0.164,2.98093) (0.165,2.98282) (0.166,2.98472) (0.167,2.98662)
    (0.168,2.98852) (0.169,2.99043) (0.17,2.99234) (0.171,2.99426) (0.172,2.99618)
    (0.173,2.9981) (0.174,3.00003) (0.175,3.00197) (0.176,3.00391) (0.177,3.00585)
    (0.178,3.00779) (0.179,3.00975) (0.18,3.0117) (0.181,3.01366) (0.182,3.01563)
    (0.183,3.01759) (0.184,3.01957) (0.185,3.02155) (0.186,3.02353) (0.187,3.02551)
    (0.188,3.02751) (0.189,3.0295) (0.19,3.0315) (0.191,3.03351) (0.192,3.03552)
    (0.193,3.03753) (0.194,3.03955) (0.195,3.04157) (0.196,3.0436) (0.197,3.04563)
    (0.198,3.04767) (0.199,3.04971) (0.2,3.05176) (0.201,3.05381) (0.202,3.05587)
    (0.203,3.05793) (0.204,3.05999) (0.205,3.06206) (0.206,3.06414) (0.207,3.06622)
    (0.208,3.0683) (0.209,3.07039) (0.21,3.07249) (0.211,3.07459) (0.212,3.07669)
    (0.213,3.0788) (0.214,3.08091) (0.215,3.08303) (0.216,3.08516) (0.217,3.08729)
    (0.218,3.08942) (0.219,3.09156) (0.22,3.09371) (0.221,3.09586) (0.222,3.09801)
    (0.223,3.10017) (0.224,3.10234) (0.225,3.10451) (0.226,3.10668) (0.227,3.10886)
    (0.228,3.11105) (0.229,3.11324) (0.23,3.11543) (0.231,3.11764) (0.232,3.11984)
    (0.233,3.12205) (0.234,3.12427) (0.235,3.12649) (0.236,3.12872) (0.237,3.13096)
    (0.238,3.1332) (0.239,3.13544) (0.24,3.13769) (0.241,3.13995) (0.242,3.14221)
    (0.243,3.14447) (0.244,3.14674) (0.245,3.14902) (0.246,3.15131) (0.247,3.15359)
    (0.248,3.15589) (0.249,3.15819) (0.25,3.16049) (0.251,3.16281) (0.252,3.16512)
    (0.253,3.16745) (0.254,3.16977) (0.255,3.17211) (0.256,3.17445) (0.257,3.17679)
    (0.258,3.17915) (0.259,3.1815) (0.26,3.18387) (0.261,3.18624) (0.262,3.18861)
    (0.263,3.19099) (0.264,3.19338) (0.265,3.19577) (0.266,3.19817) (0.267,3.20058)
    (0.268,3.20299) (0.269,3.20541) (0.27,3.20783) (0.271,3.21026) (0.272,3.2127)
    (0.273,3.21514) (0.274,3.21759) (0.275,3.22004) (0.276,3.22251) (0.277,3.22497)
    (0.278,3.22745) (0.279,3.22993) (0.28,3.23241) (0.281,3.23491) (0.282,3.2374)
    (0.283,3.23991) (0.284,3.24242) (0.285,3.24494) (0.286,3.24747) (0.287,3.25)
    (0.288,3.25254) (0.289,3.25508) (0.29,3.25763) (0.291,3.26019) (0.292,3.26276)
    (0.293,3.26533) (0.294,3.26791) (0.295,3.27049) (0.296,3.27308) (0.297,3.27568)
    (0.298,3.27829) (0.299,3.2809) (0.3,3.28352) (0.301,3.28615) (0.302,3.28878)
    (0.303,3.29142) (0.304,3.29407) (0.305,3.29672) (0.306,3.29939) (0.307,3.30205)
    (0.308,3.30473) (0.309,3.30741) (0.31,3.3101) (0.311,3.3128) (0.312,3.31551)
    (0.313,3.31822) (0.314,3.32094) (0.315,3.32366) (0.316,3.3264) (0.317,3.32914)
    (0.318,3.33189) (0.319,3.33465) (0.32,3.33741) (0.321,3.34018) (0.322,3.34296)
    (0.323,3.34575) (0.324,3.34854) (0.325,3.35135) (0.326,3.35416) (0.327,3.35697)
    (0.328,3.3598) (0.329,3.36263) (0.33,3.36547) (0.331,3.36832) (0.332,3.37118)
    (0.333,3.37404) (0.334,3.37692) (0.335,3.3798) (0.336,3.38269) (0.337,3.38558)
    (0.338,3.38849) (0.339,3.3914) (0.34,3.39432) (0.341,3.39725) (0.342,3.40019)
    (0.343,3.40314) (0.344,3.40609) (0.345,3.40905) (0.346,3.41203) (0.347,3.41501)
    (0.348,3.41799) (0.349,3.42099) (0.35,3.424) (0.351,3.42701) (0.352,3.43003)
    (0.353,3.43306) (0.354,3.4361) (0.355,3.43915) (0.356,3.44221) (0.357,3.44528)
    (0.358,3.44835) (0.359,3.45144) (0.36,3.45453) (0.361,3.45763) (0.362,3.46074)
    (0.363,3.46386) (0.364,3.46699) (0.365,3.47013) (0.366,3.47328) (0.367,3.47644)
    (0.368,3.4796) (0.369,3.48278) (0.37,3.48596) (0.371,3.48916) (0.372,3.49236)
    (0.373,3.49558) (0.374,3.4988) (0.375,3.50203) (0.376,3.50527) (0.377,3.50852)
    (0.378,3.51179) (0.379,3.51506) (0.38,3.51834) (0.381,3.52163) (0.382,3.52493)
    (0.383,3.52824) (0.384,3.53156) (0.385,3.53489) (0.386,3.53823) (0.387,3.54159)
    (0.388,3.54495) (0.389,3.54832) (0.39,3.5517) (0.391,3.55509) (0.392,3.5585)
    (0.393,3.56191) (0.394,3.56533) (0.395,3.56877) (0.396,3.57221) (0.397,3.57567)
    (0.398,3.57913) (0.399,3.58261) (0.4,3.5861) (0.401,3.58959) (0.402,3.5931)
    (0.403,3.59662) (0.404,3.60015) (0.405,3.6037) (0.406,3.60725) (0.407,3.61081)
    (0.408,3.61439) (0.409,3.61798) (0.41,3.62158) (0.411,3.62519) (0.412,3.62881)
    (0.413,3.63244) (0.414,3.63608) (0.415,3.63974) (0.416,3.64341) (0.417,3.64709)
    (0.418,3.65078) (0.419,3.65448) (0.42,3.65819) (0.421,3.66192) (0.422,3.66566)
    (0.423,3.66941) (0.424,3.67317) (0.425,3.67695) (0.426,3.68073) (0.427,3.68453)
    (0.428,3.68835) (0.429,3.69217) (0.43,3.69601) (0.431,3.69986) (0.432,3.70372)
    (0.433,3.70759) (0.434,3.71148) (0.435,3.71538) (0.436,3.7193) (0.437,3.72322)
    (0.438,3.72716) (0.439,3.73111) (0.44,3.73508) (0.441,3.73906) (0.442,3.74305)
    (0.443,3.74706) (0.444,3.75108) (0.445,3.75511) (0.446,3.75915) (0.447,3.76321)
    (0.448,3.76729) (0.449,3.77138) (0.45,3.77548) (0.451,3.77959) (0.452,3.78372)
    (0.453,3.78786) (0.454,3.79202) (0.455,3.79619) (0.456,3.80038) (0.457,3.80458)
    (0.458,3.80879) (0.459,3.81302) (0.46,3.81727) (0.461,3.82153) (0.462,3.8258)
    (0.463,3.83009) (0.464,3.83439) (0.465,3.83871) (0.466,3.84304) (0.467,3.84739)
    (0.468,3.85176) (0.469,3.85614) (0.47,3.86053) (0.471,3.86494) (0.472,3.86937)
    (0.473,3.87381) (0.474,3.87827) (0.475,3.88274) (0.476,3.88723) (0.477,3.89173)
    (0.478,3.89626) (0.479,3.90079) (0.48,3.90535) (0.481,3.90992) (0.482,3.9145)
    (0.483,3.91911) (0.484,3.92373) (0.485,3.92837) (0.486,3.93302) (0.487,3.93769)
    (0.488,3.94238) (0.489,3.94708) (0.49,3.95181) (0.491,3.95655) (0.492,3.9613)
    (0.493,3.96608) (0.494,3.97087) (0.495,3.97568) (0.496,3.98051) (0.497,3.98535)
    (0.498,3.99022) (0.499,3.9951) (0.5,4.) (0.501,4.00492) (0.502,4.00986)
    (0.503,4.01481) (0.504,4.01979) (0.505,4.02478) (0.506,4.02979) (0.507,4.03482)
    (0.508,4.03987) (0.509,4.04494) (0.51,4.05003) (0.511,4.05514) (0.512,4.06027)
    (0.513,4.06542) (0.514,4.07059) (0.515,4.07578) (0.516,4.08098) (0.517,4.08621)
    (0.518,4.09146) (0.519,4.09673) (0.52,4.10202) (0.521,4.10733) (0.522,4.11266)
    (0.523,4.11801) (0.524,4.12339) (0.525,4.12878) (0.526,4.1342) (0.527,4.13963)
    (0.528,4.14509) (0.529,4.15057) (0.53,4.15607) (0.531,4.1616) (0.532,4.16715)
    (0.533,4.17271) (0.534,4.17831) (0.535,4.18392) (0.536,4.18956) (0.537,4.19521)
    (0.538,4.2009) (0.539,4.2066) (0.54,4.21233) (0.541,4.21808) (0.542,4.22386)
    (0.543,4.22966) (0.544,4.23548) (0.545,4.24133) (0.546,4.2472) (0.547,4.2531)
    (0.548,4.25902) (0.549,4.26496) (0.55,4.27093) (0.551,4.27693) (0.552,4.28295)
    (0.553,4.28899) (0.554,4.29506) (0.555,4.30116) (0.556,4.30728) (0.557,4.31343)
    (0.558,4.3196) (0.559,4.3258) (0.56,4.33203) (0.561,4.33829) (0.562,4.34457)
    (0.563,4.35087) (0.564,4.35721) (0.565,4.36357) (0.566,4.36996) (0.567,4.37638)
    (0.568,4.38282) (0.569,4.3893) (0.57,4.3958) (0.571,4.40233) (0.572,4.40889)
    (0.573,4.41548) (0.574,4.42209) (0.575,4.42874) (0.576,4.43542) (0.577,4.44212)
    (0.578,4.44886) (0.579,4.45562) (0.58,4.46242) (0.581,4.46924) (0.582,4.4761)
    (0.583,4.48299) (0.584,4.48991) (0.585,4.49686) (0.586,4.50384) (0.587,4.51085)
    (0.588,4.5179) (0.589,4.52498) (0.59,4.53209) (0.591,4.53923) (0.592,4.54641)
    (0.593,4.55362) (0.594,4.56086) (0.595,4.56814) (0.596,4.57545) (0.597,4.5828)
    (0.598,4.59018) (0.599,4.59759) (0.6,4.60504) (0.601,4.61252) (0.602,4.62004)
    (0.603,4.6276) (0.604,4.63519) (0.605,4.64282) (0.606,4.65048) (0.607,4.65818)
    (0.608,4.66592) (0.609,4.6737) (0.61,4.68151) (0.611,4.68936) (0.612,4.69725)
    (0.613,4.70518) (0.614,4.71314) (0.615,4.72115) (0.616,4.72919) (0.617,4.73728)
    (0.618,4.7454) (0.619,4.75356) (0.62,4.76177) (0.621,4.77002) (0.622,4.7783)
    (0.623,4.78663) (0.624,4.795) (0.625,4.80341) (0.626,4.81187) (0.627,4.82037)
    (0.628,4.82891) (0.629,4.83749) (0.63,4.84612) (0.631,4.85479) (0.632,4.86351)
    (0.633,4.87227) (0.634,4.88108) (0.635,4.88993) (0.636,4.89883) (0.637,4.90777)
    (0.638,4.91676) (0.639,4.9258) (0.64,4.93489) (0.641,4.94402) (0.642,4.9532)
    (0.643,4.96244) (0.644,4.97171) (0.645,4.98104) (0.646,4.99042) (0.647,4.99985)
    (0.648,5.00933) (0.649,5.01886) (0.65,5.02845) (0.651,5.03808) (0.652,5.04777)
    (0.653,5.05751) (0.654,5.0673) (0.655,5.07715) (0.656,5.08705) (0.657,5.097)
    (0.658,5.10701) (0.659,5.11708) (0.66,5.1272) (0.661,5.13738) (0.662,5.14762)
    (0.663,5.15791) (0.664,5.16826) (0.665,5.17867) (0.666,5.18914) (0.667,5.19967)
    (0.668,5.21026) (0.669,5.22091) (0.67,5.23162) (0.671,5.24239) (0.672,5.25322)
    (0.673,5.26412) (0.674,5.27508) (0.675,5.2861) (0.676,5.29719) (0.677,5.30835)
    (0.678,5.31957) (0.679,5.33085) (0.68,5.34221) (0.681,5.35363) (0.682,5.36511)
    (0.683,5.37667) (0.684,5.3883) (0.685,5.39999) (0.686,5.41176) (0.687,5.4236)
    (0.688,5.43551) (0.689,5.44749) (0.69,5.45955) (0.691,5.47168) (0.692,5.48389)
    (0.693,5.49617) (0.694,5.50853) (0.695,5.52096) (0.696,5.53347) (0.697,5.54606)
    (0.698,5.55873) (0.699,5.57148) (0.7,5.58431) (0.701,5.59722) (0.702,5.61022)
    (0.703,5.62329) (0.704,5.63646) (0.705,5.6497) (0.706,5.66303) (0.707,5.67645)
    (0.708,5.68996) (0.709,5.70355) (0.71,5.71723) (0.711,5.73101) (0.712,5.74487)
    (0.713,5.75883) (0.714,5.77287) (0.715,5.78702) (0.716,5.80125) (0.717,5.81559)
    (0.718,5.83002) (0.719,5.84454) (0.72,5.85917) (0.721,5.8739) (0.722,5.88872)
    (0.723,5.90365) (0.724,5.91868) (0.725,5.93382) (0.726,5.94906) (0.727,5.96441)
    (0.728,5.97986) (0.729,5.99542) (0.73,6.0111) (0.731,6.02688) (0.732,6.04278)
    (0.733,6.05879) (0.734,6.07491) (0.735,6.09115) (0.736,6.10751) (0.737,6.12399)
    (0.738,6.14058) (0.739,6.1573) (0.74,6.17414) (0.741,6.19111) (0.742,6.2082)
    (0.743,6.22541) (0.744,6.24276) (0.745,6.26023) (0.746,6.27784) (0.747,6.29558)
    (0.748,6.31345) (0.749,6.33146) (0.75,6.3496) (0.751,6.36789) (0.752,6.38632)
    (0.753,6.40489) (0.754,6.4236) (0.755,6.44246) (0.756,6.46146) (0.757,6.48062)
    (0.758,6.49993) (0.759,6.51939) (0.76,6.539) (0.761,6.55877) (0.762,6.57871)
    (0.763,6.5988) (0.764,6.61905) (0.765,6.63947) (0.766,6.66006) (0.767,6.68081)
    (0.768,6.70174) (0.769,6.72284) (0.77,6.74411) (0.771,6.76557) (0.772,6.7872)
    (0.773,6.80901) (0.774,6.83101) (0.775,6.8532) (0.776,6.87558) (0.777,6.89815)
    (0.778,6.92091) (0.779,6.94387) (0.78,6.96703) (0.781,6.9904) (0.782,7.01396)
    (0.783,7.03774) (0.784,7.06173) (0.785,7.08593) (0.786,7.11035) (0.787,7.13499)
    (0.788,7.15985) (0.789,7.18493) (0.79,7.21025) (0.791,7.2358) (0.792,7.26158)
    (0.793,7.2876) (0.794,7.31387) (0.795,7.34038) (0.796,7.36714) (0.797,7.39415)
    (0.798,7.42142) (0.799,7.44895) (0.8,7.47674) (0.801,7.50481) (0.802,7.53314)
    (0.803,7.56175) (0.804,7.59064) (0.805,7.61982) (0.806,7.64928) (0.807,7.67904)
    (0.808,7.7091) (0.809,7.73945) (0.81,7.77012) (0.811,7.8011) (0.812,7.83239)
    (0.813,7.864) (0.814,7.89595) (0.815,7.92822) (0.816,7.96083) (0.817,7.99379)
    (0.818,8.02709) (0.819,8.06075) (0.82,8.09477) (0.821,8.12915) (0.822,8.1639)
    (0.823,8.19903) (0.824,8.23455) (0.825,8.27046) (0.826,8.30676) (0.827,8.34347)
    (0.828,8.38059) (0.829,8.41813) (0.83,8.4561) (0.831,8.49449) (0.832,8.53333)
    (0.833,8.57262) (0.834,8.61236) (0.835,8.65257) (0.836,8.69325) (0.837,8.73441)
    (0.838,8.77606) (0.839,8.81821) (0.84,8.86087) (0.841,8.90404) (0.842,8.94775)
    (0.843,8.99199) (0.844,9.03678) (0.845,9.08212) (0.846,9.12804) (0.847,9.17453)
    (0.848,9.22162) (0.849,9.26931) (0.85,9.31761) (0.851,9.36654) (0.852,9.4161)
    (0.853,9.46632) (0.854,9.51721) (0.855,9.56877) (0.856,9.62103) (0.857,9.67399)
    (0.858,9.72767) (0.859,9.78209) (0.86,9.83727) (0.861,9.89321) (0.862,9.94993)
    (0.863,10.0075) (0.864,10.0658) (0.865,10.125) (0.866,10.185) (0.867,10.2459)
    (0.868,10.3077) (0.869,10.3705) (0.87,10.4341) (0.871,10.4987) (0.872,10.5643)
    (0.873,10.6309) (0.874,10.6985) (0.875,10.7672) (0.876,10.8369) (0.877,10.9078)
    (0.878,10.9797) (0.879,11.0529) (0.88,11.1272) (0.881,11.2027) (0.882,11.2795)
    (0.883,11.3575) (0.884,11.4368) (0.885,11.5175) (0.886,11.5996) (0.887,11.683)
    (0.888,11.7679) (0.889,11.8543) (0.89,11.9423) (0.891,12.0318) (0.892,12.1229)
    (0.893,12.2156) (0.894,12.3101) (0.895,12.4063) (0.896,12.5043) (0.897,12.6042)
    (0.898,12.706) (0.899,12.8097) (0.9,12.9155) (0.901,13.0233) (0.902,13.1333)
    (0.903,13.2455) (0.904,13.36) (0.905,13.4768) (0.906,13.5961) (0.907,13.7178)
    (0.908,13.8421) (0.909,13.9691) (0.91,14.0989) (0.911,14.2315) (0.912,14.367)
    (0.913,14.5056) (0.914,14.6473) (0.915,14.7922) (0.916,14.9406) (0.917,15.0924)
    (0.918,15.2479) (0.919,15.4071) (0.92,15.5702) (0.921,15.7373) (0.922,15.9086)
    (0.923,16.0843) (0.924,16.2645) (0.925,16.4494) (0.926,16.6392) (0.927,16.8341)
    (0.928,17.0342) (0.929,17.2399) (0.93,17.4514) (0.931,17.6688) (0.932,17.8926)
    (0.933,18.1228) (0.934,18.36) (0.935,18.6042) (0.936,18.856) (0.937,19.1156)
    (0.938,19.3834) (0.939,19.6598) (0.94,19.9452) (0.941,20.2402) (0.942,20.5451)
    (0.943,20.8606) (0.944,21.1871) (0.945,21.5253) (0.946,21.8758) (0.947,22.2393)
    (0.948,22.6166) (0.949,23.0084) (0.95,23.4156) (0.951,23.8392) (0.952,24.2801)
    (0.953,24.7395) (0.954,25.2186) (0.955,25.7187) (0.956,26.2411) (0.957,26.7875)
    (0.958,27.3596) (0.959,27.9591) (0.96,28.5882) (0.961,29.2492) (0.962,29.9445)
    };
\end{axis}
\end{tikzpicture}

%% file: regularity_dispersion.tex
\begin{tikzpicture}
\newcommand{\pathfeasible}{
(0.,1.) (0.001,1.001) (0.002,1.00201) (0.003,1.00301) (0.004,1.00402)
(0.005,1.00504) (0.006,1.00605) (0.007,1.00707) (0.008,1.0081) (0.009,1.00912)
(0.01,1.01015) (0.011,1.01118) (0.012,1.01222) (0.013,1.01326) (0.014,1.0143)
(0.015,1.01535) (0.016,1.01639) (0.017,1.01745) (0.018,1.0185) (0.019,1.01956)
(0.02,1.02062) (0.021,1.02169) (0.022,1.02275) (0.023,1.02383) (0.024,1.0249)
(0.025,1.02598) (0.026,1.02706) (0.027,1.02815) (0.028,1.02923) (0.029,1.03033)
(0.03,1.03142) (0.031,1.03252) (0.032,1.03362) (0.033,1.03473) (0.034,1.03584)
(0.035,1.03695) (0.036,1.03807) (0.037,1.03919) (0.038,1.04031) (0.039,1.04144)
(0.04,1.04257) (0.041,1.04371) (0.042,1.04485) (0.043,1.04599) (0.044,1.04713)
(0.045,1.04828) (0.046,1.04944) (0.047,1.0506) (0.048,1.05176) (0.049,1.05292)
(0.05,1.05409) (0.051,1.05527) (0.052,1.05644) (0.053,1.05762) (0.054,1.05881)
(0.055,1.06) (0.056,1.06119) (0.057,1.06239) (0.058,1.06359) (0.059,1.06479)
(0.06,1.066) (0.061,1.06722) (0.062,1.06843) (0.063,1.06966) (0.064,1.07088)
(0.065,1.07211) (0.066,1.07335) (0.067,1.07459) (0.068,1.07583) (0.069,1.07708)
(0.07,1.07833) (0.071,1.07958) (0.072,1.08084) (0.073,1.08211) (0.074,1.08338)
(0.075,1.08465) (0.076,1.08593) (0.077,1.08721) (0.078,1.0885) (0.079,1.08979)
(0.08,1.09109) (0.081,1.09239) (0.082,1.0937) (0.083,1.09501) (0.084,1.09632)
(0.085,1.09764) (0.086,1.09897) (0.087,1.1003) (0.088,1.10163) (0.089,1.10297)
(0.09,1.10432) (0.091,1.10566) (0.092,1.10702) (0.093,1.10838) (0.094,1.10974)
(0.095,1.11111) (0.096,1.11249) (0.097,1.11386) (0.098,1.11525) (0.099,1.11664)
(0.1,1.11803) (0.101,1.11943) (0.102,1.12084) (0.103,1.12225) (0.104,1.12367)
(0.105,1.12509) (0.106,1.12651) (0.107,1.12795) (0.108,1.12938) (0.109,1.13083)
(0.11,1.13228) (0.111,1.13373) (0.112,1.13519) (0.113,1.13666) (0.114,1.13813)
(0.115,1.13961) (0.116,1.14109) (0.117,1.14258) (0.118,1.14407) (0.119,1.14557)
(0.12,1.14708) (0.121,1.14859) (0.122,1.15011) (0.123,1.15163) (0.124,1.15316)
(0.125,1.1547) (0.126,1.15624) (0.127,1.15779) (0.128,1.15935) (0.129,1.16091)
(0.13,1.16248) (0.131,1.16405) (0.132,1.16563) (0.133,1.16722) (0.134,1.16881)
(0.135,1.17041) (0.136,1.17202) (0.137,1.17363) (0.138,1.17525) (0.139,1.17688)
(0.14,1.17851) (0.141,1.18015) (0.142,1.1818) (0.143,1.18345) (0.144,1.18511)
(0.145,1.18678) (0.146,1.18846) (0.147,1.19014) (0.148,1.19183) (0.149,1.19352)
(0.15,1.19523) (0.151,1.19694) (0.152,1.19866) (0.153,1.20038) (0.154,1.20212)
(0.155,1.20386) (0.156,1.20561) (0.157,1.20736) (0.158,1.20913) (0.159,1.2109)
(0.16,1.21268) (0.161,1.21447) (0.162,1.21626) (0.163,1.21806) (0.164,1.21988)
(0.165,1.22169) (0.166,1.22352) (0.167,1.22536) (0.168,1.2272) (0.169,1.22905)
(0.17,1.23091) (0.171,1.23278) (0.172,1.23466) (0.173,1.23655) (0.174,1.23844)
(0.175,1.24035) (0.176,1.24226) (0.177,1.24418) (0.178,1.24611) (0.179,1.24805)
(0.18,1.25) (0.181,1.25196) (0.182,1.25392) (0.183,1.2559) (0.184,1.25789)
(0.185,1.25988) (0.186,1.26189) (0.187,1.2639) (0.188,1.26592) (0.189,1.26796)
(0.19,1.27) (0.191,1.27205) (0.192,1.27412) (0.193,1.27619) (0.194,1.27827)
(0.195,1.28037) (0.196,1.28247) (0.197,1.28459) (0.198,1.28671) (0.199,1.28885)
(0.2,1.29099) (0.201,1.29315) (0.202,1.29532) (0.203,1.2975) (0.204,1.29969)
(0.205,1.30189) (0.206,1.3041) (0.207,1.30632) (0.208,1.30856) (0.209,1.31081)
(0.21,1.31306) (0.211,1.31533) (0.212,1.31762) (0.213,1.31991) (0.214,1.32221)
(0.215,1.32453) (0.216,1.32686) (0.217,1.3292) (0.218,1.33156) (0.219,1.33393)
(0.22,1.33631) (0.221,1.3387) (0.222,1.3411) (0.223,1.34352) (0.224,1.34595)
(0.225,1.3484) (0.226,1.35086) (0.227,1.35333) (0.228,1.35582) (0.229,1.35831)
(0.23,1.36083) (0.231,1.36335) (0.232,1.3659) (0.233,1.36845) (0.234,1.37102)
(0.235,1.37361) (0.236,1.3762) (0.237,1.37882) (0.238,1.38145) (0.239,1.38409)
(0.24,1.38675) (0.241,1.38943) (0.242,1.39212) (0.243,1.39482) (0.244,1.39754)
(0.245,1.40028) (0.246,1.40303) (0.247,1.4058) (0.248,1.40859) (0.249,1.41139)
(0.25,1.41421) (0.251,1.41705) (0.252,1.4199) (0.253,1.42278) (0.254,1.42566)
(0.255,1.42857) (0.256,1.4315) (0.257,1.43444) (0.258,1.4374) (0.259,1.44038)
(0.26,1.44338) (0.261,1.44639) (0.262,1.44943) (0.263,1.45248) (0.264,1.45556)
(0.265,1.45865) (0.266,1.46176) (0.267,1.4649) (0.268,1.46805) (0.269,1.47122)
(0.27,1.47442) (0.271,1.47764) (0.272,1.48087) (0.273,1.48413) (0.274,1.48741)
(0.275,1.49071) (0.276,1.49404) (0.277,1.49738) (0.278,1.50075) (0.279,1.50414)
(0.28,1.50756) (0.281,1.51099) (0.282,1.51446) (0.283,1.51794) (0.284,1.52145)
(0.285,1.52499) (0.286,1.52854) (0.287,1.53213) (0.288,1.53574) (0.289,1.53937)
(0.29,1.54303) (0.291,1.54672) (0.292,1.55043) (0.293,1.55417) (0.294,1.55794)
(0.295,1.56174) (0.296,1.56556) (0.297,1.56941) (0.298,1.57329) (0.299,1.5772)
(0.3,1.58114) (0.301,1.58511) (0.302,1.5891) (0.303,1.59313) (0.304,1.59719)
(0.305,1.60128) (0.306,1.6054) (0.307,1.60956) (0.308,1.61374) (0.309,1.61796)
(0.31,1.62221) (0.311,1.6265) (0.312,1.63082) (0.313,1.63517) (0.314,1.63956)
(0.315,1.64399) (0.316,1.64845) (0.317,1.65295) (0.318,1.65748) (0.319,1.66206)
(0.32,1.66667) (0.321,1.67132) (0.322,1.676) (0.323,1.68073) (0.324,1.6855)
(0.325,1.69031) (0.326,1.69516) (0.327,1.70005) (0.328,1.70499) (0.329,1.70996)
(0.33,1.71499) (0.331,1.72005) (0.332,1.72516) (0.333,1.73032) (0.334,1.73553)
(0.335,1.74078) (0.336,1.74608) (0.337,1.75142) (0.338,1.75682) (0.339,1.76227)
(0.34,1.76777) (0.341,1.77332) (0.342,1.77892) (0.343,1.78458) (0.344,1.79029)
(0.345,1.79605) (0.346,1.80187) (0.347,1.80775) (0.348,1.81369) (0.349,1.81969)
(0.35,1.82574) (0.351,1.83186) (0.352,1.83804) (0.353,1.84428) (0.354,1.85058)
(0.355,1.85695) (0.356,1.86339) (0.357,1.86989) (0.358,1.87647) (0.359,1.88311)
(0.36,1.88982) (0.361,1.89661) (0.362,1.90347) (0.363,1.9104) (0.364,1.91741)
(0.365,1.9245) (0.366,1.93167) (0.367,1.93892) (0.368,1.94625) (0.369,1.95366)
(0.37,1.96116) (0.371,1.96875) (0.372,1.97642) (0.373,1.98419) (0.374,1.99205)
(0.375,2.) (0.376,2.00805) (0.377,2.01619) (0.378,2.02444) (0.379,2.03279)
(0.38,2.04124) (0.381,2.0498) (0.382,2.05847) (0.383,2.06725) (0.384,2.07614)
(0.385,2.08514) (0.386,2.09427) (0.387,2.10352) (0.388,2.11289) (0.389,2.12238)
(0.39,2.13201) (0.391,2.14176) (0.392,2.15166) (0.393,2.16169) (0.394,2.17186)
(0.395,2.18218) (0.396,2.19265) (0.397,2.20326) (0.398,2.21404) (0.399,2.22497)
(0.4,2.23607) (0.401,2.24733) (0.402,2.25877) (0.403,2.27038) (0.404,2.28218)
(0.405,2.29416) (0.406,2.30633) (0.407,2.31869) (0.408,2.33126) (0.409,2.34404)
(0.41,2.35702) (0.411,2.37023) (0.412,2.38366) (0.413,2.39732) (0.414,2.41121)
(0.415,2.42536) (0.416,2.43975) (0.417,2.4544) (0.418,2.46932) (0.419,2.48452)
(0.42,2.5) (0.421,2.51577) (0.422,2.53185) (0.423,2.54824) (0.424,2.56495)
(0.425,2.58199) (0.426,2.59938) (0.427,2.61712) (0.428,2.63523) (0.429,2.65372)
(0.43,2.67261) (0.431,2.69191) (0.432,2.71163) (0.433,2.73179) (0.434,2.75241)
(0.435,2.7735) (0.436,2.79508) (0.437,2.81718) (0.438,2.83981) (0.439,2.86299)
(0.44,2.88675) (0.441,2.91111) (0.442,2.9361) (0.443,2.96174) (0.444,2.98807)
(0.445,3.01511) (0.446,3.0429) (0.447,3.07148) (0.448,3.10087) (0.449,3.13112)
(0.45,3.16228) (0.451,3.19438) (0.452,3.22749) (0.453,3.26164) (0.454,3.2969)
(0.455,3.33333) (0.456,3.371) (0.457,3.40997) (0.458,3.45033) (0.459,3.49215)
(0.46,3.53553) (0.461,3.58057) (0.462,3.62738) (0.463,3.67607) (0.464,3.72678)
(0.465,3.77964) (0.466,3.83482) (0.467,3.89249) (0.468,3.95285) (0.469,4.0161)
(0.47,4.08248) (0.471,4.15227) (0.472,4.22577) (0.473,4.30331) (0.474,4.38529)
(0.475,4.47214) (0.476,4.56435) (0.477,4.66252) (0.478,4.76731) (0.479,4.8795)
(0.48,5.) (0.481,5.12989) (0.482,5.27046) (0.483,5.42326) (0.484,5.59017)
(0.485,5.7735) (0.486,5.97614) (0.487,6.20174) (0.488,6.45497) (0.489,6.742)
(0.49,7.07107) (0.491,7.45356) (0.492,7.90569) (0.493,8.45154) (0.494,9.12871)
(0.495,10.) (0.496,11.1803) (0.497,12.9099) (0.498,15.8114) (0.499,22.3607)
}
\newcommand{\pathregulardispersion}{
(0.,0.611259) (0.001,0.611813) (0.002,0.612369) (0.003,0.612926) (0.004,0.613484)
(0.005,0.614043) (0.006,0.614603) (0.007,0.615164) (0.008,0.615727) (0.009,0.616291)
(0.01,0.616856) (0.011,0.617422) (0.012,0.617989) (0.013,0.618557) (0.014,0.619127)
(0.015,0.619697) (0.016,0.620269) (0.017,0.620842) (0.018,0.621416) (0.019,0.621992)
(0.02,0.622568) (0.021,0.623146) (0.022,0.623725) (0.023,0.624306) (0.024,0.624887)
(0.025,0.62547) (0.026,0.626054) (0.027,0.626639) (0.028,0.627225) (0.029,0.627813)
(0.03,0.628401) (0.031,0.628991) (0.032,0.629583) (0.033,0.630175) (0.034,0.630769)
(0.035,0.631364) (0.036,0.63196) (0.037,0.632558) (0.038,0.633157) (0.039,0.633757)
(0.04,0.634358) (0.041,0.634961) (0.042,0.635565) (0.043,0.63617) (0.044,0.636777)
(0.045,0.637385) (0.046,0.637994) (0.047,0.638604) (0.048,0.639216) (0.049,0.639829)
(0.05,0.640444) (0.051,0.64106) (0.052,0.641677) (0.053,0.642295) (0.054,0.642915)
(0.055,0.643536) (0.056,0.644159) (0.057,0.644782) (0.058,0.645408) (0.059,0.646034)
(0.06,0.646662) (0.061,0.647292) (0.062,0.647922) (0.063,0.648554) (0.064,0.649188)
(0.065,0.649823) (0.066,0.650459) (0.067,0.651097) (0.068,0.651736) (0.069,0.652376)
(0.07,0.653018) (0.071,0.653661) (0.072,0.654306) (0.073,0.654952) (0.074,0.6556)
(0.075,0.656249) (0.076,0.6569) (0.077,0.657552) (0.078,0.658205) (0.079,0.65886)
(0.08,0.659517) (0.081,0.660174) (0.082,0.660834) (0.083,0.661495) (0.084,0.662157)
(0.085,0.662821) (0.086,0.663486) (0.087,0.664153) (0.088,0.664821) (0.089,0.665491)
(0.09,0.666163) (0.091,0.666836) (0.092,0.66751) (0.093,0.668186) (0.094,0.668864)
(0.095,0.669543) (0.096,0.670224) (0.097,0.670906) (0.098,0.67159) (0.099,0.672275)
(0.1,0.672962) (0.101,0.673651) (0.102,0.674341) (0.103,0.675032) (0.104,0.675726)
(0.105,0.676421) (0.106,0.677117) (0.107,0.677816) (0.108,0.678516) (0.109,0.679217)
(0.11,0.67992) (0.111,0.680625) (0.112,0.681331) (0.113,0.682039) (0.114,0.682749)
(0.115,0.683461) (0.116,0.684174) (0.117,0.684889) (0.118,0.685605) (0.119,0.686323)
(0.12,0.687043) (0.121,0.687765) (0.122,0.688488) (0.123,0.689213) (0.124,0.68994)
(0.125,0.690669) (0.126,0.691399) (0.127,0.692131) (0.128,0.692865) (0.129,0.6936)
(0.13,0.694338) (0.131,0.695077) (0.132,0.695818) (0.133,0.69656) (0.134,0.697305)
(0.135,0.698051) (0.136,0.698799) (0.137,0.699549) (0.138,0.700301) (0.139,0.701054)
(0.14,0.70181) (0.141,0.702567) (0.142,0.703326) (0.143,0.704087) (0.144,0.70485)
(0.145,0.705615) (0.146,0.706381) (0.147,0.70715) (0.148,0.70792) (0.149,0.708693)
(0.15,0.709467) (0.151,0.710243) (0.152,0.711021) (0.153,0.711801) (0.154,0.712583)
(0.155,0.713367) (0.156,0.714153) (0.157,0.714941) (0.158,0.715731) (0.159,0.716523)
(0.16,0.717317) (0.161,0.718112) (0.162,0.71891) (0.163,0.71971) (0.164,0.720512)
(0.165,0.721316) (0.166,0.722122) (0.167,0.72293) (0.168,0.72374) (0.169,0.724552)
(0.17,0.725366) (0.171,0.726183) (0.172,0.727001) (0.173,0.727822) (0.174,0.728644)
(0.175,0.729469) (0.176,0.730296) (0.177,0.731125) (0.178,0.731956) (0.179,0.732789)
(0.18,0.733625) (0.181,0.734462) (0.182,0.735302) (0.183,0.736144) (0.184,0.736988)
(0.185,0.737835) (0.186,0.738684) (0.187,0.739534) (0.188,0.740387) (0.189,0.741243)
(0.19,0.7421) (0.191,0.74296) (0.192,0.743822) (0.193,0.744687) (0.194,0.745554)
(0.195,0.746423) (0.196,0.747294) (0.197,0.748168) (0.198,0.749044) (0.199,0.749922)
(0.2,0.750803) (0.201,0.751686) (0.202,0.752571) (0.203,0.753459) (0.204,0.754349)
(0.205,0.755242) (0.206,0.756137) (0.207,0.757034) (0.208,0.757934) (0.209,0.758837)
(0.21,0.759741) (0.211,0.760649) (0.212,0.761558) (0.213,0.762471) (0.214,0.763385)
(0.215,0.764303) (0.216,0.765222) (0.217,0.766145) (0.218,0.76707) (0.219,0.767997)
(0.22,0.768927) (0.221,0.76986) (0.222,0.770795) (0.223,0.771732) (0.224,0.772673)
(0.225,0.773616) (0.226,0.774561) (0.227,0.77551) (0.228,0.776461) (0.229,0.777414)
(0.23,0.778371) (0.231,0.77933) (0.232,0.780291) (0.233,0.781256) (0.234,0.782223)
(0.235,0.783193) (0.236,0.784165) (0.237,0.785141) (0.238,0.786119) (0.239,0.7871)
(0.24,0.788084) (0.241,0.78907) (0.242,0.79006) (0.243,0.791052) (0.244,0.792047)
(0.245,0.793045) (0.246,0.794046) (0.247,0.79505) (0.248,0.796057) (0.249,0.797067)
(0.25,0.798079) (0.251,0.799095) (0.252,0.800113) (0.253,0.801135) (0.254,0.802159)
(0.255,0.803187) (0.256,0.804217) (0.257,0.805251) (0.258,0.806288) (0.259,0.807327)
(0.26,0.80837) (0.261,0.809416) (0.262,0.810465) (0.263,0.811517) (0.264,0.812572)
(0.265,0.813631) (0.266,0.814692) (0.267,0.815757) (0.268,0.816825) (0.269,0.817896)
(0.27,0.818971) (0.271,0.820049) (0.272,0.821129) (0.273,0.822214) (0.274,0.823301)
(0.275,0.824392) (0.276,0.825486) (0.277,0.826584) (0.278,0.827685) (0.279,0.828789)
(0.28,0.829896) (0.281,0.831007) (0.282,0.832122) (0.283,0.83324) (0.284,0.834361)
(0.285,0.835486) (0.286,0.836614) (0.287,0.837746) (0.288,0.838882) (0.289,0.840021)
(0.29,0.841163) (0.291,0.842309) (0.292,0.843459) (0.293,0.844612) (0.294,0.845769)
(0.295,0.84693) (0.296,0.848094) (0.297,0.849262) (0.298,0.850433) (0.299,0.851609)
(0.3,0.852788) (0.301,0.853971) (0.302,0.855158) (0.303,0.856348) (0.304,0.857542)
(0.305,0.85874) (0.306,0.859943) (0.307,0.861148) (0.308,0.862358) (0.309,0.863572)
(0.31,0.86479) (0.311,0.866011) (0.312,0.867237) (0.313,0.868467) (0.314,0.8697)
(0.315,0.870938) (0.316,0.87218) (0.317,0.873426) (0.318,0.874676) (0.319,0.87593)
(0.32,0.877188) (0.321,0.878451) (0.322,0.879717) (0.323,0.880988) (0.324,0.882263)
(0.325,0.883543) (0.326,0.884826) (0.327,0.886114) (0.328,0.887407) (0.329,0.888703)
(0.33,0.890004) (0.331,0.89131) (0.332,0.89262) (0.333,0.893934) (0.334,0.895253)
(0.335,0.896576) (0.336,0.897904) (0.337,0.899237) (0.338,0.900574) (0.339,0.901915)
(0.34,0.903262) (0.341,0.904612) (0.342,0.905968) (0.343,0.907328) (0.344,0.908693)
(0.345,0.910063) (0.346,0.911438) (0.347,0.912817) (0.348,0.914201) (0.349,0.91559)
(0.35,0.916984) (0.351,0.918383) (0.352,0.919787) (0.353,0.921195) (0.354,0.922609)
(0.355,0.924028) (0.356,0.925452) (0.357,0.926881) (0.358,0.928315) (0.359,0.929754)
(0.36,0.931198) (0.361,0.932648) (0.362,0.934103) (0.363,0.935563) (0.364,0.937028)
(0.365,0.938499) (0.366,0.939975) (0.367,0.941457) (0.368,0.942943) (0.369,0.944436)
(0.37,0.945933) (0.371,0.947437) (0.372,0.948946) (0.373,0.95046) (0.374,0.95198)
(0.375,0.953506) (0.376,0.955037) (0.377,0.956574) (0.378,0.958117) (0.379,0.959665)
(0.38,0.96122) (0.381,0.96278) (0.382,0.964346) (0.383,0.965918) (0.384,0.967496)
(0.385,0.96908) (0.386,0.97067) (0.387,0.972266) (0.388,0.973868) (0.389,0.975477)
(0.39,0.977091) (0.391,0.978712) (0.392,0.980339) (0.393,0.981972) (0.394,0.983611)
(0.395,0.985257) (0.396,0.98691) (0.397,0.988568) (0.398,0.990233) (0.399,0.991905)
(0.4,0.993583) (0.401,0.995268) (0.402,0.99696) (0.403,0.998658) (0.404,1.00036)
(0.405,1.00207) (0.406,1.00379) (0.407,1.00552) (0.408,1.00725) (0.409,1.00899)
(0.41,1.01074) (0.411,1.01249) (0.412,1.01425) (0.413,1.01602) (0.414,1.01779)
(0.415,1.01957) (0.416,1.02136) (0.417,1.02316) (0.418,1.02496) (0.419,1.02677)
(0.42,1.02859) (0.421,1.03042) (0.422,1.03225) (0.423,1.03409) (0.424,1.03594)
(0.425,1.0378) (0.426,1.03966) (0.427,1.04154) (0.428,1.04342) (0.429,1.0453)
(0.43,1.0472) (0.431,1.0491) (0.432,1.05101) (0.433,1.05293) (0.434,1.05486)
(0.435,1.0568) (0.436,1.05874) (0.437,1.0607) (0.438,1.06266) (0.439,1.06463)
(0.44,1.06661) (0.441,1.06859) (0.442,1.07059) (0.443,1.07259) (0.444,1.07461)
(0.445,1.07663) (0.446,1.07866) (0.447,1.0807) (0.448,1.08275) (0.449,1.0848)
(0.45,1.08687) (0.451,1.08895) (0.452,1.09103) (0.453,1.09313) (0.454,1.09523)
(0.455,1.09734) (0.456,1.09946) (0.457,1.1016) (0.458,1.10374) (0.459,1.10589)
(0.46,1.10805) (0.461,1.11022) (0.462,1.1124) (0.463,1.1146) (0.464,1.1168)
(0.465,1.11901) (0.466,1.12123) (0.467,1.12346) (0.468,1.1257) (0.469,1.12796)
(0.47,1.13022) (0.471,1.13249) (0.472,1.13478) (0.473,1.13707) (0.474,1.13938)
(0.475,1.14169) (0.476,1.14402) (0.477,1.14636) (0.478,1.14871) (0.479,1.15107)
(0.48,1.15344) (0.481,1.15583) (0.482,1.15822) (0.483,1.16063) (0.484,1.16305)
(0.485,1.16548) (0.486,1.16792) (0.487,1.17038) (0.488,1.17284) (0.489,1.17532)
(0.49,1.17781) (0.491,1.18031) (0.492,1.18283) (0.493,1.18536) (0.494,1.1879)
(0.495,1.19045) (0.496,1.19302) (0.497,1.19559) (0.498,1.19819) (0.499,1.20079)
(0.5,1.20341) (0.501,1.20604) (0.502,1.20868) (0.503,1.21134) (0.504,1.21401)
(0.505,1.2167) (0.506,1.2194) (0.507,1.22211) (0.508,1.22484) (0.509,1.22758)
(0.51,1.23034) (0.511,1.23311) (0.512,1.23589) (0.513,1.23869) (0.514,1.24151)
(0.515,1.24433) (0.516,1.24718) (0.517,1.25004) (0.518,1.25291) (0.519,1.2558)
(0.52,1.25871) (0.521,1.26163) (0.522,1.26457) (0.523,1.26752) (0.524,1.27049)
(0.525,1.27347) (0.526,1.27647) (0.527,1.27949) (0.528,1.28253) (0.529,1.28558)
(0.53,1.28864) (0.531,1.29173) (0.532,1.29483) (0.533,1.29795) (0.534,1.30109)
(0.535,1.30424) (0.536,1.30742) (0.537,1.31061) (0.538,1.31381) (0.539,1.31704)
(0.54,1.32029) (0.541,1.32355) (0.542,1.32683) (0.543,1.33013) (0.544,1.33346)
(0.545,1.3368) (0.546,1.34015) (0.547,1.34353) (0.548,1.34693) (0.549,1.35035)
(0.55,1.35379) (0.551,1.35725) (0.552,1.36073) (0.553,1.36423) (0.554,1.36775)
(0.555,1.37129) (0.556,1.37486) (0.557,1.37844) (0.558,1.38205) (0.559,1.38568)
(0.56,1.38933) (0.561,1.393) (0.562,1.3967) (0.563,1.40042) (0.564,1.40416)
(0.565,1.40793) (0.566,1.41171) (0.567,1.41553) (0.568,1.41936) (0.569,1.42322)
(0.57,1.42711) (0.571,1.43101) (0.572,1.43495) (0.573,1.43891) (0.574,1.44289)
(0.575,1.4469) (0.576,1.45094) (0.577,1.455) (0.578,1.45908) (0.579,1.4632)
(0.58,1.46734) (0.581,1.47151) (0.582,1.4757) (0.583,1.47993) (0.584,1.48418)
(0.585,1.48846) (0.586,1.49277) (0.587,1.4971) (0.588,1.50147) (0.589,1.50587)
(0.59,1.51029) (0.591,1.51475) (0.592,1.51923) (0.593,1.52375) (0.594,1.52829)
(0.595,1.53287) (0.596,1.53748) (0.597,1.54213) (0.598,1.5468) (0.599,1.55151)
(0.6,1.55625) (0.601,1.56102) (0.602,1.56583) (0.603,1.57067) (0.604,1.57555)
(0.605,1.58046) (0.606,1.5854) (0.607,1.59038) (0.608,1.5954) (0.609,1.60046)
(0.61,1.60555) (0.611,1.61068) (0.612,1.61584) (0.613,1.62105) (0.614,1.62629)
(0.615,1.63157) (0.616,1.63689) (0.617,1.64225) (0.618,1.64765) (0.619,1.65309)
(0.62,1.65858) (0.621,1.6641) (0.622,1.66967) (0.623,1.67528) (0.624,1.68093)
(0.625,1.68663) (0.626,1.69237) (0.627,1.69815) (0.628,1.70398) (0.629,1.70986)
(0.63,1.71578) (0.631,1.72175) (0.632,1.72777) (0.633,1.73383) (0.634,1.73995)
(0.635,1.74611) (0.636,1.75232) (0.637,1.75859) (0.638,1.7649) (0.639,1.77127)
(0.64,1.77768) (0.641,1.78416) (0.642,1.79068) (0.643,1.79726) (0.644,1.80389)
(0.645,1.81058) (0.646,1.81733) (0.647,1.82413) (0.648,1.83099) (0.649,1.83791)
(0.65,1.84489) (0.651,1.85193) (0.652,1.85903) (0.653,1.8662) (0.654,1.87342)
(0.655,1.88071) (0.656,1.88806) (0.657,1.89548) (0.658,1.90296) (0.659,1.91051)
(0.66,1.91813) (0.661,1.92582) (0.662,1.93357) (0.663,1.9414) (0.664,1.9493)
(0.665,1.95727) (0.666,1.96531) (0.667,1.97343) (0.668,1.98162) (0.669,1.9899)
(0.67,1.99824) (0.671,2.00667) (0.672,2.01518) (0.673,2.02376) (0.674,2.03243)
(0.675,2.04119) (0.676,2.05003) (0.677,2.05895) (0.678,2.06796) (0.679,2.07706)
(0.68,2.08625) (0.681,2.09553) (0.682,2.1049) (0.683,2.11436) (0.684,2.12392)
(0.685,2.13358) (0.686,2.14334) (0.687,2.15319) (0.688,2.16314) (0.689,2.1732)
(0.69,2.18336) (0.691,2.19363) (0.692,2.204) (0.693,2.21449) (0.694,2.22508)
(0.695,2.23579) (0.696,2.24661) (0.697,2.25754) (0.698,2.2686) (0.699,2.27977)
(0.7,2.29106) (0.701,2.30248) (0.702,2.31403) (0.703,2.3257) (0.704,2.3375)
(0.705,2.34943) (0.706,2.3615) (0.707,2.3737) (0.708,2.38605) (0.709,2.39853)
(0.71,2.41115) (0.711,2.42393) (0.712,2.43685) (0.713,2.44992) (0.714,2.46314)
(0.715,2.47652) (0.716,2.49006) (0.717,2.50376) (0.718,2.51762) (0.719,2.53165)
(0.72,2.54585) (0.721,2.56023) (0.722,2.57478) (0.723,2.58951) (0.724,2.60442)
(0.725,2.61952) (0.726,2.6348) (0.727,2.65028) (0.728,2.66596) (0.729,2.68184)
(0.73,2.69792) (0.731,2.7142) (0.732,2.7307) (0.733,2.74742) (0.734,2.76435)
(0.735,2.78151) (0.736,2.7989) (0.737,2.81652) (0.738,2.83437) (0.739,2.85247)
(0.74,2.87082) (0.741,2.88941) (0.742,2.90827) (0.743,2.92738) (0.744,2.94676)
(0.745,2.96642) (0.746,2.98635) (0.747,3.00656) (0.748,3.02707) (0.749,3.04787)
(0.75,3.06897) (0.751,3.09038) (0.752,3.1121) (0.753,3.13414) (0.754,3.15652)
(0.755,3.17922) (0.756,3.20227) (0.757,3.22567) (0.758,3.24943) (0.759,3.27355)
(0.76,3.29805) (0.761,3.32292) (0.762,3.34819) (0.763,3.37386) (0.764,3.39993)
(0.765,3.42642) (0.766,3.45334) (0.767,3.48069) (0.768,3.50849) (0.769,3.53675)
(0.77,3.56548) (0.771,3.59469) (0.772,3.62438) (0.773,3.65458) (0.774,3.6853)
(0.775,3.71654) (0.776,3.74832) (0.777,3.78065) (0.778,3.81355) (0.779,3.84703)
(0.78,3.88111) (0.781,3.9158) (0.782,3.95111) (0.783,3.98707) (0.784,4.02369)
(0.785,4.06098) (0.786,4.09896) (0.787,4.13766) (0.788,4.17709) (0.789,4.21727)
(0.79,4.25822) (0.791,4.29996) (0.792,4.34251) (0.793,4.3859) (0.794,4.43015)
(0.795,4.47528) (0.796,4.52132) (0.797,4.56829) (0.798,4.61622) (0.799,4.66514)
(0.8,4.71507) (0.801,4.76605) (0.802,4.8181) (0.803,4.87127) (0.804,4.92557)
(0.805,4.98104) (0.806,5.03773) (0.807,5.09566) (0.808,5.15488) (0.809,5.21543)
(0.81,5.27734) (0.811,5.34066) (0.812,5.40544) (0.813,5.47171) (0.814,5.53954)
(0.815,5.60896) (0.816,5.68003) (0.817,5.75281) (0.818,5.82736) (0.819,5.90372)
(0.82,5.98197) (0.821,6.06217) (0.822,6.14438) (0.823,6.22867) (0.824,6.31512)
(0.825,6.40381) (0.826,6.49481) (0.827,6.5882) (0.828,6.68408) (0.829,6.78253)
(0.83,6.88364) (0.831,6.98753) (0.832,7.09428) (0.833,7.20401) (0.834,7.31684)
(0.835,7.43288) (0.836,7.55225) (0.837,7.6751) (0.838,7.80155) (0.839,7.93175)
(0.84,8.06585) (0.841,8.20401) (0.842,8.3464) (0.843,8.49319) (0.844,8.64457)
(0.845,8.80074) (0.846,8.96189) (0.847,9.12824) (0.848,9.30002) (0.849,9.47747)
(0.85,9.66085) (0.851,9.85041) (0.852,10.0464) (0.853,10.2493) (0.854,10.4591)
(0.855,10.6764) (0.856,10.9015) (0.857,11.1347) (0.858,11.3765) (0.859,11.6272)
(0.86,11.8872) (0.861,12.1572) (0.862,12.4375) (0.863,12.7288) (0.864,13.0315)
(0.865,13.3462) (0.866,13.6737) (0.867,14.0146) (0.868,14.3696) (0.869,14.7394)
(0.87,15.1251) (0.871,15.5273) (0.872,15.947) (0.873,16.3854) (0.874,16.8434)
(0.875,17.3223) (0.876,17.8232) (0.877,18.3475) (0.878,18.8967) (0.879,19.4723)
(0.88,20.0759) (0.881,20.7095) (0.882,21.3748) (0.883,22.0741) (0.884,22.8096)
(0.885,23.5836) (0.886,24.399) (0.887,25.2584) (0.888,26.1652) (0.889,27.1226)
(0.89,28.1343) (0.891,29.2043) (0.892,30.3371) (0.893,31.5373) (0.894,32.8102)
(0.895,34.1615) (0.896,35.5974) (0.897,37.1248) (0.898,38.7512) (0.899,40.4849)
(0.9,42.335) (0.901,44.3115) (0.902,46.4255) (0.903,48.6894) (0.904,51.1167)
(0.905,53.7227) (0.906,56.524) (0.907,59.5394) (0.908,62.7898) (0.909,66.2986)
(0.91,70.092) (0.911,74.1992) (0.912,78.6533) (0.913,83.4916) (0.914,88.7559)
(0.915,94.4939)
}
\begin{axis}[
 	scale=1.1,
	xmin=0, xmax=1,
	ymin=0, ymax=5.5,
	axis x line*= bottom,
    every outer y axis line/.append style={-stealth},
    ytick={0,0.611259,1,3,5,7,9},
	minor ytick={0.611259,1},
	yminorgrids=true,
    xtick={0,0.25,...,1},
	minor xtick={0.5},
    xmajorgrids=false,
	xminorgrids=true,
    grid style=dashed,
    xlabel = {$\lambda$},
    ylabel = {$\sigma/\mu$},
]
\addplot [black,name path=feasibility, very thick,smooth] coordinates{\pathfeasible};
\addplot [red,name path=regulardispersion, very thick,smooth] coordinates{\pathregulardispersion};
\addplot [yellow!80,fill opacity=0.4] fill between [of=feasibility and regulardispersion,reverse=true];
\path [name path=yaxis]
        (0,\pgfkeysvalueof{/pgfplots/ymin}) --
        (0,\pgfkeysvalueof{/pgfplots/ymax});
\addplot [pattern color=black, pattern=north east lines] fill between [of=feasibility and yaxis,reverse=true];
\path [name path=bottomrightaxis]
		(0,0) --
        (1,0) --
        (1,\pgfkeysvalueof{/pgfplots/ymax});
\addplot [blue!80,fill opacity=0.4] fill between [of=regulardispersion and bottomrightaxis,reverse=true];

\node[draw,fill=white,inner sep=3pt] at (0.25,4) {Infeasible};
\node at (0.60,3) {Price at mean};
\node at (0.82,1.4) {Log-price lottery};
\end{axis}
\end{tikzpicture}

%% file: ms.bbl
\begin{thebibliography}{67}
\providecommand{\natexlab}[1]{#1}
\providecommand{\url}[1]{\texttt{#1}}
\expandafter\ifx\csname urlstyle\endcsname\relax
  \providecommand{\doi}[1]{doi: #1}\else
  \providecommand{\doi}{doi: \begingroup \urlstyle{rm}\Url}\fi

\bibitem[Allouah et~al.(2022)Allouah, Bahamou, and Besbes]{Allouah2022}
A.~Allouah, A.~Bahamou, and O.~Besbes.
\newblock Pricing with samples.
\newblock \emph{Operations Research}, 2022.
\newblock \doi{10.1287/opre.2021.2200}.
\newblock Forthcoming.

\bibitem[Azar and Micali(2012)]{Azar:2012aa}
P.~Azar and S.~Micali.
\newblock Optimal parametric auctions.
\newblock Technical Report MIT-CSAIL-TR-2012-015, MIT CSAIL, Oct. 2012.
\newblock URL \url{http://hdl.handle.net/1721.1/70556}.

\bibitem[Azar et~al.(2013)Azar, Daskalakis, Micali, and Weinberg]{Azar2013}
P.~Azar, C.~Daskalakis, S.~Micali, and S.~M. Weinberg.
\newblock Optimal and efficient parametric auctions.
\newblock In \emph{Proceedings of the 24th Annual ACM-SIAM Symposium on
  Discrete Algorithms (SODA)}, pages 596--604, 2013.
\newblock \doi{10.1137/1.9781611973105.43}.

\bibitem[Azar and Micali(2013)]{Azar:2013aa}
P.~D. Azar and S.~Micali.
\newblock Parametric digital auctions.
\newblock In \emph{Proceedings of the 4th Conference on Innovations in
  Theoretical Computer Science (ITCS)}, pages 231--232, 2013.
\newblock \doi{10.1145/2422436.2422464}.

\bibitem[Babaioff et~al.(2014)Babaioff, Immorlica, Lucier, and
  Weinberg]{Babaioff2014a}
M.~Babaioff, N.~Immorlica, B.~Lucier, and S.~M. Weinberg.
\newblock A simple and approximately optimal mechanism for an additive buyer.
\newblock In \emph{Proceedings of the 55th Annual Symposium on Foundations of
  Computer Science (FOCS)}, pages 21--30, 2014.
\newblock \doi{10.1109/FOCS.2014.11}.

\bibitem[Bandi and Bertsimas(2014)]{BandiBertsimas2014}
C.~Bandi and D.~Bertsimas.
\newblock Optimal design for multi-item auctions: A robust optimization
  approach.
\newblock \emph{Mathematics of Operations Research}, 39\penalty0 (4):\penalty0
  1012--1038, 2014.
\newblock \doi{10.1287/moor.2014.0645}.

\bibitem[Barlow and Proschan(1996)]{Barlow1996}
R.~E. Barlow and F.~Proschan.
\newblock Failure distributions.
\newblock In \emph{Mathematical Theory of Reliability}, chapter~2. Society for
  Industrial and Applied Mathematics, 1996.
\newblock \doi{10.1137/1.9781611971194}.

\bibitem[Bei et~al.(2019)Bei, Gravin, Lu, and Tang]{Bei:2019aa}
X.~Bei, N.~Gravin, P.~Lu, and Z.~G. Tang.
\newblock Correlation-robust analysis of single item auction.
\newblock In \emph{Proceedings of the 30th Annual ACM-SIAM Symposium on
  Discrete Algorithms (SODA)}, pages 193--208, 2019.
\newblock \doi{10.1137/1.9781611975482.13}.

\bibitem[Bergemann and Schlag(2011)]{Bergemann:2011aa}
D.~Bergemann and K.~Schlag.
\newblock Robust monopoly pricing.
\newblock \emph{Journal of Economic Theory}, 146\penalty0 (6):\penalty0 2527 --
  2543, 2011.
\newblock \doi{10.1016/j.jet.2011.10.018}.

\bibitem[Bergemann and Schlag(2008)]{Bergemann2008}
D.~Bergemann and K.~H. Schlag.
\newblock Pricing without priors.
\newblock \emph{Journal of the European Economic Association}, 6\penalty0
  (2-3):\penalty0 560--569, 2008.
\newblock \doi{10.1162/jeea.2008.6.2-3.560}.

\bibitem[Borodin and El-Yaniv(1998)]{Borodin1998a}
A.~Borodin and R.~El-Yaniv.
\newblock \emph{Online Computation and Competitive Analysis}.
\newblock Cambridge University Press, 1998.

\bibitem[Boucheron et~al.(2013)Boucheron, Lugosi, and Massart]{BLM:2013}
S.~Boucheron, G.~Lugosi, and P.~Massart.
\newblock \emph{Concentration inequalities: A nonasymptotic theory of
  independence}.
\newblock Oxford university press, 2013.

\bibitem[Bulow and Klemperer(1996)]{Bulow1996}
J.~Bulow and P.~Klemperer.
\newblock Auctions versus negotiations.
\newblock \emph{The American Economic Review}, 86\penalty0 (1):\penalty0
  180--194, 1996.
\newblock URL \url{http://www.jstor.org/stable/2118262}.

\bibitem[Cai and Daskalakis(2017)]{Cai:2017ab}
Y.~Cai and C.~Daskalakis.
\newblock Learning multi-item auctions with (or without) samples.
\newblock In \emph{Proceedings of the 58th Annual Symposium on Foundations of
  Computer Science (FOCS)}, pages 516--527, Oct 2017.
\newblock \doi{10.1109/FOCS.2017.54}.

\bibitem[Cai and Zhao(2017)]{Cai:2017aa}
Y.~Cai and M.~Zhao.
\newblock Simple mechanisms for subadditive buyers via duality.
\newblock In \emph{Proceedings of the 49th Annual ACM SIGACT Symposium on
  Theory of Computing (STOC)}, pages 170--183, 2017.
\newblock \doi{10.1145/3055399.3055465}.

\bibitem[Cai et~al.(2016)Cai, Devanur, and Weinberg]{Cai2016}
Y.~Cai, N.~R. Devanur, and S.~M. Weinberg.
\newblock {A duality based unified approach to Bayesian mechanism design}.
\newblock In \emph{Proceedings of the 48th Annual ACM SIGACT Symposium on
  Theory of Computing (STOC)}, pages 926--939. ACM Press, 2016.
\newblock \doi{10.1145/2897518.2897645}.

\bibitem[Carrasco et~al.(2018)Carrasco, Luz, Kos, Messner, Monteiro, and
  Moreira]{Carrasco2018}
V.~Carrasco, V.~F. Luz, N.~Kos, M.~Messner, P.~Monteiro, and H.~Moreira.
\newblock Optimal selling mechanisms under moment conditions.
\newblock \emph{J. Econ. Theory}, 177:\penalty0 245--279, 2018.
\newblock \doi{10.1016/j.jet.2018.05.005}.

\bibitem[Carroll(2017)]{Carroll:2017aa}
G.~Carroll.
\newblock Robustness and separation in multidimensional screening.
\newblock \emph{Econometrica}, 85\penalty0 (2):\penalty0 453--488, 2017.
\newblock \doi{10.3982/ECTA14165}.

\bibitem[Carroll(2019)]{Carroll2019}
G.~Carroll.
\newblock Robustness in mechanism design and contracting.
\newblock \emph{Annual Review of Economics}, 11\penalty0 (1):\penalty0
  139--166, 2019.
\newblock \doi{10.1146/annurev-economics-080218-025616}.

\bibitem[Chawla et~al.(2010)Chawla, Hartline, Malec, and Sivan]{Chawla2010a}
S.~Chawla, J.~D. Hartline, D.~L. Malec, and B.~Sivan.
\newblock Multi-parameter mechanism design and sequential posted pricing.
\newblock In \emph{Proceedings of the 42nd ACM symposium on Theory of Computing
  (STOC)}, pages 311--320, 2010.
\newblock \doi{10.1145/1806689.1806733}.

\bibitem[Chawla et~al.(2014)Chawla, Fu, and Karlin]{Chawla2014a}
S.~Chawla, H.~Fu, and A.~R. Karlin.
\newblock Approximate revenue maximization in interdependent value settings.
\newblock \emph{CoRR}, abs/1408.4424, 2014.
\newblock URL \url{https://arxiv.org/abs/1408.4424}.

\bibitem[Chen et~al.(2018{\natexlab{a}})Chen, Li, Li, and Lu]{Chen:2018ab}
J.~Chen, B.~Li, Y.~Li, and P.~Lu.
\newblock Bayesian auctions with efficient queries.
\newblock \emph{CoRR}, abs/1804.07451, 2018{\natexlab{a}}.
\newblock URL \url{https://arxiv.org/abs/1804.07451}.

\bibitem[{Chen} et~al.(2015){Chen}, {Diakonikolas}, {Orfanou}, {Paparas},
  {Sun}, and {Yannakakis}]{Chen:2015aa}
X.~{Chen}, I.~{Diakonikolas}, A.~{Orfanou}, D.~{Paparas}, X.~{Sun}, and
  M.~{Yannakakis}.
\newblock On the complexity of optimal lottery pricing and randomized
  mechanisms.
\newblock In \emph{Proceedings of 56th Annual Symposium on Foundations of
  Computer Science (FOCS)}, pages 1464--1479, Oct. 2015.
\newblock \doi{10.1109/FOCS.2015.93}.

\bibitem[Chen et~al.(2018{\natexlab{b}})Chen, Matikas, Paparas, and
  Yannakakis]{Chen:2018aa}
X.~Chen, G.~Matikas, D.~Paparas, and M.~Yannakakis.
\newblock On the complexity of simple and optimal deterministic mechanisms for
  an additive buyer.
\newblock In \emph{Proceedings of the 29th Annual ACM-SIAM Symposium on
  Discrete Algorithms (SODA)}, pages 2036--2049, 2018{\natexlab{b}}.
\newblock \doi{10.1137/1.9781611975031.133}.

\bibitem[Cole and Rao(2017)]{Cole2017}
R.~Cole and S.~Rao.
\newblock Applications of $\alpha$-strongly regular distributions to bayesian
  auctions.
\newblock \emph{ACM Trans. Econ. Comput.}, 5\penalty0 (4):\penalty0
  18:1--18:29, Dec. 2017.
\newblock \doi{10.1145/3157083}.

\bibitem[Cole and Roughgarden(2014)]{Cole2014a}
R.~Cole and T.~Roughgarden.
\newblock The sample complexity of revenue maximization.
\newblock In \emph{Proceedings of the 46th Annual ACM Symposium on Theory of
  Computing (STOC)}, pages 243--252, 2014.
\newblock \doi{10.1145/2591796.2591867}.

\bibitem[Daskalakis and Zampetakis(2020)]{Zampetakis2020}
C.~Daskalakis and M.~Zampetakis.
\newblock More revenue from two samples via factor revealing sdps.
\newblock In \emph{Proceedings of the 21st {ACM} Conference on Economics and
  Computation}, page 257–272, 2020.
\newblock \doi{10.1145/3391403.3399543}.

\bibitem[Daskalakis et~al.(2013)Daskalakis, Deckelbaum, and
  Tzamos]{Daskalakis2013}
C.~Daskalakis, A.~Deckelbaum, and C.~Tzamos.
\newblock The complexity of optimal mechanism design.
\newblock In \emph{Proceedings of the 25th Annual {ACM}-{SIAM} Symposium on
  Discrete Algorithms (SODA)}, pages 1302--1318, 2013.
\newblock \doi{10.1137/1.9781611973402.96}.

\bibitem[Daskalakis et~al.(2017)Daskalakis, Deckelbaum, and
  Tzamos]{Daskalakis:2017aa}
C.~Daskalakis, A.~Deckelbaum, and C.~Tzamos.
\newblock Strong duality for a multiple-good monopolist.
\newblock \emph{Econometrica}, 85\penalty0 (3):\penalty0 735--767, 2017.
\newblock \doi{10.3982/ECTA12618}.

\bibitem[Dhangwatnotai et~al.(2014)Dhangwatnotai, Roughgarden, and
  Yan]{Dhangwatnotai2014a}
P.~Dhangwatnotai, T.~Roughgarden, and Q.~Yan.
\newblock Revenue maximization with a single sample.
\newblock \emph{Games and Economic Behavior}, 91\penalty0 (C):\penalty0
  318--333, 2014.

\bibitem[Dughmi et~al.(2014)Dughmi, Han, and Nisan]{Dughmi2014}
S.~Dughmi, L.~Han, and N.~Nisan.
\newblock Sampling and representation complexity of revenue maximization.
\newblock In \emph{Proceedings of the 10th International Conference on Web and
  Internet Economics (WINE)}, pages 277--291, 2014.
\newblock \doi{10.1007/978-3-319-13129-0_22}.

\bibitem[D\"{u}tting et~al.(2019)D\"{u}tting, Roughgarden, and
  Talgam-Cohen]{Dutting2019}
P.~D\"{u}tting, T.~Roughgarden, and I.~Talgam-Cohen.
\newblock Simple versus optimal contracts.
\newblock In \emph{Proceedings of the 20th ACM Conference on Economics and
  Computation (EC)}, pages 369--387, 2019.
\newblock \doi{10.1145/3328526.3329591}.

\bibitem[Fu et~al.(2015)Fu, Immorlica, Lucier, and Strack]{Fu2015}
H.~Fu, N.~Immorlica, B.~Lucier, and P.~Strack.
\newblock Randomization beats second price as a prior-independent auction.
\newblock In \emph{Proceedings of the 16th {ACM} Conference on Economics and
  Computation (EC)}, pages 323--323, 2015.
\newblock \doi{10.1145/2764468.2764489}.

\bibitem[Giannakopoulos and Koutsoupias(2018)]{gk2014_sicomp}
Y.~Giannakopoulos and E.~Koutsoupias.
\newblock Duality and optimality of auctions for uniform distributions.
\newblock \emph{{SIAM} Journal on Computing}, 47\penalty0 (1):\penalty0
  121--165, 2018.
\newblock \doi{10.1137/16M1072218}.

\bibitem[Giannakopoulos et~al.(2020{\natexlab{a}})Giannakopoulos, Poças, and
  Tsigonias-Dimitriadis]{Giannakopoulos2019}
Y.~Giannakopoulos, D.~Poças, and A.~Tsigonias-Dimitriadis.
\newblock Robust revenue maximization under minimal statistical information.
\newblock In \emph{Proceedings of the 16th Conference on Web and Internet
  Economics (WINE)}, pages 177--190, 2020{\natexlab{a}}.
\newblock \doi{10.1007/978-3-030-64946-3_13}.

\bibitem[Giannakopoulos et~al.(2020{\natexlab{b}})Giannakopoulos, Poças, and
  Zhu]{gpz19_arxiv}
Y.~Giannakopoulos, D.~Poças, and K.~Zhu.
\newblock Optimal pricing for {MHR} and $\lambda$-regular distributions.
\newblock \emph{ACM Trans. Econ. Comput.}, 9\penalty0 (1):\penalty0 2:1--2:28,
  Dec. 2020{\natexlab{b}}.
\newblock \doi{10.1145/3434423}.

\bibitem[Gonczarowski and Weinberg(2021)]{Gonczarowski:2018aa}
Y.~A. Gonczarowski and S.~M. Weinberg.
\newblock The sample complexity of up-to-$\varepsilon$ multi-dimensional
  revenue maximization.
\newblock \emph{Journal of the {ACM}}, 68\penalty0 (3):\penalty0 1--28, 2021.
\newblock \doi{10.1145/3439722}.

\bibitem[Gravin and Lu(2018)]{Gravin:2018aa}
N.~Gravin and P.~Lu.
\newblock Separation in correlation-robust monopolist problem with budget.
\newblock In \emph{Proceedings of the 29th Annual ACM-SIAM Symposium on
  Discrete Algorithms (SODA)}, pages 2069--2080, 2018.
\newblock \doi{10.1137/1.9781611975031.135}.

\bibitem[Haghpanah and Hartline(2015)]{Haghpanah:2015aa}
N.~Haghpanah and J.~Hartline.
\newblock Reverse mechanism design.
\newblock In \emph{Proceedings of the 16th ACM Conference on Economics and
  Computation (EC)}, pages 757--758, 2015.
\newblock \doi{10.1145/2764468.2764498}.

\bibitem[Hart and Nisan(2013)]{Hart2013a}
S.~Hart and N.~Nisan.
\newblock The menu-size complexity of auctions.
\newblock In \emph{Proceedings of the 14th ACM Conference on Electronic
  Commerce (EC)}, pages 565--566, 2013.
\newblock \doi{10.1145/2482540.2482544}.

\bibitem[Hart and Nisan(2017)]{Hart:2017aa}
S.~Hart and N.~Nisan.
\newblock Approximate revenue maximization with multiple items.
\newblock \emph{J. Econ. Theory}, 172:\penalty0 313--347, 2017.
\newblock \doi{10.1016/j.jet.2017.09.001}.

\bibitem[Hartline(2013)]{Hartlinea}
J.~D. Hartline.
\newblock Mechanism design and approximation.
\newblock Manuscript, 2013.
\newblock URL \url{http://jasonhartline.com/MDnA/}.

\bibitem[Hartline and Roughgarden(2009)]{Hartline2009a}
J.~D. Hartline and T.~Roughgarden.
\newblock Simple versus optimal mechanisms.
\newblock In \emph{Proceedings of the 10th ACM Conference on Electronic
  Commerce (EC)}, pages 225--234, 2009.
\newblock \doi{10.1145/1566374.1566407}.

\bibitem[Hu et~al.(2021)Hu, Huang, Shen, and Wang]{Hu2021}
Y.~Hu, Z.~Huang, Y.~Shen, and X.~Wang.
\newblock Targeting makes sample efficiency in auction design.
\newblock In \emph{Proceedings of the 22nd {ACM} Conference on Economics and
  Computation}, page 610–629, 2021.
\newblock \doi{10.1145/3465456.3467631}.

\bibitem[Huang et~al.(2018)Huang, Mansour, and Roughgarden]{Huang:2018aa}
Z.~Huang, Y.~Mansour, and T.~Roughgarden.
\newblock Making the most of your samples.
\newblock \emph{SIAM Journal on Computing}, 47\penalty0 (3):\penalty0 651--674,
  2018.
\newblock \doi{10.1137/16M1065719}.

\bibitem[Kendall(1948)]{Kendall:1948aa}
M.~G. Kendall.
\newblock \emph{The Advanced Theory of Statistics}, volume~I.
\newblock Charles Griffin, 4th edition, 1948.

\bibitem[Ko{\c{c}}yi{\u{g}}it et~al.(2020)Ko{\c{c}}yi{\u{g}}it, Iyengar, Kuhn,
  and Wiesemann]{Ko_yi_it_2020}
{\c{C}}.~Ko{\c{c}}yi{\u{g}}it, G.~Iyengar, D.~Kuhn, and W.~Wiesemann.
\newblock Distributionally robust mechanism design.
\newblock \emph{Management Science}, 66\penalty0 (1):\penalty0 159--189, 2020.
\newblock \doi{10.1287/mnsc.2018.3219}.

\bibitem[Li and Yao(2013)]{Li2013a}
X.~Li and A.~C.-C. Yao.
\newblock On revenue maximization for selling multiple independently
  distributed items.
\newblock \emph{Proceedings of the National Academy of Sciences}, 110\penalty0
  (28):\penalty0 11232--11237, 2013.
\newblock \doi{10.1073/pnas.1309533110}.

\bibitem[Li et~al.(2019)Li, Lu, and Ye]{Li:2019aa}
Y.~Li, P.~Lu, and H.~Ye.
\newblock Revenue maximization with imprecise distribution.
\newblock In \emph{Proceedings of the 18th International Conference on
  Autonomous Agents and MultiAgent Systems (AAMAS)}, pages 1582--1590, 2019.
\newblock URL \url{http://dl.acm.org/citation.cfm?id=3306127.3331877}.

\bibitem[Mallows and Richter(1969)]{Mallows:1969aa}
C.~L. Mallows and D.~Richter.
\newblock Inequalities of chebyshev type involving conditional expectations.
\newblock \emph{The Annals of Mathematical Statistics}, 40\penalty0
  (6):\penalty0 1922--1932, 1969.
\newblock URL \url{http://www.jstor.org/stable/2239511}.

\bibitem[Manelli and Vincent(2007)]{Manelli2007a}
A.~M. Manelli and D.~R. Vincent.
\newblock Multidimensional mechanism design: Revenue maximization and the
  multiple-good monopoly.
\newblock \emph{Journal of Economic Theory}, 137\penalty0 (1):\penalty0
  153--185, 2007.
\newblock \doi{10.1016/j.jet.2006.12.007}.

\bibitem[Milgrom(2004)]{Milgrom:2004aa}
P.~Milgrom.
\newblock \emph{Putting Auction Theory to Work}.
\newblock Cambridge University Press, 2004.
\newblock \doi{10.1017/CBO9780511813825.009}.

\bibitem[Mood et~al.(1974)Mood, Graybill, and Boes]{MGB:1974}
A.~M. Mood, F.~A. Graybill, and D.~C. Boes.
\newblock \emph{Introduction to the Theory of Statistics}.
\newblock McGraw-Hill, 3rd edition, 1974.

\bibitem[Morgenstern and Roughgarden(2016)]{Morgenstern:2016aa}
J.~Morgenstern and T.~Roughgarden.
\newblock Learning simple auctions.
\newblock In \emph{Proceedings of the 29th Annual Conference on Learning Theory
  (COLT)}, pages 1298--1318, 2016.
\newblock URL \url{http://proceedings.mlr.press/v49/morgenstern16.html}.

\bibitem[Moriguti(1951)]{Moriguti1951a}
S.~Moriguti.
\newblock Extremal properties of extreme value distributions.
\newblock \emph{The Annals of Mathematical Statistics}, 22\penalty0
  (4):\penalty0 523--536, 1951.
\newblock URL \url{http://www.jstor.org/stable/2236921}.

\bibitem[Motwani and Raghavan(1995)]{Motwani1995a}
R.~Motwani and P.~Raghavan.
\newblock \emph{Randomized Algorithms}.
\newblock Cambridge University Press, 1995.

\bibitem[Myerson(1981)]{Myerson1981}
R.~B. Myerson.
\newblock Optimal auction design.
\newblock \emph{Mathematics of Operations Research}, 6\penalty0 (1):\penalty0
  58--73, 1981.
\newblock \doi{10.1287/moor.6.1.58}.

\bibitem[Nisan(2007)]{Nisan07}
N.~Nisan.
\newblock Introduction to mechanism design (for computer scientists).
\newblock In N.~Nisan, T.~Roughgarden, {\'{E}}.~Tardos, and V.~Vazirani,
  editors, \emph{Algorithmic Game Theory}, chapter~9. Cambridge University
  Press, 2007.

\bibitem[Raiffa(1961)]{Raiffa1961}
H.~Raiffa.
\newblock Risk, ambiguity, and the savage axioms: Comment.
\newblock \emph{The Quarterly Journal of Economics}, 75\penalty0 (4):\penalty0
  690--694, 1961.
\newblock \doi{10.2307/1884326}.

\bibitem[Roughgarden and Talgam-Cohen(2019)]{Roughgarden:2019aa}
T.~Roughgarden and I.~Talgam-Cohen.
\newblock Approximately optimal mechanism design.
\newblock \emph{Annual Review of Economics}, 11\penalty0 (1), 2019.
\newblock \doi{10.1146/annurev-economics-080218-025607}.

\bibitem[Rubinstein and Weinberg(2018)]{Rubinstein:2018aa}
A.~Rubinstein and S.~M. Weinberg.
\newblock Simple mechanisms for a subadditive buyer and applications to revenue
  monotonicity.
\newblock \emph{ACM Trans. Econ. Comput.}, 6\penalty0 (3-4):\penalty0
  19:1--19:25, Oct. 2018.
\newblock \doi{10.1145/3105448}.

\bibitem[Schweizer and Szech(2019)]{SchweizerSzech2019}
N.~Schweizer and N.~Szech.
\newblock Performance bounds for optimal sales mechanisms beyond the monotone
  hazard rate condition.
\newblock \emph{Journal of Mathematical Economics}, 82:\penalty0 202--213,
  2019.

\bibitem[Suzdaltsev(2020)]{Suzdaltsev2020}
A.~Suzdaltsev.
\newblock An optimal distributionally robust auction.
\newblock \emph{CoRR}, abs/2006.05192, 2020.
\newblock URL \url{https://arxiv.org/abs/2006.05192}.

\bibitem[Syrgkanis(2017)]{Syrgkanis:2017aa}
V.~Syrgkanis.
\newblock A sample complexity measure with applications to learning optimal
  auctions.
\newblock In \emph{Proccedings of the 31st Conference on Neural Information
  Processing Systems (NIPS)}, pages 5352--5359, 2017.
\newblock URL
  \url{http://papers.nips.cc/paper/7119-a-sample-complexity-measure-with-applications-to-learning-optimal-auctions.pdf}.

\bibitem[Tao(2011)]{tao2011measure}
T.~Tao.
\newblock \emph{An Introduction to Measure Theory}, volume 126 of
  \emph{Graduate Studies in Mathematics}.
\newblock American Mathematical Society, 2011.

\bibitem[Wilson(1987)]{Wilson:1987aa}
R.~Wilson.
\newblock Game-theoretic analyses of trading processes.
\newblock In \emph{Advances in Economic Theory: Fifth World Congress},
  Econometric Society Monographs, pages 33--70. Cambridge University Press,
  1987.
\newblock \doi{10.1017/CCOL0521340446.002}.

\bibitem[Yao(2015)]{Yao2015}
A.~C.-C. Yao.
\newblock An $n$-to-$1$ bidder reduction for multi-item auctions and its
  applications.
\newblock In \emph{Proceedings of the 26th Annual ACM-SIAM Symposium on
  Discrete Algorithms (SODA)}, pages 92--109, 2015.
\newblock \doi{10.1137/1.9781611973730.8}.

\end{thebibliography}
